\documentclass[runningheads]{llncs}

\usepackage[appendix=append,bibliography=common]{apxproof}

\renewcommand{\appendixsectionformat}[2]{Proofs and Additional Material for \S#1\ (#2)}

\usepackage{amsmath}

\usepackage{hyperref}

\usepackage{color}
\renewcommand\UrlFont{\color{blue}\rmfamily}
\usepackage[T1]{fontenc}

\usepackage[utf8]{inputenc}
\usepackage[linesnumbered, ruled]{algorithm2e} %
\usepackage{subcaption}
\usepackage{graphicx}
\usepackage{orcidlink}
\usepackage{amssymb}
\usepackage{cleveref}

\usepackage{tikz}
\usetikzlibrary{calc,matrix,arrows,shapes,automata,backgrounds,petri,decorations.pathreplacing,babel}

\usetikzlibrary{arrows}
\usetikzlibrary{positioning}
\usetikzlibrary{automata}
\usetikzlibrary{shapes.geometric}
\tikzstyle{anode}=[rectangle, draw]
\tikzstyle{enode}=[diamond, shape aspect=2, draw]
\tikzset{
  invisible/.style={opacity=0},
  alt/.code args={<#1>#2#3}{%
    \alt<#1>{\pgfkeysalso{#2}}{\pgfkeysalso{#3}} 
  },
  visible on/.style={alt={#1{}{invisible}}},
}
\usetikzlibrary{calc,shapes.callouts,shapes.arrows}
\usepackage{tikz-cd}
\usetikzlibrary{decorations.pathmorphing}

\usepackage{pgfplots}
\usepackage{mathtools}
\usepackage{nicefrac}
\usepackage{ifthen}

\allowdisplaybreaks

\newtheoremrep{theorem}{Theorem}[section]
\newtheoremrep{maintheorem}[theorem]{Main Theorem}
\newtheoremrep{corollary}[theorem]{Corollary}
\newtheoremrep{lemma}[theorem]{Lemma}
\newtheoremrep{proposition}[theorem]{Proposition}
\newtheoremrep{definition}[theorem]{Definition}

\newtheorem{assumption}{Assumption}

\makeatletter
\def\@axp@newtheorem#1#2#3#4{%
\ifx\relax#4\relax
\ifx\relax#2\relax
\spnewtheorem{#1}{#3}{\itshape}{\normalfont}%
\else
\spnewtheorem{#1}[#2]{#3}{\itshape}{\normalfont}%
\fi
\else
\spnewtheorem{#1}{#3}[#4]{\itshape}{\normalfont}%
\fi
}
\makeatother

\newtheoremrep{remark}[theorem]{Remark}
\newtheoremrep{example}[theorem]{Example}

\crefname{assumption}{assumption}{assumptions}
\Crefname{assumption}{Assumption}{Assumptions}
\crefname{maintheorem}{theorem}{theorems}
\Crefname{maintheorem}{Theorem}{Theorems}

\newcommand{\R}{\mathbb{R}} %
\newcommand{\N}{\mathbb{N}}
\newcommand{\Z}{\mathbb{Z}}
\newcommand{\Q}{\mathbb{Q}}
\newcommand{\C}{\mathbb{C}}
\newcommand{\Rp}{\R_{\geq0}}
\newcommand{\fix}[1]{\mathrm{Fix}(#1)}
\newcommand{\id}{\mathrm{id}}
\newcommand{\E}{\mathbb{E}}
\renewcommand{\P}{\mathbb{P}}

\newcommand{\like}{witness-bounded}
\newcommand{\witness}{witness}

\DeclareMathOperator{\pow}{\mathcal{P}} %

\newcommand{\app}{\mathrm{a}}
\renewcommand{\epsilon}{\varepsilon}

\newcommand{\p}[2]{\ensuremath{({#1})^{#2}}}

\newcommand{\minf}[1]{\textstyle{\min_{#1}}}
\newcommand{\maxf}[1]{\textstyle{\max_{#1}}}

\newcommand{\rank}[1][]{\ensuremath{\mathit{r}\ifthenelse { \equal {#1} {} }{}{({#1})}}}

\newcommand{\constf}[1]{\ensuremath{\overline{#1}}}

\newcommand{\merge}{r}

\newcommand{\Tr}[4][]{\ensuremath{T{#1}({#2},{#3},{#4})}}
\newcommand{\MTr}[3][]{\ensuremath{T{#1}({#2},{#3})}}

\newcommand{\opts}[1]{\ensuremath{v^*_{#1}}}
\newcommand{\optsa}[1]{\ensuremath{q^*_{#1}}}

\newcommand{\reindex}[1]{\ensuremath{{#1}^\bullet}}

\newcommand\restr[2]{{
  \left.\kern-\nulldelimiterspace
  #1
  \vphantom{\big|}
  \right|_{#2}
  }}

\definecolor{dmagenta}{rgb}{0.81,0,0.81}
\definecolor{dcyan}{rgb}{0,0.6,0.6}
\definecolor{dgreen}{rgb}{0,0.6,0}

\newcommand{\dgreen}{\color{dgreen}}
\newcommand{\blue}{\color{blue}}
\newcommand{\red}{\color{red}}
\newcommand{\dmagenta}{\color{dmagenta}}
\newcommand{\dcyan}{\color{dcyan}}

\usepackage{csquotes}

\newcommand{\cut}[1]{{\color{dgreen}#1}}
\newcommand{\uncut}[1]{{\color{blue}#1}} 

\begin{document}

\title{Approximating Fixpoints\\
  of Approximated Functions}

\author{Paolo Baldan\inst{1}\orcidlink{0000-0001-9357-5599}
  \and
  Sebastian Gurke\inst{2}\orcidlink{0009-0008-4343-1384}
  \and
  Barbara K\"onig\inst{2}\orcidlink{0000-0002-4193-2889}
  \and
  Tommaso Padoan\inst{3}\orcidlink{0000-0001-7814-1485}
  \and
  Florian Wittbold\inst{2}\orcidlink{0000-0001-8307-503X}
\institute{Dipartimento di Matematica ``Tullio Levi-Civita'', \\Universit\`a di Padova\\\email{baldan@math.unipd.it}
  \and
  Universit\"at Duisburg-Essen\\
  \email{\{sebastian.gurke,barbara\_koenig,florian.wittbold\}@uni-due.de}
  \and
  Universit\`a di Trieste\\
  \email{tommaso.padoan@units.it}}
}
\authorrunning{P. Baldan, S. Gurke, B. K\"onig, T. Padoan, and F. Wittbold}

\maketitle              %

\begin{abstract}
  Fixpoints are ubiquitous in computer science and when dealing with
  quantitative semantics and verification one often considers least
  fixpoints of (higher-dimensional) functions over the non-negative
  reals.
  We show how to approximate the least fixpoint of such functions,
  focusing on the case in which they are not known precisely, but
  represented by a sequence of approximating functions that converge
  to them.
  We concentrate on monotone and non-expansive functions, for which
  uniqueness of fixpoints is not guaranteed and standard fixpoint
  iteration schemes might get stuck at a fixpoint that is not the
  least.
  Our main contribution is the identification of an iteration scheme,
  a variation of Mann iteration with a dampening factor, which, under
  suitable conditions, is shown to guarantee convergence to the least
  fixpoint of the function of interest.
  We then argue that these results are relevant in the context of
  model-based reinforcement learning for Markov decision processes,
  showing how the proposed iteration scheme instantiates and allows us
  to derive convergence to the optimal expected return.
  More generally, we show that our results can be used to iterate to
  the least fixpoint almost surely for systems where the function of
  interest can be approximated with given probabilistic error bounds,
  as it happens for probabilistic systems, such as simple stochastic games,
  which can be explored via sampling.

\end{abstract}

\section{Introduction}
\label{se:intro}

Fixpoints are fundamental in computer science as they
arise as a natural way of providing a meaning to inductive and
recursive definitions. When dealing with systems or programming
languages embodying quantitative aspects, such as probability, time, or
cost, we are often led to consider fixpoints of functions
associating states with real values.
Think, e.g., of the semantics of
iterative constructs of a probabilistic language, behavioural metrics
for probabilistic systems and model checking of quantitative logics.
Fixpoints also often
occur in optimization methods where one 
typically determines solutions
to a problem
that are
stable under suitable transformations. For instance, in
probabilistic systems such as finite
Markov decision processes (MDPs)~\cite{Bell:MDP}, the
expected payoff or the likelihood of satisfying a property can be
characterized as a (least) fixpoint, from which one can then derive
optimal policies.

The aim of this paper is to develop a theory of fixpoints for
functions that are not completely known but can be approximated.
Estimating fixpoints of unknown functions is a common task in computer
science for several reasons. For instance, model checking quantitative logics on probabilistic systems or verifying probabilistic programs with nested loops typically leads to the computation of nested fixpoints, which introduces subtle difficulties and requires a solid theory of approximation: since fixpoints might be only approximated, in the computation of an outer fixpoint, one has to resort to an approximation of an inner fixpoint. 

Moreover, in the aforementioned optimization setting, a standard framework is the one in which the parameters of a dynamic system of interest (e.g. probabilities, rewards, costs) are partially unknown  and can only be estimated by experimenting with the system. This is exactly what happens in reinforcement learning,
a branch of machine learning that is intended to provide methods for
learning optimal policies for an agent in order to maximize rewards
in an unknown dynamic environment.
Correctness guarantees for such
learning methods are necessary for verifying programs operating in uncertain environments (e.g., controllers in hybrid systems).

We exemplify our motivations by
focusing on
MDPs and (a variant of) 
the Dyna-Q
algorithm~\cite{s:dyna-integrated-architecture,s:planning-incremental-dyn-prog,kaelbling1996reinforcement}.
Consider the MDP $M$ in \Cref{fig:mdp-intro} with state set
$S = \{s_1,s_2,s_3,s_4\}$.  In each state $s\in S$, the agent can
choose an action $a$ from an action set $A = \{b,c\}$.  Performing an action $a\in A$ in state $s$ yields a successor state $s'$ with
a probability $\Tr[]{s}{a}{s'}$ and results in reward $R(s,a,s')$ (in
the picture, arrows are labelled ${\blue \Tr[]{s}{a}{s'}}\mid {\dgreen
  R(s,a,s')}$; if a
reward is $0$, the vertical line and the number 0 are omitted). Rewards can be discounted by a factor $\gamma < 1$ so
that rewards obtained in the future are valued less than immediate
rewards.

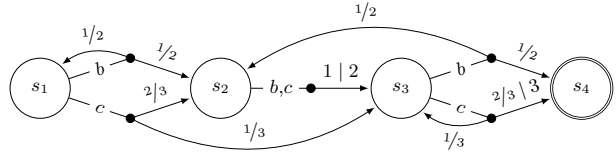
\begin{figure}[t]
  \centering
\clearpage{}%
\scalebox{0.8}{
\begin{tikzpicture}[every node/.style={draw}]
  \tikzstyle{action}=[>=latex, nodes={font=\footnotesize,sloped,draw=none},fill,draw,black]
  \path[shape=circle, minimum size=1cm]
    (0,0) node(s1)[]{$s_1$}
    (3,0) node(s2)[]{$s_2$}
    (6,0) node(s3)[]{$s_3$}
    (9,0) node(s4)[accepting]{$s_4$};

  \path[action] (1.5,0.5) circle (2pt)
    edge node[fill=white]{$b$}  (s1)
    edge [post] node[above]{$\blue\nicefrac{1}{2}$} (s2)
    edge [post, bend right] node[above]{$\blue\nicefrac{1}{2}$} (s1);
  \path[action] (1.5,-0.5) circle (2pt)
    edge node[fill=white]{$c$} (s1)
    edge [post] node[above]{$\blue \nicefrac{2}{3}$} (s2)
    edge [bend right, post] node[above]{$\blue \nicefrac{1}{3}$} (s3);
  \path[action] (4.5,0) circle (2pt)
    edge node[fill=white]{$b$,$c$} (s2)
    edge [post] node[above]{${\blue 1} \mid {\dgreen 2}$} (s3);
  \path[action] (7.5,0.5) circle (2pt)
    edge node[fill=white]{$b$} (s3)
    edge [post] node[above]{$\blue\nicefrac{1}{2}$} (s4)
    edge [post, bend right] node[above]{$\blue\nicefrac{1}{2}$} (s2);
  \path[action] (7.5,-0.5) circle (2pt)
    edge node[fill=white]{$c$} (s3)
    edge [post] node[above]{${\blue \nicefrac{2}{3}} \mid {\dgreen 3}$} (s4)
    edge [bend left, post] node[below]{$\blue \nicefrac{1}{3}$} (s3);
\end{tikzpicture}
}

\clearpage{}%

  \caption{A Markov decision process}
  \label{fig:mdp-intro}
\end{figure}

A typical aim is to determine a policy for the agent which maximises
the expected return.
The expected return (also called payoff or total reward) $q(s,a)$ at state $s$ taking action $a$ can be naturally expressed by means of a
recursive equation, the so-called Bellman equation
\begin{equation}
  \label{eq:bellman-ex}
  q(s,a) = \sum_{s'\in S}\Tr[]{s}{a}{s'}(R(s,a,s')+\gamma\cdot \max_{a'\in A}q(s',a'))
\end{equation}
stating that $q(s,a)$ coincides with the (discounted) expected returns
based on the successor states.  If we denote by
$g_M\colon \R^{S \times A} \to \R^{S \times A}$ the operator
associated with the MDP, defined by
$g_M(q)(s,a) = \sum_{s'\in
  S}\Tr[]{s}{a}{s'}(R(s,a,s')+\gamma\cdot\max_{a'\in A}q(s',a'))$, the
equation reduces to $q = g_M(q)$ and the maximal expected reward is
the least fixpoint of $g_M$ (unique when $\gamma < 1$ and thus $g_M$
contractive).

Various techniques have been devised for determining such
fixpoints. However, in reinforcement learning, the MDP describes the
interaction with an environment that is unknown or only partially
known. In particular, the probabilities $\Tr[]{s}{a}{s'}$ 
of arriving in $s'$ after choosing action $a$ in state $s$
can only be
estimated, e.g., by sampling.
In fact, a natural approach consists in
interacting with the model and record how
many times one arrives in each state $s'$, after choosing action $a$
in $s$.  If we denote by $N(s,a,s')$ this number, then $\Tr[]{s}{a}{s'}$
can be estimated as $N(s,a,s') / \sum_{s''} N(s,a,s'')$.  Clearly, as
we proceed and the number of interactions increases, we can expect to
obtain better and better approximations of the MDP.
Various reinforcement learning algorithms, called
\emph{model-based}, compute (optimal)
policies by estimating the probabilities in such a way.
For instance,
the Dyna-Q algorithm starts from arbitrary values. At each step,
the model is sampled and updated, and then also the values $q(s,a)$ are updated by taking
a weighted sum of the previously computed
value and the ($\gamma$-discounted) expectation of the
return according to 
equation~\eqref{eq:bellman-ex}
(either for
selected pairs or for all pairs, the variant that we consider here)
according to the following schema:
\begin{equation}
  \small
  \label{eq:dyna}
  q(s,a) := (1-\alpha)\cdot q(s,a)+\alpha\cdot \sum_{s'\in S}\Tr[]{s}{a}{s'}(R(s,a,s')+\gamma\cdot\max_{a'\in A}q(s',a'))
\end{equation}
which, using
function $g_M$,
becomes
$q := (1-\alpha) \cdot q + \alpha \cdot g_M(q)$. The weighted sum makes
the update more conservative, avoiding large oscillations.  The
parameter $\alpha$ can be chosen in order to fine-tune this behaviour:
a value of $\alpha$ closer to $0$ gives more relevance to the past
knowledge, while a value of $\alpha$ closer to
$1$ puts more emphasis on what has been learned up to the last step.
As the number of samplings increases, better approximations of the MDP are obtained, leading to better $q$-values.
At each step, the (current) best policy is obtained by choosing the
action with the highest $q$-value.

Several other model-based
reinforcement learning
algorithms use
variations of this scheme, where the updates (determining the value
vector and the optimal policy) and the exploration of the model are
interleaved \cite{ma:prioritized-sweeping,kaelbling1996reinforcement}.
Model-free
versions, such as Q-learning~\cite{WD:QL} or SARSA~\cite{r:sarsa},
do not explicitly build a model of the MDP.

The schema~\eqref{eq:dyna} above resembles iteration
algorithms for determining the fixpoint of the function
$g_M\colon \R^{S\times A}\to \R^{S\times A}$ associated to the MDP.
When $\alpha=1$, we obtain Kleene iteration which converges to the
least fixpoint for monotone functions over a complete lattice,
starting the iteration from the least element.  In general, it
corresponds to a Mann
iteration~\cite{b:iterative-approximation-fixed-points} which, under
suitable hypotheses, converges to a fixpoint of a continuous function
starting from any initial state.
The twist here is that the function $f$ is not fixed, but can only be
approximated, since the probabilities $\Tr[]{s}{a}{s'}$ (and possibly also the
rewards $R(s,a,s')$) change and will be updated during the iteration.

\paragraph{Aim of the paper.}
Our aim is to develop a fixpoint theory of approximated functions of
which the scenario above is a special case.  There is a large body of
work guaranteeing the existence of fixpoints for certain classes of
functions and providing methods for computing them.  This includes,
for instance, Banach's fixpoint theorem for contractive functions over
complete metric spaces~\cite{Banach1922}, Knaster-Tarski theorem for
monotone functions over complete lattices~\cite{t:lattice-fixed-point}
that is frequently employed in computer science, and Kleene
iteration~\cite{CC:CCVTFP}.

There is much less work on how to compute (least) fixpoints for a function
which is not known exactly, but
can only be obtained by a sequence of subsequently better
approximations.
As we will see, developing such a theory is relatively simple when the
functions of interest are contractions (or power contractions) whose
(repeated) application decreases the distance of two vectors by a factor
$\gamma<1$, by relying on Banach's fixpoint theorem.
This is, for instance, true for reinforcement learning in the discounted
case (cf. the proof of correctness of Q-learning \cite{WD:QL} based on
stochastic approximation theory \cite{bt:neuro-dynamic-programming}).
However, in this paper, we are also and, actually, mostly interested in
the non-discounted case ($\gamma=1$) where the functions are just non-expansive
(the distance of two vectors after function application is
bounded by their original distance) and the aim is to determine their least fixpoints.
We remark that working on the non-discounted case is sometimes the only
appropriate choice, e.g., when the reward represents the likelihood of
eventually satisfying a given property.

We will
concentrate on functions that are monotone and
non-expansive with respect to the supremum norm.  Besides policy
computation for MDPs, a number of other applications consider
value vectors which can be characterized as (least)
fixpoints of functions of this kind, e.g., computing bisimilarity
metrics~\cite{bblm:on-the-fly-exact-journal}, solving (simple)
stochastic
games~\cite{c:complexity-stochastic-game,kkkw:value-iteration-ssg}, or
model-checking quantitative logics on probabilistic systems
\cite{hk:quantitative-analysis-mc,ms:lukasiewicz-mu}.
Other than for
contractive functions, there is no guarantee of uniqueness of the
fixpoint, making it more difficult to approximate the least fixpoint
by methods different from Kleene iteration.  And, as we will see, also
Kleene iteration fails if the function under consideration can be only
approximated.

Hence, the problem statement can be formally expressed as follows:

\smallskip

\fbox{\parbox{0.9\textwidth}{\textsc{Approximation of least fixpoint:}
    Given a sequence of (higher-dimensional) non-expansive and
    monotone functions $f_n\colon \Rp^d\to \Rp^d$, $n\in\N$, that
    converges pointwise to $f\colon \Rp^d\to \Rp^d$, compute a
    sequence $x_0,x_1,x_2,\dots \in \Rp^d$ that converges to the least
    fixpoint of $f$ (if it exists).}}

\smallskip

The main contribution in this paper is the definition of an iteration
scheme, referred to as \emph{dampened Mann iteration scheme} for
(higher-dimensional) functions on the non-negative reals that
converges to the least fixpoint from every starting point.  We
identify conditions which ensure convergence, first for known
functions and then for functions which can only be approximated, as
required in the problem statement, allowing us to handle the MDP
scenario outlined before.

\paragraph*{Least fixpoints of simple functions are not easy to compute.}
As a warm up, we consider some simple examples showing that computing
the least fixpoint of an approximated function is non-trivial.  Let
$f\colon \Rp \to \Rp$ be a non-expansive and monotone function over the
non-negative reals.  A non-expansive function
is also continuous.  Moreover, if it has a fixpoint, it has a least
fixpoint $\mu f$ which can be computed by Kleene iteration, that is,
the sequence $0,f(0),f^2(0),\dots$ converges to $\mu f$.  Now, assume
that $f$ is not known exactly, but can only be approximated via a
sequence of functions $f_n$ converging to $f$.  One might
think that Kleene iteration can be easily adapted to deal with this
situation.  We provide some examples showing that, instead, non-trivial problems arise.

\begin{example}[Non-continuity of the least fixpoint operator]
  \label{it:simple-1}
  Take 
  $f,f_n\colon [0,1] \to [0,1]$, $n\in\N$, with
  $f_n(x) = \frac{1}{n} + (1-\frac{1}{n})\cdot x$, $f(x) = x$, as
  depicted in \Cref{fig:motivation-f}.

  \begin{figure}[t]
    \centering
    \begin{subfigure}[b]{.45\textwidth}
      \centering
      \begin{tikzpicture}[scale=2]
        \draw[->, semithick] (-0.05,0) -- (1.1,0) node[right] {$x$};
        \foreach \x in {0,0.5,1} \draw (\x,0.02) -- +(0,-0.04)
        node[below] {\x};
        \draw[->, semithick] (0,-0.05) -- (0,1.1) node[above] {$f(x)$};
        \foreach \y in {0,0.5,1} \draw (-0.02,\y) -- +(0.04,0)
        node[left] {\y};
        \draw[color=dgreen, domain=0:1, thick, samples=50]
        plot (\x, {\x});
        \draw[color=blue, domain=0:1, thick, samples=50]
        plot (\x, {1});
        \draw[color=blue, domain=0:1, thick, samples=50]
        plot (\x, {0.5+0.5*\x});
        \draw[color=blue, domain=0:1, thick, samples=50]
        plot (\x, {0.33+0.66*\x});
        \draw[color=blue, domain=0:1, thick, samples=50]
        plot (\x, {0.25+0.75*\x});
      \end{tikzpicture}
      \caption{Sequence $(f_n)$ converging to $f$}
      \label{fig:motivation-f}
    \end{subfigure}
    \begin{subfigure}[b]{.45\textwidth}
      \centering
      \begin{tikzpicture}[scale=2]
        \draw[->, semithick] (-0.05,0) -- (1.1,0) node[right] {$x$};
        \foreach \x in {0,0.5,1} \draw (\x,0.02) -- +(0,-0.04)
        node[below] {\x};
        \draw[->, semithick] (0,-0.05) -- (0,1.1) node[above] {$g(x)$};
        \foreach \y in {0,0.5,1} \draw (-0.02,\y) -- +(0.04,0)
        node[left] {\y};
        \draw[color=dgreen, domain=0:1, thick, samples=50]
        plot (\x, {\x});
        \draw[color=blue, domain=0:1, thick, samples=50]
        plot (\x, {1});
        \draw[color=blue, domain=0:0.5, thick, samples=25]
        plot (\x, {0.5});
        \draw[color=blue, domain=0.5:1, thick, samples=25]
        plot (\x, {\x});
        \draw[color=blue, domain=0:0.33, thick, samples=50]
        plot (\x, {0.33});
        \draw[color=blue, domain=0.33:1, thick, samples=50]
        plot (\x, {\x});
      \end{tikzpicture}
      \caption{Sequence $(g_n)$ converging to $g$}
      \label{fig:motivation-g}
    \end{subfigure}
    \caption{}
  \end{figure}

  The least (and only) fixpoint of each approximation $f_n$ is $1$,
  while the least fixpoint of $f$ is $0$.  Hence,
  $\lim_{n\to \infty} \mu f_n = 1 \neq 0 = \mu f$, i.e., the least
  fixpoint operator is non-continuous.  In particular, trying to
  obtain an approximation of the least fixpoint of $f$ by naively
  performing a Kleene iteration on any fixed approximation $f_n$ is
  doomed to fail, as we will always converge to $1$ instead of $0$.
\end{example}

\begin{example}[Kleene over improving approximations fails]
  \label{it:simple-2}
  Take $g,g_n\colon [0,1]\to[0,1]$
  with $g_n(x) = \frac{1}{n}$ if $x\le \frac{1}{n}$, $g_n(x) = x$
  otherwise, and $g(x) = x$, as depicted in \Cref{fig:motivation-g}.
  In this case, $\lim_{n\to \infty} \mu g_n = 0 = \mu g$, hence the
  sequence of least fixpoints of the approximations will converge to
  the least fixpoint of $g$.
  But imagine that we want to reuse intermediate results once we have
  obtained a better approximating function.
  This is what is done in reinforcement learning where one alternates
  sampling and iteration of the estimated value function.
  In this scenario, this would mean that we have already obtained an
  approximation of $\mu g$, based on a Kleene iteration with some $g_n$.
  This will, however, over-estimate the least fixpoint of $g$ by
  $\nicefrac{1}{n}$
  and further iterating on the previously obtained approximation of $\mu g$  with ``better'' functions
  $g_m$ ($m>n$) will never decrease it.
\end{example}
\begin{example}[Approximations might not have fixpoints]
  \label{it:no-fixpoints}
  Let
  $h,h_n\colon \Rp\to \Rp$ where $h_n(x)=x+\nicefrac{1}{n}$ and
  $h(x)=x$.  Here, functions $h_n$ have no fixpoint while the
  least fixpoint of $h$ is $0$.  We want to be able to compute
  $\mu h$ even in such cases.
\end{example}

\paragraph*{Non-expansive functions and MDPs.}
Before identifying potential solutions to the problem illustrated so
far, we discuss how it instantiates to the case of MDPs when
functions are approximated by sampling, a possible and important
application scenario. 

First, estimation by sampling, as explained above, has the property
that whenever a transition probability is $0$ it will never be
over-estimated. (If there is no possibility of taking a
transition, we will never do so.) Based on this, we will argue in
the following that the sequence of (least) fixpoints of approximated
MDPs converges to the least fixpoint of the true MDP value
function.

Second, MDPs with discount factor ($\gamma<1$) induce contractive
value functions.  But in fact, another subclass of MDPs -- MDPs
without end-components -- has the property of inducing power
contractions.  An end-component~\cite{bk:principles-mc} is a subset
$S'\subseteq S$ of states of an MDP for which the agent has a policy
which stays within $S'$ with probability $1$, that is for every
$s\in S'$ there exists an action $a$ such that all $s'$ with
$\Tr[]{s}{a}{s'} > 0$ are contained in $S'$.
Alternatively, one can
characterize MDPs without end-components as those MDPs where each
policy induces a Markov chain in which every state can reach a final
state with non-zero probability.
An example for an MDP with an end
component is given in~\Cref{fig:mdp-end-component}.  In this case,
$S' = \{s_2,s_3\}$ forms an end-component for which the agent has a
policy, which consists in choosing always action $b$, that will never
leave~$S'$.

For MDPs without end-components, it can be shown that there exists a
constant $n$ (the maximal length of a shortest path to a final state)
such that the $n$-th power of the value function is contractive (such
functions are called power contractions).
We will see that this power contraction property of MDPs without end
components simplifies proving the convergence of the fixpoint
iteration.

\begin{figure}[b]
  \centering
\clearpage{}%
\scalebox{0.9}{
\begin{tikzpicture}[every node/.style={draw}]
  \tikzstyle{action}=[>=latex, nodes={font=\footnotesize,sloped,draw=none},fill,draw,black]
  \path[shape=circle, minimum size=1cm]
    (0,0) node(s1)[]{$s_1$}
    (3,0) node(s2)[]{$s_2$}
    (6,0) node(s3)[]{$s_3$}
    (9,0) node(s4)[accepting]{$s_4$};

  \path[action] (-1,0) circle (2pt)
    edge [bend left] (s1)
    edge [bend right, post] (s1);
  \path[action] (1.5,0) circle (2pt)
    edge node[fill=white]{$a$} (s2) 
    edge [post] node[above]{$1\mid 2$} (s1);
  \path[action] (4.5,0.5) circle (2pt)
    edge node[fill=white]{$b$} (s2)
    edge [post] node[above]{$1 \mid 0$} (s3);
  \path[action] (4.5,-0.5) circle (2pt)
    edge node[fill=white]{$b$} (s3)
    edge [post] node[below]{$1 \mid 0$} (s2);
  \path[action] (7.5,0) circle (2pt)
    edge node[fill=white]{$a$} (s3)
    edge [post] node[above]{$1 \mid 1$} (s4);
\end{tikzpicture}
} 

\clearpage{}%

  \caption{An MDP with an end-component}
  \label{fig:mdp-end-component}
\end{figure}
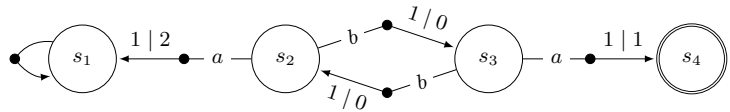

However, general MDPs possibly including end-components induce
non-ex\-pan\-sive fixpoint functions that are not contractive (unless the
fixpoint equation is equipped with a discount $\gamma<1$).
In particular, such MDPs have non-unique solutions.
In the MDP in~\Cref{fig:mdp-end-component}, one can assign any value
larger or equal to $2$ to $s_2,s_3$, obtaining a non-minimal fixpoint.
Intuitively, $s_2,s_3$ form a ``vicious cycle'' where $s_2,s_3$ convince
each other erroneously that they can get an expected return strictly larger
than $2$, because the other ``says so''.
In fact, the maximal payoff achievable from both $s_2,s_3$ is $2$ (from
$s_3$ choose $b$ and go to $s_2$ and from there choose $a$ to go to $s_1$
and collect the reward).
The presence of end-components is known to cause complications, see for
instance~\cite{bccfkkpu:mdps-learning,hm:interval-iteration-mdps}.

The issues with end-components are even more pronounced when we can only
approximate the parameters of an MDP.
In this case, we might over-approximate the reward $2$ in early stages
(assume that the reward arises from a probabilistic choice), such that
$s_2,s_3$ might think that they can, for instance, obtain reward $3$.
Then, in later stages, the approximation will get better and the promise
of reward more realistic, and the reward achieved by going to $s_1$ will
approach $2$.
However, since assigning $3$ to $s_2,s_3$ also results in a fixpoint,
future iterations will never decrease this over-estimated value.
This problem is very similar to the one that arises in
Example~\ref{it:simple-2} above.
Hence, we need to develop new techniques to deal with such a scenario.

\paragraph*{A fixpoint theory for approximated functions.}
Coming back to our general problem statement, before illustrating our
proposal, we observe that some simple solutions which could come to
mind have relevant drawbacks.
As a first option,
one could think of confining the attention to contractive or power
contractive functions, where the problem mentioned above do not occur.
This, however, would restrict applicability to probabilistic systems that
have some form of stopping condition which typically produce functions of
this kind (e.g.\ MDPs without end
components) or force the use of a discount factor.
However, as explained above, MDPs \emph{with} end
components cannot be handled in this way, since the fixpoint
functions induced by such MDPs are not power contractive, but only
non-expansive and introducing a discount factor can be inappropriate.

Another option would be to restart the (Kleene) fixpoint iteration
from scratch for each newly obtained approximation $f_n$, a
technique, called \emph{resetting} in this paper, which resembles certainty
equivalent methods in reinforcement
learning~\cite{kv:certainty-equivalent}.  This is possible, in
principle, in particular when the least fixpoints of the
approximations converge to the least fixpoint of the exact function.
If this is not the case, one has to pay attention not to iterate ``too
much'', as shown by Example~\ref{it:simple-1} above, where one would
converge to $1$ (instead of $0$).  Hence, this method requires to be
able to choose appropriate iteration indices $k_n$ such that the sequence
$(f_n^{k_n}(0))$ converges to $\mu f$.  Furthermore,
even when this approach is viable, it precludes
the reuse of approximations obtained in a previous step, making it
very inefficient and inadequate in scenarios where live value estimates are required.

We propose a solution that does neither: we allow non-expansive
functions and interleave sampling and fixpoint iteration, which allows
us to continue from previously computed approximations.  The central
idea is to combine the fixpoint iteration with a dampening factor that
can be seen as a form of discount factor that ``vanishes'' over time.
It is inspired by work in~\cite{KIM200551} which, however, does not
apply to our setting since it does not deal with approximations.
Furthermore, it puts stronger restrictions on the iteration parameters and the
space where the functions are defined, preventing its use for the
applications we have in mind.

To be more precise, given a function $f$ over $d$-tuples of positive
reals we propose what we call a \emph{dampened Mann iteration}:
\begin{equation}
  \label{eq:damp-mann-intro}
  x_{n+1} = (1 - \beta_n) \cdot  (\alpha_n\cdot x_n + (1-\alpha_n)\cdot f(x_n))
\end{equation}

The genesis of this recursion scheme can be explained as follows: As a first step the right-hand side of a simple Kleene-iteration is replaced, as in the case of
Dyna-Q, with a weighted sum $\alpha \cdot x_n + (1-\alpha) \cdot f(x_n)$ for some parameter $\alpha$. Next, we allow this parameter to be potentially different in every step, i.e. for a fixed sequence $(\alpha_n)$ in the unit interval we consider the sequence generated by the recursion $x_{n+1} := \alpha_n\cdot x_n + (1-\alpha_n)\cdot f(x_n)$. Recursion schemes of this form are called Mann iterations and can in some cases be used to improve convergence properties of the generated sequence to a larger class of functions $f$ than only contractive functions. Indeed, one of the more elementary results on Mann iterations states that if $f\colon [a, b] \to [a, b]$ is any continuous function on a bounded interval, then choosing a parameter sequence $(\alpha_n)$ with $\alpha_n \to 1$ and $\sum_n (1-\alpha_n) = \infty$ guarantees that the corresponding Mann iteration always converges to some fixed point of $f$ for any choice of initial value $x_0 \in [a,b]$ \cite{BORWEIN1991112}.

There are many results that guarantee the (weak) convergence of a Mann iteration in more general spaces under more restrictive assumptions on the function $f$ and/or other assumptions on the parameter sequences \cite{b:iterative-approximation-fixed-points,alma99371219572606441}. The existing results in the literature are, however, not directly applicable to the problem we want to solve since the assumptions on the surrounding space are often too restrictive and, more importantly, since we specifically want to find the \emph{least fixed point} of a function. Note that a Mann iteration that starts on some fixed point $x_0$ of $f$ which is not the least will get stuck there forever. For the notion of a least fixed point to be well-defined we, on the other hand, will -- unlike most literature on Mann iterations we are aware of -- always %
assume that the map $f$ is monotone with respect to the pointwise order.

To be able to always guarantee the convergence to the least fixed point and also to deal with approximations later on, we add -- as a last step -- the dampening factor $(1-\beta_n)$ to the iteration scheme.

Depending on the setting, we will require conditions
on the coefficients $\alpha_n$, they must either converge to $0$
(Mann-Kleene scheme) or to a value strictly less than $1$ (relaxed
Mann-Kleene scheme).
This is combined with a sequence $(1-\beta_n)$ of dampening factors
such that $\lim_{n\to\infty} \beta_n = 0$ and $\sum_n \beta_n = \infty$.
The first condition ensures that we dampen less and less during the
iteration so that we
converge to a Mann iteration, while the second condition
guarantees that the dampening factors always have enough ``power''
left to decrease an over-approximation to the least fixpoint.

We will show that, for a fixed non-expansive and monotone function $f$,
this form of iteration converges to the least fixpoint of $f$ starting from
every point in the domain.
We will also identify sufficient conditions so that the scheme works
when in \eqref{eq:damp-mann-intro} function $f$ is replaced by $f_n$,
a sequence of approximations
converging to $f$.
The resulting theory will allow to treat the case of MDPs, possibly
non-discounted and with end-components.

For more general probabilistic systems which can be explored by
sampling, we will show how to iterate to the least fixpoint
by making sure that the functions are sampled ``fast enough'': using
the law of large numbers and Bernstein's inequality, one can
estimate the probability that the functions $f_n$ are close enough to
$f$ after some sampling.  Then, based on the Borel-Cantelli lemma, we
can give an algorithm which guarantees convergence to the least
fixpoint almost surely.

In summary, the main contributions of our work are:
\begin{itemize}
  
\item The definition of an iteration scheme, referred to as dampened
  Mann iteration, which converges to the least fixpoint of a given
  non-expansive and monotone function $f$ from any starting
  point (\Cref{se:mann}).

\item The identification of various sufficient conditions for applying
  such iteration scheme in a setting where the function $f$ is unknown and can only be
  approximated via a sequence of functions $f_n$ converging
  to it (\Cref{se:approximation}).
  
\item The instantiation of dampened Mann iteration to the
  computation of the maximal expected return for MDPs where the model is
  unknown and explored via sampling (\Cref{se:mdp}).

\item The instantiation of dampened Mann iteration to systems where
  the function of interest can be approximated with given
  probabilistic error bounds (this includes various probabilistic
  systems which can be explored via sampling, such as simple stochastic games) (\Cref{se:mann-error-algo}).

\end{itemize}

We view this work as an important step towards extending fixpoint
theory to deal with scenarios arising from machine learning and
data-driven applications.

An appendix contains full proofs of the results and some additional material.

\section{Preliminaries and Notation}
\label{se:preliminaries}

We introduce some basic notions and notation on finite-dimensional
real spaces and Markov decision processes~\cite{Bell:MDP}, that will
be our main application.

\subsection{Finite-Dimensional Real Spaces}

We denote by $\Rp$ the set of non-negative reals, by $\N_0$ and $\N$
the sets of naturals with and without $0$, respectively.
For sets $X,Y$, we will denote by $Y^X$ the set of functions
$f\colon X\to Y$.  If $X$ is finite, we will identify $Y^X$ with the
set of vectors $Y^{|X|}$.  For $x\in\R^d,$ we will write
$\p{x}{i}\in\R$ for its $i$-th component and, by abuse of notation,
we will sometimes identify $x\in\R$ with the vector
$(x,\dots,x)\in\R^d$.

We equip $\R^d$ with the \emph{supremum norm}
defined as $\|x\|=\sup\{|\p{x}{i}|\mid i\in\{1,\ldots,d\}\}$ for
$x\in\R^d$ and extended to functions $f\colon A\to B$, where $A,B\subseteq\R^d$ by
$\|f\|=\sup_{x\in A}\|f(x)\|\in[0,\infty]$.  Given $x\in\R^d$, we
denote by $|x|$ the vector in $\Rp^d$ defined by
$\p{|x|}{i} = |\p{x}{i}|$.  Note that $\|x\|=\||x|\|$.

For $x,y\in\R^d,$ we write $x\leq y$ for the pointwise (partial)
order, i.e., $\p{x}{i}\leq\p{y}{i}$ for all $i\in\{1,\ldots,d\}$.  A
function $f\in Y^X$ with $X,Y\subseteq\R^d$ is  \emph{monotone}
if, for all $x,x'\in X$, we have that $x\leq x'$ implies
$f(x)\leq f(x')$. A \emph{fixpoint} of $f\colon X \to X$ is $x \in X$ such
that $f(x) = x$. The \emph{set of fixpoints} of a function $f$ is denoted by $\fix{f}$ and, in case it exists, the \emph{least fixpoint} of $f$
is denoted by $\mu f$.

A function $f\in Y^X$ with $X,Y\subseteq\R^d$ is
$L$-Lipschitz for a constant $L\in\Rp$ if, for all $x,x'\in X$,
$\|f(x) - f(x')\| \leq L \|x-x'\|$.  The function is
\emph{non-expansive} if it is $1$-Lipschitz.  It is 
\emph{contractive} if it is $q$-Lipschitz for some $q < 1$, and it is
 a \emph{power contraction} if there is some $n\in\N$ such that
$f^n$, the $n$-fold composition of $f$, is a contraction. If
$f\colon X \to X$ is a contraction (and, more generally, power
contraction) with $X \subseteq \R^d$ closed (hence
complete), by Banach's theorem, it has a unique fixpoint given by
$\lim_{n \to \infty} f^n(x_0)$ where $x_0 \in X$ is arbitrary.

\subsection{Markov Decision Processes}
\label{se:MDP-basics}

We will focus on Markov decision processes~\cite{Bell:MDP} in a
non-discounted non-negative reward setting.

\begin{definition}[Markov decision process]
  A \emph{Markov decision process}\linebreak \emph{(MDP)} is a tuple
  $M=(S,A,T,R)$ where $S$ and $A$ are the finite sets of states and
  actions. Moreover $T\colon S \times A \times S \to [0,1]$ provides
  the probability $\Tr[]{s}{a}{s'}$ of transitioning from state
  $s\in S$ to state $s'\in S$ when action $a\in A$ is chosen, in a way
  that $\sum_{s' \in S} \Tr[]{s}{a}{s'} \in \{0,1\}$. When
  $\sum_{s' \in S} \Tr[]{s}{a}{s'} = 1$ we say that action $a$ is
  \emph{enabled} in $s$ and we write $A(s)$ for the set of actions
  enabled in $s$.
  Finally, $R\colon S\times A\times S\rightarrow\Rp$ is a non-negative
  step-wise reward function.  We let
  $F=F(M) \coloneqq \{s\in S\mid A(s)=\emptyset\}$ and call it the set
  of \emph{final states}.  The Markov decision process is  a
  Markov chain if, for all $s \in S$, we have $|A(s)| \leq 1$.
\end{definition}

The idea is that $M$ describes an interactive system that, when in
state $s$, transitions to another state $s'$ with probability
$\Tr[]{s}{a}{s'}$ based on an action $a \in A(s)$ chosen by an external agent.
Possible strategies of the agent can be described by (positional) policies
$\pi\colon S\setminus F\to A$ with $\pi(s)\in A(s)$ for any
non-final state $s\in S\setminus F$.  We denote the set of all such
policies by $\Pi(M)$.

For MDPs, one is interested in finding a policy
that optimizes the expected return.
It should be noted that
positional policies (as defined above) are sufficient for optimal
behaviour in \emph{finite} MDPs where the goal is to maximize the
non-discounted total reward~\cite{p:mdp-disc-stoch-dyn-prog}.
The essential step 
in solving this optimization problem
often lies in finding the (least) fixpoint of a so-called Bellman
operator. In reinforcement learning, one often considers the state-action-value
operator, discussed in the introduction,
whose
fixpoint gives the expected optimal return $q^*(s,a)$ at state $s$ when taking
action $a$. Alternatively, one can consider the state-value operator
$f_M$ whose fixpoint gives the expected optimal return $v^*(s)$ at state
$s$.

\begin{definition}[Bellman operator]
  \label{de:bellman}
  Given an MDP $M = (S, A, T, R)$,  the (\emph{state-value}) Bellman operator
  $f_M\colon \Rp^S\rightarrow\Rp^S$ is defined, for $v \in  \Rp^S$, by
  \begin{equation}\label{eq:state-value-operator}
    f_M(v)(s)=\max_{a\in A(s)}\sum_{s'\in S}\Tr[]{s}{a}{s'}(R(s,a,s')+v(s'))
  \end{equation}
  and the \emph{state-action-value} Bellman operator
  $g_M\colon \Rp^{S\times A}\rightarrow\Rp^{S\times A}$ is defined, for $q\in\Rp^{S \times A}$,  by\footnote{We assume $\max\emptyset\coloneqq 0$.}
  \begin{equation}\label{eq:state-action-value-operator}
    g_M(q)(s,a)=\sum_{s'\in S}\Tr[]{s}{a}{s'}(R(s,a,s')+\max_{a'\in A}q(s',a')).
  \end{equation}
\end{definition}

In contrast to the usually treated \emph{discounted} case, the
operator $f_M$ is not necessarily contractive. Hence, it might have
several fixpoints or might not have any fixpoint in $\Rp^S$,
intuitively because the expected return might sum up to infinity
(thus, to ensure the existence of a fixpoint, one should work in
  $[0,\infty]^S$ or $[0,\infty]^{S\times A}$, respectively.)
Hereafter, we will work with so-called \emph{MDPs with
  finite value}, for which a fixpoint exists. In
\Cref{ss:general-mdps}, they will be characterized as those MDPs where
positive rewards are given only outside end-components.
By a later result (\Cref{le:ex-fix}),
existence of a fixpoint for the Bellman operator implies the existence of a least fixpoint.
For an MDP with finite value $M = (S,A,T,R)$, we let
\begin{center}
  $\opts{M} = \mu f_M$
 \qquad and \qquad $\optsa{M} = \mu g_M$.
\end{center}
It can be easily seen that
\begin{align*}
  \opts{M}(s) &= \max_{a \in A(s)} \optsa{M}(s,a) \\
  \optsa{M}(s,a) &= \sum_{s'\in S}\Tr[]{s}{a}{s'}(R(s,a,s')+\opts{M}(s')).
\end{align*}

For a fixed policy $\pi\in\Pi(M)$, the MDP can be interpreted as a
Markov chain by removing all actions not prescribed by the policy.
We denote this Markov chain by $M^\pi$, the corresponding
state-value iteration functions by $f^\pi_M=f_{M^\pi}$, and its least
fixpoints by $v_M^\pi=\mu f_M^\pi$.
  Then
   , one can see that
   \[
     f_M(v)(s)=\max_{\pi\in\Pi(M)} f^\pi_M(v)(s)
   \]
  and, by existence of memoryless optimal policies,
   \[
     \opts{M}(s)=\max_{\pi\in\Pi(M)}v_M^\pi(s).
   \]

\section{Dampened Mann Iteration
  for Known Functions}
\label{se:mann}

We define an iteration scheme for approximating, with arbitrary
precision, the least fixpoint of monotone and non-expansive functions
over the non-negative reals. It is a variation of Mann
iteration~\cite{b:iterative-approximation-fixed-points} suitably
modified in order to be ``robust'' with respect to perturbations in
the computation with the introduction of a dampening
factor~\cite{KIM200551,hz:modified-mann-iteration}.  In this section,
we will focus on the case in which the function of interest is known
and identify conditions which ensure convergence to the least
fixpoint. The case in which the function can only be approximated will
be discussed in the next section.

\subsection{Dampened Mann Iteration Scheme}

We start by clarifying the class of functions of interest in the paper.

\begin{assumption}
  \label{as:setting}
  Given a \emph{closed} and \emph{convex} set $X\subseteq\Rp^d$ with $0\in X$, where
  $d \in \N_0$ is a fixed arity, let
   $f\colon X\to X$ be \emph{monotone} and \emph{non-expansive}
  with
  $\fix{f}\neq\emptyset$.
\end{assumption}

Note that, under the above assumption,
 since $X$ is closed,
if $f$ is a (power) contraction, by Banach's theorem it has a unique
fixpoint, i.e., $\fix{f}=\{\mu f\}$.  Otherwise, if $f$ is just
non-expansive, the condition $\fix{f}\neq\emptyset$ is non-trivial
(e.g., the function $x\mapsto x+1$ over $\Rp$ is monotone and
non-expansive, but it has no fixpoints).  However, when
$\fix{f}\neq\emptyset$, we can show that $f$ also admits a
least fixpoint $\mu f$.

\begin{lemmarep}[Existence of least fixpoints]
  \label{le:ex-fix}
  Let $X\subseteq\Rp^d$ be a closed set with $0\in X$ and let
  $f\colon X \to X$ be a monotone function with
  $\fix{f} \neq \emptyset$.  Then $f$ has a least fixpoint $\mu f$. Moreover, if
  $f$ is continuous then $\mu f = \sup_{n \in \N} f^n(0)$.
\end{lemmarep}

\begin{proof}
  Let $Y = \{ x \in X \mid \forall z \in \fix{f}: x \leq z \}$,
  i.e., $Y$ is the set of elements of $X$ below all fixpoints of $f$.

  Observe that $Y$ is a pointed dcpo.
  In fact, it has a bottom element $0 \in Y$.
  Moreover, given any directed set $D \subseteq Y$, its
  supremum $d = \sup D$, obtained as the pointwise supremum, is in
  $Y$.
  In fact, for all components $i = \{ 1, \ldots, d \}$, we can
  consider a sequence $(x_{i,n})$ such that
  $\sup_n \p{x_{i,n}}{i} = \p{d}{i}$.
  Using the fact that $D$ is directed we can obtain a sequence
  $(y_n)$ in $D$ such that $x_{n,i} \leq y_n$ for all
  $i \in \{1,\ldots, d\}$.
  Then $\sup_n y_n = d$ and thus $d \in X$ by closedness.
  Moreover, since each $y_n \in Y$, for all $z \in \fix{f}$, we
  have $y_n \leq z$ and thus $d = \sup_n y_n \leq z$.
  Therefore $d \in Y$, as desired.

  Additionally observe that for all $x \in Y$ we have $f(x) \in Y$.
  In fact, take $x \in Y$.
  For all $z \in \fix{f}$ it holds $x \leq z$ and thus, by monotonicity
  $f(x) \leq f(z) \leq z$.
  Hence $f(x) \in Y$.

  By the above observation, we can consider the restriction of $f$ to
  $Y$, denoted $\restr{f}{Y}\colon Y \to Y$.
  We have that $\restr{f}{Y}$ is a monotone function over a dcpo.
  By Pataraia's theorem~\cite{PAT97:CP,DP:ILFPT}, $\restr{f}{Y}$ has a
  least fixpoint, which is clearly a fixpoint of $f$ and it is the
  least by definition of $Y$.

  The last part follows immediately by the fact that whenever $f$ and
  thus $\restr{f}{Y}$ is continuous
  $\mu f = \mu \restr{f}{Y} = \sup_{n \in \N} \restr{f}{Y}^n(0) =
  \sup_{n \in \N} f^n(0)$.
  \qed
\end{proof}

Since our functions are non-expansive and thus continuous, by the lemma above, the least fixpoint can be obtained by iterating the function on $0$, according to what is often called \emph{Kleene iteration} (see, e.g.~\cite{CC:CCVTFP}).

However, as noted in the introduction, Kleene iteration is not ``robust'' with
respect to perturbations in the computation. For this reason, we
introduce a variation of a Mann iteration scheme. 

\begin{definition}[Mann scheme/iteration]
  \label{de:dampened-mann}
  A \emph{(dampened) Mann scheme} 
  $\mathcal{S}=((\alpha_n)_{n\in\N_0},(\beta_n)_{n\in\N_0})$ is a pair of parameter
  sequences in $[0,1)$ such that
  \begin{align}
    \label{eq:beta-condition-1}
    \lim_{n \to \infty} \beta_n &= 0, \\
    \label{eq:beta-condition-2}
    \sum_{n \in \mathbb{N}_0} \beta_n &= \infty 
    \quad \text{(or equivalently, $\prod_{n \in \mathbb{N}_0} (1 - \beta_n) = 0$)}.
  \end{align}
  A Mann-Scheme is called a \emph{Mann-Kleene scheme} whenever
  $\lim_{n\to\infty} \alpha_n = 0$ and a \emph{relaxed Mann-Kleene
    scheme} if $\lim_{n\to\infty} \alpha_n < 1$.
  
  Given a convex $X\subseteq\Rp^d$ with $0\in X$, a Mann scheme
  $\mathcal{S}$ defines a sequence $(T^\mathcal{S}_n)_{n\in\N_0}$ of
  operators $T^\mathcal{S}_n\colon X^X\times X \to X$ given by
  \[
    T_n^\mathcal{S}(f,x) = (1-\beta_n)\cdot\left(\alpha_n x
      +(1-\alpha_n)f(x)\right).
  \]
  Together with a sequence $(f_n)_{n\in\N_0}$
  of functions $f_n\colon X\to X$ and an initial point $x_0\in X$, it gives
  rise to a \emph{(dampened) Mann iteration}
  $\mathcal{F}=(\mathcal{S},(f_n),x_0)$, determining a sequence
  $(x_n^{\mathcal{F}})_{n\in\N_0}$ defined as
  \[
    x_0^{\mathcal{F}}=x_0, \qquad x_{n+1}^{\mathcal{F}}= T_n^{\mathcal{S}}(f_n,x_n^{\mathcal{F}}).
  \]
\end{definition}

Note that the iteration can start at any
point (not just $0$). This might look irrelevant in cases where the
function of interest is known exactly, but it will be extremely useful
when
it can only be approximated.

Intuitively, when trying to approximate the least fixpoint,
Condition~\eqref{eq:beta-condition-1} ensures that dampening
eventually reduces -- meaning that, in the long run, when we are close
to the least fixpoint of $f$, we stay close to it.  On the other
hand, Condition~\eqref{eq:beta-condition-2} guarantees that at any
stage there is still enough ``dampening power'' left to correct
possible over-approximations. In the exact case, when $f_n = f$ for all $n$, this is needed for
the iteration to be convergent for \emph{all} initial points.
In the approximated case, when $(f_n)$ converges to $f$, dampening will be
indispensable since over-approximations at some components can be introduced along the way.

  Conditions~\eqref{eq:beta-condition-1} and \eqref{eq:beta-condition-2}
are indeed necessary as shown by the examples below:

\begin{itemize}
  \item for $f\colon [0,1]\to[0,1]$, $x\mapsto1$ and $x_0\in X$ arbitrary, we have
    \[
      x_{n+1}^{\mathcal{F}}=(1-\beta_n)(\alpha_n x_n^{\mathcal{F}}+(1-\alpha_n) f(x_n^{\mathcal{F}})) = (1-\beta_n)(\alpha_n x_n^{\mathcal{F}}+1-\alpha_n)\leq(1-\beta_n).\]
    Hence, if $x_n^{\mathcal{F}}\rightarrow 1=\mu f$, then also $\beta_n\to 0$.
    
  \item for $f = \id\colon [0,1]\to[0,1]$, $x\mapsto x$ and $x_0>0$ arbitrary, we have
    \[
      x_{n+1}^{\mathcal{F}}=(1-\beta_n)(\alpha_n x_n^{\mathcal{F}}+(1-\alpha_n)x_n^{\mathcal{F}})=(1-\beta_n)x_n^{\mathcal{F}}=x_0\prod_{k=0}^n(1-\beta_k).\]
    Hence, if $x_n^{\mathcal{F}}\to 0=\mu f$, then necessarily also $\prod_{n\in\N_0}(1-\beta_n)=0$.
\end{itemize}

\subsection{Approximating the Fixpoint of Known Functions}

For an approximation scheme to be reasonable, it should at least
converge to the correct solution $\mu f$ when we use the exact
function $f$ at every step. Hence, we first focus on the case in which
$f_n=f$ for all $n\in\N_0$ and identify conditions which ensure
convergence.  By abuse of notation, we will sometimes identify $f$ with the
sequence $(f_n)$ with $f_n=f$ for all $n\in\N_0$.

\begin{definition}[Exact Mann scheme]
  \label{de:exact}
  Let $\mathcal{S}=((\alpha_n),(\beta_n))$ be a Mann scheme. We call it \emph{exact} if for all
   $f\colon X\to X$ as in \Cref{as:setting} and
  $x_0\in X$, the sequence
  $(x_n^{\mathcal{F}})$ generated by the iteration
  $\mathcal{F}=(\mathcal{S},f,x_0)$ converges to $\mu f$.
\end{definition}

One can prove that, 
for functions satisfying
\Cref{as:setting}, even if the domain $X$
might not be bounded, the sequence generated
by a Mann scheme is bounded.

\begin{lemmarep}
  \label{lem:bounded-exact}
  Under \Cref{as:setting}, for any Mann scheme
  $\mathcal{S}=((\alpha_n),(\beta_n))$ and $x_0\in X$, the sequence
  $(x_n^{\mathcal{F}})$ generated by the iteration
  $\mathcal{F}=(\mathcal{S},f,x_0)$ is bounded. In fact, for all
  $\bar{x}\in\fix{f}$ and $n\in\N$
  \[\|x_{n}^{\mathcal{F}}-\bar{x}\|\leq\max(\|\bar{x}\|,\|x_0-\bar{x}\|)\]  
\end{lemmarep}
\begin{proof}
  Let $\bar{x}\in\fix{f}$.
  Then, for all $n\in\N_0$, we have
  \begin{align*}
    \|x_{n+1}^\mathcal{F}-\bar{x}\|
    & \leq \beta_n\|\bar{x}\| + (1-\beta_n)
      \|\alpha_n(x_n^\mathcal{F}-\bar{x})+(1-\alpha_n) (f(x_n^\mathcal{F})-\bar{x})\|\\
    & \leq\beta_n\|\bar{x}\|+(1-\beta_n)\|x_n^\mathcal{F}-\bar{x}\|\\
    & \leq\max(\|\bar{x}\|,\|x_n^\mathcal{F}-\bar{x}\|).
  \end{align*}
  Inductively, we immediately conclude.
  \qed
\end{proof}
However,
as one might expect,
it is not generally the case that the
sequence converges to the least fixpoint of $f$ (or that it converges
at all) without imposing additional restrictions on the parameters.

We next prove that convergence of the iteration to the least fixed point is ensured by working with Mann-Kleene schemes, i.e., Mann schemes satisfying
\begin{equation}
  \label{eq:alpha-condition-1}
  \lim_{n\to\infty} \alpha_n = 0.
\end{equation}
Intuitively, Condition~\eqref{eq:alpha-condition-1} implies that the
iteration gets closer and closer to a Kleene iteration (in fact, when
$\alpha_n \to 0$, the operator $T^\mathcal{S}_n(f,\cdot)$ tends to
$f$).
At the end of this section we will show that convergence even
holds for relaxed Mann-Kleene schemes where the $\alpha_n$ converge to
a value strictly less than $1$.

Canonical choices of the parameters for obtaining a Mann-Kleene scheme
are $\beta_n = \nicefrac{1}{n}$ and either
$\alpha_n = \nicefrac{1}{n}$ as well or 
$\alpha_n = 0$ constantly
(note that we do not require that  $\sum \alpha_n = \infty$).
In the sequel we sometimes start at
$x_n$ with $n>0$ rather than $x_0$ to ensure that all parameters are
well-defined, particular when they are based on fractions such as
$\nicefrac{1}{n}$.

\begin{theoremrep}[Approximating the fixpoint of known functions]
  \label{th:fixed}
  Every Mann-Kleene scheme $\mathcal{S}$ is exact.
\end{theoremrep}
\begin{proof}
  let $\mathcal{S}$ be a Mann-Kleene scheme.  Let $f$ as in
  \Cref{as:setting} and let $x_0\in X$. Consider the iteration
  $\mathcal{F}=(\mathcal{S},f,x_0)$ and let $(x_n^{\mathcal{F}})$ the
  sequence generated.
  
Assume without loss of generality that $\alpha_n = 0$ for all $n$.
Define a map $s\colon \N \to \R$ by
\[s(n) := \|f^n(0) - \mu f\|\]
By Kleene's Theorem the map $s$ is non-increasing and
$\lim_{n\to\infty} s(n) = 0$.
We already know by \Cref{lem:bounded-exact} that the sequence $(x^{\mathcal{F}}_n)$
is norm bounded, that means we can choose $M$ such that
$\|\mu f\| + \|x^{\mathcal{F}}_n - \mu f\| \le M$
for all $n\in\N_0$.
Define for all $n\in\N$ a number $\delta_n$ as
\[\delta_n := \min\Big\{s(k) + M\sum_{i=m}^{m+k-1} \beta_i \mid n = m + k\Big\}\]
It is easy to see that $\delta_n \to 0$.
We claim that $x^{\mathcal{F}}_n \ge \mu f - \delta_n$. Fix a partition $n = m+k$ and note that $f^k(x^{\mathcal{F}}_m) \ge f^k(0)$ and therefore
\[x^{\mathcal{F}}_{m+k} \ge f^k(0) - M\sum_{i=m}^{m+k-1} \beta_i \ge \mu f - \Big( s(k) + M\sum_{i=m}^{m+k-1} \beta_i \Big)\]
For an upper bound define $\eta_n := M \cdot \prod_{i=0}^{n-1} (1-\beta_i)$. By assumption on the sequence $(\beta_n)$ we know that $\lim_{n\to\infty} \eta_n = 0$. Let $(y_n)$ be the dampened Mann Iteration with the same parameters but starting from $0$, then $y_n \le \mu f$ for all $n\in\N$ and by induction we see that
\[\|x^{\mathcal{F}}_n - y_n\| \le \|x_0 - y_0\| \cdot \prod_{i=0}^{n-1} (1 - \beta_i) \le M \cdot \prod_{i=0}^{n-1} (1-\beta_i)\]
In particular $x^{\mathcal{F}}_n \le \mu f + \eta_n$. Putting $\epsilon_n := \max\{\delta_n, \eta_n\}$ we obtain
\[\|x^{\mathcal{F}}_n - \mu f\| \le \epsilon_n\]
with $\lim_{n\to\infty} \epsilon_n = 0$. The case where $(\alpha_n)$ only converges to $0$ is analogous (see the proof of \Cref{lm:lower}).
  \qed
\end{proof}

Our proof of the theorem actually shows more.  Given a Mann-Kleene
scheme $\mathcal{S}$, let $(x_n^{\mathcal{F}})$ be the sequence
generated by the corresponding iteration $\mathcal{F}=(\mathcal{S},f,x)$ from an arbitrary $x\in X$.
If we can estimate the error between the usual Kleene iteration and
$\mu f$, say we have a computable function $s\colon \N_0 \to \Rp$ fulfilling
$\|f^n(0) - \mu f\| \le s(n)$ and $\lim_{n\to\infty} s(n) = 0$, then
we can also calculate error bounds for Mann-Kleene schemes. That means
that we have also a computable function $n \mapsto \epsilon_n$ such
that $\lim_{n\to\infty} \epsilon_n = 0$ and
$\|x^{\mathcal{F}}_n - \mu f\| \le \epsilon_n$.
 
Note that, although so far we have only considered the case in which
$f_n = f$ for every $n\in\N_0$, \Cref{th:fixed} already yields a
significant improvement to the usual Kleene approximation,
where one sets $x_{n+1} = f(x_n)$: The Mann-Kleene iteration with a
non-expansive function converges to the least fixpoint of $f$ for
\emph{every} starting point $x_0$ and not just for $x_0 = 0$.

With some additional effort one can prove convergence also for relaxed
Mann-Kleene schemes where Condition~\eqref{eq:alpha-condition-1}
$\alpha_n \to 0$ is relaxed to $\alpha_n\to\alpha < 1$.  This for
instance allows to set $\alpha_n=\alpha\in[0,1)$.

\begin{corollaryrep}[Convergence for relaxed Mann-Kleene]
  \label{co:RelaxedMannKleene}
  Every relaxed Mann-Kleene scheme is exact.
\end{corollaryrep}

\begin{proof}
  Let $\mathcal{S}$ be a relaxed Mann-Kleene scheme and let
  $x_0\in X$. Consider the sequence $(x_n^{\mathcal{F}})$ determined by
  the iteration $\mathcal{F}=(\mathcal{S},f,x_0)$.
  For showing that it converges to $\mu f$, we can use the upper bound
  given in the proof of \Cref{th:fixed} and, for the lower bound, we use
  the more general \Cref{lm:lower}.
  \qed
\end{proof}

While a general proof requires non-trivial arguments and can be found
in the appendix, when $\alpha_n$
converges to $\alpha$ from above, an elementary proof can be
given. Write $f^\alpha\coloneqq(1-\alpha)f+\alpha\,\id$ and
$\alpha_n'\coloneqq\frac{\alpha_n-\alpha}{1-\alpha}\in[0,1]$. Then we
have
\begin{align*}
  x_{n+1}&=(1-\beta_n)(\alpha_n x_n+(1-\alpha_n)f(x_n))\\
         &=(1-\beta_n)((\alpha_n-\frac{(1-\alpha_n)\alpha}{1-\alpha})x_n+\frac{1-\alpha_n}{1-\alpha}f^\alpha(x_n))\\
         &=(1-\beta_n)(\frac{\alpha_n-\alpha}{1-\alpha}x_n+\frac{1-\alpha_n}{1-\alpha}f^\alpha(x_n))\\
         &=(1-\beta_n)(\alpha_n'x_n+(1-\alpha_n')f^\alpha(x_n)).
\end{align*}
As $\alpha_n'\to0$ and $\fix{f^\alpha}=\fix{f}$, convergence
directly follows from the convergence under Condition~\eqref{eq:alpha-condition-1}.
Additional results for  more general parameter sequences can be found in \Cref{app:KnownFunctions}.

\begin{toappendix}
\label{app:KnownFunctions}

Towards a convergence result for more general parameter sequences, we introduce the important concept of asymptotic regularity.

\begin{definition}[Asymptotic $f$-regularity]
  \label{de:asym-reg}
  Let $f\colon X \to X$ be a function where $(X, d)$ is a metric
  space. A sequence $(x_n)$ in $X$ is
  \emph{asymptotically regular with respect to $f$}, or
  \emph{$f$-asymptotically regular} if
  $\lim_{n\to\infty} d(x_n, f(x_n)) = 0$.
\end{definition}

It is easy to see that \emph{$f$-asymptotic regularity} is indeed
necessary for the sequence $(x_n)$ to converge to a fixpoint. Simple
examples involving the harmonic series show that $f$-asymptotic
regularity is, in general, not sufficient to guarantee convergence
though. It turns out that in the cases we are interested in
$f$-asymptotic regularity is indeed sufficient to show convergence, it
is  even equivalent to convergence to the least fixpoint.

\begin{theorem}
  \label{th:f-asympt-muf}
  Under \Cref{as:setting}, given a Mann scheme $\mathcal{S}=((\alpha_n),(\beta_n))$ and an arbitrary $x\in X$,
  we consider the iteration $\mathcal{F}=(\mathcal{S},f,x)$.

  Then, if the sequence $(x_n^{\mathcal{F}})$ is asymptotically regular with respect to $f$,
  it converges to $\mu f$.
\end{theorem}
\begin{proof}
  In order to simplify the notation let us write $(x_n)$ for $(x_n^{\mathcal{F}})$.

  By assumption, we immediately obtain that every cluster point of $(x_n)$ is a fixpoint of $f$.
  In fact, let $x$ be a cluster point and let $(x_{n_k})_k$ be a subsequence converging to $x$.
  By continuity of $f$, we have that $0=\lim_{k\to\infty}\|f(x_{n_k})-x_{n_k}\|=\|f(x)-x\|$.
  Hence $f(x)=x$ as desired.

  By \Cref{lem:bounded-exact}, $(x_n)$ is bounded whence Bolzano-Weierstraß gives the
  existence of a cluster point.
  We next show that the sequence even has a unique cluster point, which is thus, again using
  boundedness, the limit of the sequence.
  Assume for a contradiction that there are two distinct cluster points $x$ and $x'$ and take
  $1\le i\le d$ with $\p{x}{i}\neq\p{x'}{i}$.
  We can assume without loss of generality that $\p{x}{i}<\p{x'}{i}$ and take $\epsilon>0$ with
  $\epsilon<|\p{x}{i}-\p{x'}{i}|/2$.
  Since $x$ is a cluster point, there is $N\in\N$ such that $\|x_N-x\|\le\epsilon$.
  It is easy to see by induction, that $x_{N+n}\le x+\epsilon$ for all $n\in\N$ due to $f$ being
  monotone and non-expansive, and $x$ being a fixpoint of $f$:

  The base case $n=0$ holds by assumption.
  For the inductive case, assume that $x_{N+n} \le x + \epsilon$.
  Then observe that by using monotonicity, non-expansiveness and the inductive hypothesis
  \[f(x_{N+n})\leq f(x+\epsilon)\leq f(x)+\epsilon=x+\epsilon\]
  and thus we have
  \begin{align*}
    x_{N+n+1}
    & = (1-\beta_{N+n})(\alpha_{N+n} \, x_{N+n} + (1-\alpha_{N+n}) f(x_{N+n}) )\\
    & \leq \alpha_{N+n} \, x_{N+n} + (1-\alpha_{N+n}) f(x_{N+n})\\
    & \leq \alpha_{N+n} (x + \epsilon) + (1-\alpha_{N+n}) (x + \epsilon)\\
    & = x + \epsilon
  \end{align*}

  From the above observation, it follows that $x'$ cannot be a cluster point, since
  $\|x'-x_{N+n}\|\geq\p{x'}{i} - \p{x_{N+n}}{i} = \p{x'}{i} - \p{x}{i} + \p{x}{i} - \p{x_{N+n}}{i} \geq 2
  \epsilon - \epsilon = \epsilon$, a contradiction.

  \medskip

  It remains to show that the unique cluster point of the sequence is equal to $\mu f$.
  Assume $\bar{x}=\lim_{n\rightarrow\infty}x_n\neq\mu f$.
  As $\bar{x}$ is a fixpoint of $f$, we have $\mu f\leq\bar{x}$.
  Let $K\coloneqq\{k\in\{1,\dots,d\}\mid|\p{\mu f}{k}-\p{\bar{x}}{k}|=\|\mu f-\bar{x}\|\}$, i.e.,
  the set of all dimensions/indices where $\mu f$ and $\bar{x}$ have maximal distance.

  Furthermore, fix $\epsilon > 0$ such that
  \begin{equation}
    \label{eq:eps}
    \epsilon < (\|\mu f-\bar{x}\|-\max_{k\notin K}|\p{\mu f}{k} - \p{\bar{x}}{k}|)/2
  \end{equation}
  and let $N$ be such that
  $\|x_n-\bar{x}\|<\epsilon$ for all $n>N$.

  Then, for all $n > N$, if $k \notin K$ then
  \[
    |\p{x_n}{k}- \p{\mu f}{k} |
    \leq |\p{x_n}{k} - \p{\bar{x}}{k}| + |\p{\bar{x}}{k} - \p{\mu f}{k}|
    < \|\mu f-\bar{x}\|-\epsilon
  \]
  where the last passage uses the fact that
  $|\p{x_n}{k} - \p{\bar{x}}{k}| \leq \| x_n - \bar{x} \| \leq \epsilon$ and, by~\eqref{eq:eps},
  $|\p{\bar{x}}{k} - \p{\mu f}{k}|  \leq \| \mu f - \bar{x} \| - 2 \epsilon$.

  Instead, for $k\in K$ we have
  \[
    |\p{x_n}{k} -\p{\mu f}{k}|
    \geq |\p{\bar{x}}{k} - \p{\mu f}{k}| - |\p{x_n}{k} - \p{\bar{x}}{k}|
    >\|\mu f-\bar{x}\|-\epsilon
  \]

  Thus, in particular, for all $n>N$
  \[
    \|x_n-\mu f\|=\max_{k\in K}|\p{x_n}{k} - \p{\mu f}{k}|,
  \]
  i.e., the distance between $x_n$ and $\mu f$ is realised on components $k \in K$ (with $K$ being
  \emph{independent} of $n$).

  Let $k\in K$ and $n>N$ be arbitrary.
  Note that we have
  \[
    |\p{x_{n+1}}{k} - \p{\bar{x}}{k}|
    < \epsilon \quad \land \quad \p{\bar{x}}{k} - \p{\mu f}{k}=\|\bar{x}-\mu f\|>2\epsilon\]
  whence $\p{x_{n+1}}{k} > \p{\mu f}{k}$.
  But then
  \begin{align*}
    & |\p{x_{n+1}}{k} - \p{\mu f}{k}|\\
    & = \p{x_{n+1}}{k} - \p{\mu f}{k}\\
    & < (1-\beta_n)\left(\alpha_n \p{x_n}{k} + (1-\alpha_n) \p{f(x_n)}{k}- \p{\mu f}{k}\right)\\
    & = (1-\beta_n)\left(\alpha_n (\p{x_n}{k} - \p{\mu f}{k}) + (1-\alpha_n) (\p{f(x_n)}{k}- \p{\mu f}{k})\right)\\
    &\leq (1-\beta_n) \left(\alpha_n\|x_n-\mu f\|+(1-\alpha_n)\|f(x_n)-\mu f\|\right)\\
    & = (1-\beta_n) \left(\alpha_n\|x_n-\mu f\|+(1-\alpha_n)\|f(x_n)-f(\mu f)\|\right)\\
    & \leq (1-\beta_n) \left(\alpha_n\|x_n-\mu f\|+(1-\alpha_n)\|x_n-\mu f\|\right)\\
    &\leq (1-\beta_n)\|x_n-\mu f\|
  \end{align*}
  for all $n>N$.

  Therefore $\| x_{n+1} - \mu f\| \leq (1 - \beta_{n+1}) \|x_{n} - \mu f\|$ and thus for all $h \in \N$
  \[
    \| x_{n+h} - \mu f\| \leq \prod_{i=1}^h (1 - \beta_{n+i}) \|x_{n} - \mu f\| \leq \prod_{i=1}^h (1 - \beta_{n+i})
  \]
  yielding $x_n\rightarrow\mu f$ as $\prod_{n \in \N} (1 - \beta_{n}) = 0$.
  This contradicts
  $\lim_{n\rightarrow\infty}x_n=\bar{x}\neq\mu f$.
\end{proof}
\end{toappendix}

\Cref{th:fixed} and \Cref{co:RelaxedMannKleene} above instantiate to
our main application scenario, i.e., finding the optimal value
function of an MDP $M$ and imply that, given a (relaxed) Mann-Kleene
scheme $\mathcal{S}$, the iterations
$\mathcal{F}=(\mathcal{S},f_M,v_0)$ and
  $\mathcal{G}=(\mathcal{S},g_M,q_0)$ yields a converging sequence
$x_n^{\mathcal{F}}\to \opts{M}$ and
  $x_n^{\mathcal{G}}\to \optsa{M}$ for all initial values
$v_0,q_0$. In fact, the domain $X=\Rp$ is closed and the Bellman
operator $f_M$ is monotone, non-expansive and admit fixpoints
(since we work with MDPs with finite value).

\section{Approximated Fixpoints of Approximated Functions}
\label{se:approximation}

As discussed in the introduction, in many applications, it is not
possible to iterate with the target function as we did in the previous
section, since this function is not known exactly, but it can only be
approximated.
For example, the internal dynamics of an unknown system might be only
derivable through the interaction with the system. As a relevant
example, we outlined the \emph{model-based reinforcement learning}
approach where the agent derives an approximation of the system
through its experiences, meaning that it learns the dynamics of the
MDP $M$ that is assumed to underlie the system's behaviour, and
creates a sequence of successively more accurate approximated models
$M_n$.  It then learns the optimal value $\opts{M}$ or $\optsa{M}$ of
$M$ using the approximations $M_n$.  We will discuss this setting in
more detail in \Cref{se:mdp}.

In this section, we are, more generally, interested in the least
fixpoint of a function $f$ under the assumption that we can
get access only to a sequence $(f_n)$ of
approximations
converging to $f$.

\begin{assumption}
  \label{as:approximate-setting}
  Given a closed and convex set $X\subseteq\Rp^d$ with $0\in X$ where $d\in\N_0$ is a fixed arity,
  let $(f_n)$ be a sequence of monotone and non-expansive functions $f_n\colon X\to X$, pointwise
  converging to
   a function
  $f\colon X\to X$ with $\fix{f}\neq\emptyset$.
\end{assumption}

Note that, under \Cref{as:approximate-setting}, also $f$ is
guaranteed to be monotone and non-expansive.
Furthermore,  a standard result from analysis ensures that, given
a sequence $(f_n)$ of $L$-Lipschitz functions, $f_n:X\to X$, with
$X\subseteq\R^d$ compact, if the sequence converges to a function $f$,
then the convergence is uniform.
Therefore, if we consider the function sequence $(f_n)$ only
on a compact subset of $X$, we can assume w.l.o.g.\ that it converges
uniformly. Indeed, in the sequel we will often show that the sequences generated by Mann iterations are bounded and use this fact to restrict to a bounded and thus compact domain $X \subseteq \Rp^d$.

Now, given an iteration scheme $\mathcal{S}$ that is exact, i.e.,
which works when iterating the exact function, a naive idea to
construct a sequence of approximations $(x_n)$ of the least fixpoint
of the target function $f$ might be to perform \emph{resetting}, i.e., to
restart the iteration in each step and calculate $x_k$ as
$T_{n_k}^{\mathcal{S}}(f_k,T_{n_k-1}^{\mathcal{S}}(f_k,\dots
T_0^{\mathcal{S}}(f_k,x_0)\dots))$, namely approximate
$\mu f$ by iterating appropriately  many times with an approximation
$f_k$ sufficiently close to $f$.

However, as already
noted
in the introduction,
the least fixpoint operator $\mu$ is not continuous for non-expansive
functions, i.e., the sequence $\mu f_n$ might not converge
to $\mu f$ (see Example~\ref{it:simple-1}). %
Even worse, the approximations might not even have a fixpoint (see
Example~\ref{it:no-fixpoints}). Hence, simply
choosing an approximation $f_n$ sufficiently close to $f$, and
iterating with it, will not
provide any
guarantee to get close
to $\mu f$.

In case we are able to estimate the distance between the
target function $f$ and its approximations $f_n$, it is still possible
to ensure convergence by carefully choosing the amount $n_k$ of times
iterated with function $k$: Assuming a Kleene iteration starting from
the bottom element $0$ and choosing $n_k\in\N$ such that
$n_k=\lfloor\|f_k-f\|^{-\frac{1}{1+r}}\rfloor$ for some
(fixed) $r>0$, we get by non-expansiveness:
\begin{align*}
  \|f^{n_k}_k(0) - \mu f\| &\leq \|\mu f - f^{n_k}(0)\| + \|f^{n_k}(0) - f^{n_k}_k(0)\|\\
  &\le \|\mu f - f^{n_k}(0)\| + n_k \cdot \|f - f_k\|\\
  &\le \|\mu f - f^{n_k}(0)\| + \|f - f_k\|^{\frac{r}{1+r}} \to 0
\end{align*}
for $k\to\infty$ if $(f_n)$ is uniformly converging.  Intuitively,
performing iteration in a controlled way, we can ensure to get close
to $\mu f$.

\begin{example}
  Consider again Example~\ref{it:simple-1} discussed in the introduction,
  i.e., $f,f_n\colon [0,1]\to[0,1]$ where
  $f(x)=x,f_n(x)=(1-\nicefrac{1}{n})\cdot x+\nicefrac{1}{n},n\in\N$.  Here,
  we have $\mu f_n=1\neq 0=\mu f$. Choosing $x_n=f_n^n(0)$ results in
  the sequence
  \[x_n=1-(\frac{n-1}{n})^n\to 1-\frac{1}{e}.\]
  However, we have $\|f_n-f\|=\nicefrac{1}{n}$ whence, choosing $n_k=\lfloor\sqrt{k}\rfloor$ (for
  $r=1$), we indeed get
  \[x_n=1-(\frac{n-1}{n})^{\lfloor\sqrt{n}\rfloor}\to 0.\]
\end{example}

While the naive approach illustrated above works, it has some obvious
drawbacks, the most prominent being that we have to
restart the computation at each iteration step.  That means, to
compute the $k$-th approximation, we have to compute $n_k$ \emph{new}
iterates of the function $f_k$ without being able to reuse the result
of the previous step.  In addition to a worse performance, this makes
the iteration very unstable as it only relies on the current
approximation $f_n$, making it prone to measurement errors and
perturbations since progress from previous good approximations cannot be
reused. Furthermore, the approach could be simply not viable as it requires
to be able to estimate how close the approximations $f_n$ are to $f$ in order to determine suitable indices $n_k$. 

In the rest of this section we show that, under suitable assumptions,
dampened Mann schemes, which instead avoid restarting and properly
reuse previous iterates, work in the approximated setting given by
\Cref{as:approximate-setting}.  Here,
Condition~\eqref{eq:beta-condition-2} on the dampening parameters
$(\beta_n)$ is essential -- even when the iteration starts from $0$ (bottom value) -- as
the approximations $f_n$ constantly introduce errors, possibly causing
over-approximations which need to be \enquote{dampened}.

We first show that the least fixpoint $\mu f$ of the target function
serves as an asymptotic lower bound for all relaxed Mann-Kleene
schemes. %
\begin{lemmarep}
  \label{lm:lower}
  Under \Cref{as:approximate-setting}, given a relaxed Mann-Kleene scheme $\mathcal{S}$ and arbitrary $x_0\in X$,
  consider the iteration $\mathcal{F}=(\mathcal{S},(f_n),x_0)$.
  If $(x_n^{\mathcal{F}})$ is bounded, we have $\liminf_{n\to\infty}x_n^{\mathcal{F}}\geq\mu f$.
\end{lemmarep}

\begin{proof}
  For the sake of readability, we write $x_n\coloneqq x_n^{\mathcal{F}}$.
  Since $(x_n)$ is bounded, we can w.l.o.g.\ assume that $(f_n)$ is uniformly converging.

  We define $M\coloneqq\sup_{n\in\N_0}\|x_n-\mu f\|+\|\mu f\|$ ($<\infty$ by assumption).
  Furthermore, we define $\alpha_n'\coloneqq\frac{\alpha_n-\alpha}{1-\alpha}$ (with $\alpha=\lim_{n\to\infty}\alpha_n<1$)
  and $f_n'=(1-\alpha)f_n+\alpha$ so that
  \[x_{n+1}=(1-\beta_n)(\alpha_n x_n+(1-\alpha_n)f_n(x_n))=(1-\beta_n)(\alpha_n' x_n+(1-\alpha_n')f_n'(x_n)).\]
  Note that $\alpha_n'\in(-\infty,1],\alpha_n'\to 0,$ and all $f_n'$ are monotone, non-expansive,
  and converge to $f'=(1-\alpha)f+\alpha$.

  We write $s_n=\|(f')^n(0)-\mu f\|$ and, for $k,m\in\N_0$,
  \[t_m^k\coloneqq\sum_{i=m}^{m+k-1}(\prod_{\ell=i+1}^{m+k-1}|1-\alpha_\ell'|)(|1-\alpha_i'|\|f_i'-f'\|+M(2|\alpha_i'|+\beta_i)).\]
  For $n\in\N_0$, we then define
  \[\delta_n\coloneqq\min\{s_k+t^k_m\mid n=m+k\}.\]
  We note that, by Kleene's theorem and $\mu f=\mu f'$, we have $s_k\to 0$ and, for each $k\in\N_0$,
  also $t_m^k\to0$ since $\alpha_i'\to 0,\|f_i'-f'\|\to0,$ and $\beta_i\to 0$.
  Thus, it is clear that also $\delta_n\to0$.

  Now, for $n+1=m+k$, we have
  \begin{align*}
    &\|x_{n+1}-(f')^k(x_m)\|\\
    &\leq(1-\beta_n)\|\alpha_n'
    (x_n-(f')^k(x_m))+(1-\alpha_n')(f_n'(x_n)-(f')^k(x_m))\| \\
    & \qquad +\beta_n\|(f')^k(x_m)\|\\
    &\leq|1-\alpha_n'|\|f_n'(x_n)-(f')^k(x_m)\|+|\alpha_n'|\|x_n-(f')^k(x_m)\|+\beta_n M\\
    &\leq|1-\alpha_n'|(\|x_n-f'^{k-1}(x_m)\|+\|f_n'-f'\|)+M(2|\alpha_n'|+\beta_n)\\
    &\leq|1-\alpha_n'|\|x_n-f'^{k-1}(x_m)\|+|1-\alpha_n'|\|f_n'-f'\|+M(2|\alpha_n'|+\beta_n)\\
    &\leq\sum_{i=m}^{n}(\prod_{\ell=i+1}^n|1-\alpha_\ell|)(|1-\alpha_i'|\|f_i'-f'\|+M(2|\alpha_i'|+\beta_i))\\
    &=t_m^k
  \end{align*}
  In particular, since $f'^k(x_m)\geq f'^k(0)$ ($f'$ is monotone), we have
  \[x_{m+k}\geq f'^k(0)-t_m^k\geq\mu f-\delta_{m+k}.\]
  \qed
\end{proof}

From this, the following corollary immediately follows.

\begin{corollary}
  Under \Cref{as:approximate-setting}, given a relaxed Mann-Kleene
  scheme $\mathcal{S}$ and arbitrary $x_0\in X$, we consider the
  iteration $\mathcal{F}=(\mathcal{S},(f_n),x_0)$.  Then, we have
  $\lim_{n\to\infty}x_n^{\mathcal{F}}=\mu f$ if and only if
  $\limsup_{n\to\infty}x_n^{\mathcal{F}}\leq\mu f$.
\end{corollary}
Moreover, restricting to the one-dimensional case, i.e., for functions $f_n : X \to X$ with $X \subseteq \Rp$, the condition $\mu f = \lim_{n\to\infty} \mu f_n$, ensures convergence for every (relaxed) Mann-Kleene scheme.

\begin{theoremrep}[Approximated Mann-Kleene - one dimension]
  \label{it:thm:mt-one-dim}
  Let $f, f_n : X \to X$ be as in \Cref{as:approximate-setting} with
  $d = 1$, let $\mathcal{S}$ be a (relaxed) Mann-Kleene scheme,
  $x_0\in X$, and consider the iteration
  $\mathcal{F}=(\mathcal{S},(f_n),x_0)$. If
  $\mu f = \lim_{n\to\infty} \mu f_n$ then $(x_n^{\mathcal{F}})$
  converges to $\mu f$.
\end{theoremrep}

\begin{proof}
  Assume that $(f_n)$ converges pointwise to $f$ with
  $\lim \mu f_n = \mu f$ in a one-dimensional problem. In particular,
  the sequence $(\mu f_n)$ is bounded by some constant $C$. Consider
  the monotone and non-expansive function
  $f_C\colon \R_{\ge0} \to \R_{\ge0}$ given by
  \[
    f_C(x) = \begin{cases}
      C & \text{if } x \le C \\
      x & \text{if } x \ge C
    \end{cases}
  \]
  Note that by monotonicity and non-expansiveness we have
  $f_n \le f_C$ for all $n$. Denote by $(y_n)$ the sequence obtained
  from the iteration $(\mathcal{S}, f_C, x)$. Then
  $x_n^\mathcal{F} \le y_n$ for all $n$ and since $(y_n)$ converges by
  exactness of $\mathcal{S}$, we can conclude that $(x_n^\mathcal{F})$
  is bounded. In fact, for any $C>\mu f$ the inequality $f_n \le f_C$
  does hold for all sufficiently large $n$ and therefore
  $\limsup_{n\to\infty} x_n^\mathcal{F} \le \mu f$.

  Since $(x_n^\mathcal{F})$ is bounded we may assume that $(f_n)$ converges to $f$ uniformly. This implies that $\liminf_{n\to\infty} x_n^\mathcal{F} \ge \mu f$ by \Cref{lm:lower}.
  \qed
\end{proof}

In the above result the hypothesis that $\lim_{n\to\infty}\mu f_n=\mu f$ is essential, as witnessed by the example below.

\begin{example}[Mann-Kleene iteration converging to the wrong value]
  \label{ex:running-example}
  Consider again \Cref{it:simple-1} from the introduction,
  with $f,f_n\colon [0,1]\to[0,1]$ where $f(x)=x$ and
  $f_n(x)=(1-\nicefrac{1}{n})\cdot x+\nicefrac{1}{n}$ for $n\in\N$. We
  choose the Mann-Kleene scheme $\alpha_n=0,\beta_n=\nicefrac{1}{n}$,
  and iterate starting with $x_1=0$. This generates a sequence
  $(x_n^{\mathcal{F}})$ such that
  $x_{n+1}^{\mathcal{F}}=\frac{n-1}{2n}$
  and thus
  $x_n^{\mathcal{F}}\to\nicefrac{1}{2}\neq 0=\mu f$.
\end{example}

\Cref{it:thm:mt-one-dim} does not generalise to the
multidimensional case.
As shown by the example below, there are Mann-Kleene schemes satisfying
$\lim_{n\to\infty}\mu f_n=\mu f$ where iterating with the
approximations does not even converge.

\begin{example}[Mann-Kleene iteration diverging]
  \label{ex:cex-flip-map}
  Let $f : [0,1]^2 \to [0,1]^2$ be the flip map given by
  $f(x, y) = (y, x)$
  and consider a sequence %
  of maps $f_n\colon [0,1]^2 \to [0,1]^2$ as follows:
  \begin{enumerate}
  \item If $n$ is even let $f_n = f$ and
  \item if $n$ is odd define $f_n$ by
    $f_n(x, y) = \big(y \ominus \epsilon_n, x \oplus \epsilon_n\big)$
    where $\epsilon_n\coloneqq \nicefrac{2}{n}$ and $\ominus,\oplus$
    stand for truncated subtraction/addition in the interval $[0,1]$.
  \end{enumerate}
  Then, the maps $f$ and $f_n$ for $n\in\N$ are monotone and non-expansive.
  Moreover, $f_n$ converges (uniformly) to $f$ and since $f_n(0, \epsilon_n) = (0, \epsilon_n)$
  we also have $\lim_{n\to\infty} \mu f_n = \mu f = (0,0)$.

  However, if we choose the Mann-Kleene scheme $\mathcal{S}=((\alpha_n),(\nicefrac{1}{n+1}))$, with $\alpha_n=0$ for all $n$,
  it is easy to check that the sequence $(x_n^{\mathcal{F}})$ generated by the iteration
  $\mathcal{F}=(\mathcal{S},(f_n),(0,1))$ (starting at index $1$) is given by
  \[x_n^{\mathcal{F}} =
    \begin{cases}
      (\frac{n-2}{n}, 0) & \text{if $n$ is odd}  \\
      (0, \frac{n-1}{n}) & \text{if $n$ is even}
    \end{cases}
  \] for $n>1$.  Thus, both $(1,0)$ and $(0,1)$ are
  cluster points of $(x_n^{\mathcal{F}})$.  In particular, the
  sequence $(x_n^{\mathcal{F}})$ is not even asymptotically regular, i.e.\ $\|x_n^{\mathcal{F}} - x_{n+1}^{\mathcal{F}}\|$ does not converge to $0$.
\end{example}

We provide a first positive result for the multidimensional case when
the limit function $f$ is a power contraction. %
Power contractions naturally arise in the study of
MDPs. Hence this result will play an important role in
\Cref{se:mdp}.

\begin{theoremrep}[Approximated Mann-Kleene - power contractions]
  \label{th:convergence-for-power-contractions}
  Let $f, f_n : X \to X$ be as in \Cref{as:approximate-setting} with
  $f$ a power contraction. Given a Mann-Kleene scheme $\mathcal{S}$
  and $x_0\in X$, consider the iteration
  $\mathcal{F} = (\mathcal{S},(f_n),x_0)$.
  
  Then, if $(x_n^{\mathcal{F}})$ is bounded, it converges to the
  unique fixpoint of $f$.
\end{theoremrep}

An inspection of the proof shows that the result above actually holds
even if the scheme $\mathcal{S}$ does not satisfy
Condition~\eqref{eq:beta-condition-2}.

\begin{toappendix}
For the proof we use the following result of Walter:

\begin{theorem}[\cite{b5b89b91-299b-36c3-88dd-77287491029b}]
  \label{th:conv-power}
  Let $(X, d)$ be a complete metric space, $f\colon X \to X$ be a
  Lipschitz continuous power contraction and $x^*$ be the unique
  fixpoint of $f$. Let $(y_n)$ be any sequence in $X$ and define
  $\epsilon_n := d(y_{n+1}, f(y_n))$, then
  \[
    \lim_{n\to\infty} y_n = x^* \iff \lim_{n\to\infty} \epsilon_n = 0
  \]
\end{theorem}

\begin{proof}[Proof of \Cref{th:convergence-for-power-contractions}]
Since $\alpha_n \to 0$ and $\beta_n \to 0$, the sequence $(T_n^{\mathcal{F}}(f_n,\cdot))$ converges
pointwise to $f$. Since by assumption $(x_n^\mathcal{F})$ lives on a bounded domain, we can assume without loss of generality that the convergence is uniform.
Hence (instantiating \Cref{th:conv-power} with $y_n = x_n^{\mathcal{F}}$) we get
\[\epsilon_n = \|x_{n+1}^{\mathcal{F}} - f(x_n^{\mathcal{F}})\| = \|T_n(f_n,x_n^{\mathcal{F}}) - f(x_n^{\mathcal{F}})\| \le \|T_n(f_n,\cdot) - f\| \to 0\]
Therefore $(x_n^{\mathcal{F}})$ converges to the fixpoint of $f$.
  \qed
\end{proof}

At the cost of restricting to proper contractions, we can prove a version of this result for other choices of parameter sequences (in particular, ones with $\alpha_n\to1$).

\begin{theoremrep}
  Under \Cref{as:approximate-setting} with a \emph{contraction} $f$, given a Mann scheme
  $\mathcal{S}$ fulfilling the additional conditions $\lim_{n\to\infty} \alpha_n = 1$,  $\sum_{n\in\N} (1-\alpha_n) = \infty$ and $\lim_{n\to\infty} \beta_n / (1-\alpha_n) = 0$ (but not necessarily Condition~\eqref{eq:beta-condition-2}), and $x_0\in X$, consider
the iteration $\mathcal{F}=(\mathcal{S},(f_n),x_0)$.

Then, if the sequence $(x_n^{\mathcal{F}})$ is bounded, it converges to the fixpoint of $f$.
\end{theoremrep}
\begin{proof}
  Let $q<1$ be a contraction factor for $f$.
  Since the sequence $(x_n^{\mathcal{F}})$ is bounded by hypothesis, we can assume that $f_n\to f$
  uniformly.
  By \Cref{th:conv-power}, it suffices to show that
  $\|x_{n+1}^{\mathcal{F}}-f(x_n^{\mathcal{F}})\|\to0$.
  We will show the stronger statement that the sequence $(s_n)$ with
  \[s_n\coloneqq\|x_{n+1}^{\mathcal{F}} - f(x_n^{\mathcal{F}})\| + q \|x_n^\mathcal{F} - x_{n+1}^\mathcal{F}\|\]
  converges to zero.

  To this end, note that
  \begin{align*}
  &\|x_{n+1}^{\mathcal{F}} - f(x_n^{\mathcal{F}})\|\\
  &=\|(1-\beta_n)(\alpha_n x_n^{\mathcal{F}} + (1-\alpha_n)f_n(x_n^\mathcal{F})) - f(x_n^{\mathcal{F}})\|\\
  &\le (1-\beta_n)\alpha_n\|f(x_n^\mathcal{F}) - x_n^\mathcal{F}\| + \beta_n\|f(x_n^\mathcal{F})\| + (1-\alpha_n)\|f - f_n\|
  \end{align*}
  and
  \begin{align*}
  &\|x_n^\mathcal{F} - x_{n+1}^\mathcal{F}\| \\
  &=\|(1-\beta_n)(\alpha_n x_n^{\mathcal{F}} + (1-\alpha_n)f_n(x_n^\mathcal{F})) - x_n^{\mathcal{F}}\|\\
  &\le (1-\beta_n)(1-\alpha_n) \|f(x_n^\mathcal{F}) - x_n^\mathcal{F}\| + \beta_n\|x_n^\mathcal{F}\| + (1-\alpha_n)\|f - f_n\|
  \end{align*}
  Therefore, we can choose a constant $C$ such that
  \begin{align*}
  s_n &= \|x_{n+1}^{\mathcal{F}} - f(x_n^{\mathcal{F}})\| + q \|x_n^\mathcal{F} - x_{n+1}^\mathcal{F}\|\\
  &\le (1-(1-q)(1-\alpha_n))\|f(x_n^\mathcal{F}) - x_n^\mathcal{F}\| +
  \beta_n C \\
  &\qquad + (1-\alpha_n)\|f - f_n\| \\
  &= (1-(1-q)(1-\alpha_n))\|f(x_n^\mathcal{F}) - x_n^\mathcal{F}\| \\
  &\qquad + (1-\alpha_n) (\beta_n/(1-\alpha_n)) C + (1-\alpha_n)\|f -
  f_n\| \\
  &\le (1-(1-q)(1-\alpha_n)) \big( \|f(x_{n-1}^\mathcal{F}) -
  x_n^\mathcal{F}\|  \\
  &\qquad + \|f(x_n^\mathcal{F}) - f(x_{n-1}^\mathcal{F})\|
  \big) + (1-\alpha_n) (\beta_n/(1-\alpha_n)) C + (1-\alpha_n)\|f - f_n\|\\
  &\le (1-(1-q)(1-\alpha_n)) \big( \|f(x_{n-1}^\mathcal{F}) -
  x_n^\mathcal{F}\| + q\|x_n^\mathcal{F} - x_{n-1}^\mathcal{F}\| \big)
  \\
  &\qquad + (1-\alpha_n) (\beta_n/(1-\alpha_n)) C + (1-\alpha_n)\|f - f_n\|\\
  &= (1-(1-q)(1-\alpha_n)) s_{n-1} + (1-\alpha_n) ((\beta_n/(1-\alpha_n)) C +\|f - f_n\|)
  \end{align*}
It now follows from Lemma 4 in \cite{KIM200551} that $s_n \to 0$.
  \qed
\end{proof}
\end{toappendix}

Next we come to one of our main results for general non-expansive
functions. It provides sufficient conditions on the approximation
sequence $(f_n)$, ensuring that an exact Mann scheme converges to the
least fixpoint when iterating with approximations.
Note that it applies, in particular, to (relaxed) Mann-Kleene
schemes which are known to be exact by \Cref{th:fixed} and
\Cref{co:RelaxedMannKleene}.

\begin{maintheoremrep}[Approximated Mann iteration]
  \label{thm:approximate-convergence}
  Let $f, f_n : X \to X$ be as in \Cref{as:approximate-setting}, let $\mathcal{S}$ be an exact
  Mann scheme, $x_0\in X$, and consider the iteration
  $\mathcal{F}=(\mathcal{S},(f_n),x_0)$. Then $(x_n^{\mathcal{F}})$ converges to $\mu f$ if either
  \begin{enumerate}
  \item \label{it:thm:mt-mon}the sequence $(f_n)$ is monotone and $\mu f = \lim_{n\to\infty} \mu f_n$ or
  \item \label{it:thm:mt-norm}the sequence $(f_n)$ converges normally to
    $f$, i.e.
    $\sum_{n\in\N_0} \|f - f_n\| < \infty$.
  \end{enumerate}  
\end{maintheoremrep}

\begin{proof}
Throughout the proof we denote by $(x_n)$ the sequence obtained from the iteration $(\mathcal{S}, f, x_0)$. By exactness of $\mathcal{S}$ we know that $(x_n)$ converges to $\mu f$.

  \begin{enumerate}
  \item Note that in the case where $(f_n)$ is increasing, that
    means $f_1 \le f_2 \dots \le f$, the condition
    $\mu f = \lim \mu f_n$ is satisfied automatically, by the
    Scott-continuity of the fixpoint operator $\mu$.

    Assume without loss of generality that $f_1 \ge f_2 \ge \dots \ge f$, then also
    $\mu f_1 \ge \mu f_2 \ge \dots \ge \mu f$. For a fixed $m$
    consider the sequence $(y_n)$ that is inductively defined by
    $y_0 = x_m^{\mathcal{F}}$ and
    \[y_{n+1} = \big( \alpha_n y_n + (1-\alpha_n) f_m(y_n) \big) \cdot
      (1 - \beta_n)\]

    It is easy to see by induction that $y_n \ge x_{m+n}^\app$ for all
    $n\in\N$. Also, by exactness of $\mathcal{S}$, the sequence $(y_n)$ converges to
    $\mu f_m$ and therefore every cluster point $p$ of $(x_n^{\mathcal{F}})$
    satisfies $p \le \mu f_m$.
    Since $m\in\N$ is arbitrary we obtain that
    $p \le \inf_m \mu f_m = \mu f$. Similarly, since
    $f \le f_n$ for all $n \in \N$, we have $x_n \le x_n^{\mathcal{F}}$ for all $n\in\N$.
    Since the sequence $(x_n)$ converges
    to $\mu f$ every cluster point $p$ of $(x_n^{\mathcal{F}})$ is
    above $\mu f$. It follows that $\mu f$ is the unique cluster point
    of the sequence $(x_n^{\mathcal{F}})$. The case where
    $f_1 \le f_2 \le \dots \le f$ is analogous.

    \item Assume that $(f_n)$ converges normally to $f$.
     In order to conclude that $(x_n^{\mathcal{F}})$
    converges to the same limit as $(x_n)$ we show that
    $\lim_{n \to \infty} \|x_n - x_n^{\mathcal{F}}\|=0$.
    To this aim, by induction we can show that
    \begin{equation}
      \label{eq:normal}
      \|x_n - x_n^{\mathcal{F}}\| \le \sum_{i=0}^{n-1} \prod_{j=i}^{n-1} (1 - \beta_j) \cdot (1-\alpha_i) \cdot \|f_i - f\|
    \end{equation}

    In fact, first note
    \[
      x_{n+1} - x_{n+1}^{\mathcal{F}} = \big( \alpha_n (x_n -x_n^{\mathcal{F}}) +
      (1-\alpha_n) (f(x_n) - f_n(x_n^{\mathcal{F}})\big) \cdot (1 - \beta_n)
    \]
    and thus
    \begin{align*}
      & \|x_{n+1} - x_{n+1}^{\mathcal{F}}\|\\
      & \quad \leq (\alpha_n \|x_n -x_n^{\mathcal{F}}\| + (1-\alpha_n) \|f(x_n) - f_n(x_n^{\mathcal{F}})\|) \cdot (1 - \beta_n)\\
      & \quad \leq (\alpha_n \|x_n -x_n^{\mathcal{F}}\| +
      (1-\alpha_n) (\|f(x_n) - f(x_n^{\mathcal{F}})\| \\
      & \qquad + \|f(x_n^{\mathcal{F}}) - f_n(x_n^{\mathcal{F}})\|))
      \cdot (1 - \beta_n) \\
      &  \quad \leq  (\alpha_n \|x_n -x_n^{\mathcal{F}}\| + (1-\alpha_n) (\|x_n - x_n^{\mathcal{F}}\| + \|f - f_n\|)) \cdot (1 - \beta_n)\\
      &  \quad = (\|x_n -x_n^{\mathcal{F}}\| + (1-\alpha_n) \|f - f_n\|)) \cdot (1 - \beta_n)\\
      &  \quad \leq (\sum_{i=0}^{n-1} \prod_{j=i}^{n-1} (1 - \beta_j) \cdot (1-\alpha_i)\|f_i - f\|+ (1-\alpha_n) \|f - f_n\|)) \cdot (1 - \beta_n)\\
      &  \quad =  (\sum_{i=0}^{n-1} \prod_{j=i}^{n} (1 - \beta_j) \cdot (1-\alpha_i)\|f_i - f\|) + (1-\alpha_n) (1 - \beta_n) \|f - f_n\|\\
      & \quad = \sum_{i=0}^{n} \prod_{j=i}^{n} (1 - \beta_j) \cdot
      (1-\alpha_i)\|f_i - f\|
    \end{align*}
    as desired.

    Now observe that the right hand side of \eqref{eq:normal}
    converges to $0$. Let $\epsilon > 0$ and choose $N\in\N$ such
    that $\sum_{i=N}^\infty \|f_i - f\| \le \epsilon/2$. Now choose
    $M\in\N$ such that
    $\prod_{j=N}^M (1 - \beta_j) \le \epsilon/(2N\max_{n\in\N_0}\|f_n-f\|)$. Then for all
    $n \ge M$ we have
    \begin{align*}
      & \sum_{i=0}^n \prod_{j=i}^n (1 - \beta_j) \cdot \|f_i - f\| \\
      & \quad \le \sum_{i=0}^N \prod_{j=i}^n (1 - \beta_j) \cdot \|f_i - f\| + \sum_{i=N+1}^\infty \prod_{j=i}^n (1 - \beta_j) \cdot \|f_i - f\|\\
      & \quad \le N\max_{n\in\N_0}\|f_n-f\| \cdot \prod_{j=N}^n (1 - \beta_j) + \sum_{i=N+1}^\infty \|f_i - f\|\\
      &\quad \le \epsilon/2 + \epsilon/2 = \epsilon
    \end{align*}

\qedhere

  \end{enumerate}
  \qed
\end{proof}

It is remarkable that under the assumption of normal convergence of
$(f_n)$ to $f$, the sequence $(x_n^\mathcal{F})$ converges to the
correct solution even if $\lim_{n\to\infty} \mu f_n \neq \mu f$.
Combining with estimates obtained in \Cref{th:fixed} one can even
calculate error bounds $\epsilon_n$ (converging to $0$) such that
$\|x_n^\mathcal{F} - \mu f\| \le \epsilon_n$ for all $n\in\N_0$,
provided that we can calculate such estimates for
$\|f^n(0) - \mu f\|$.

\begin{example}
  \label{ex:running-continue}
  Consider again the sequence of functions $(f_n)$ from
  \Cref{ex:running-example}.  We observed that the Mann-Kleene
  iteration does not converge to the least fixpoint in this
  case. Indeed, $(f_n)$ does not converge normally as
  $\|f_n-f\|=\nicefrac{1}{n}$.

  However, let us consider a sped up sequence of approximations as to
  guarantee normal convergence, e.g., put $g_n = f_{n^2}$ and thus
  $\|g_n-f\|=\nicefrac{1}{n^2}$.  If $(x_n^{\mathcal{G}})$ is the
  sequence generated from $(g_n)$ by
  \Cref{thm:approximate-convergence} we get
  $\lim_{n\to\infty} x_n^{\mathcal{G}} = \mu f =0$.
\end{example}

Normal convergence
intuitively expresses the fact that $(f_n)$ converges sufficiently
fast to $f$.  In the context of a sampling-based approximation such as
model-based reinforcement learning, this can be obtained by taking a
sufficiently large number of samples. As discussed in
\Cref{se:mann-error-algo}, this number can be bounded using
quantitative versions of the law of large numbers (Bernstein's inequality).

\section{Value Functions of MDPs and Probabilistic Systems}
\label{se:mdp}

\subsection{MDP Sampling}
\label{sec:mdp-sampling}

We now come to our main application: approximating the
optimal value functions of MDPs in a model-based reinforcement
learning setting.  We show that dampened Mann iteration on the Bellman
operators of sampled MDPs always almost surely converges to the
optimal value function for any Mann-Kleene scheme, even though
  the requirements on the corresponding function sequence derived
  above are not necessarily met.

For an MDP $M =(S,A,T,R)$ we assume that the states, actions, and reward function are
known to the agent, and that the transition function is learned through sampling in the
following sense:
For every state $s\in S$ and every action $a\in A(s)$, there is a family of independent discrete
random variables $(X^{s,a}_n)_{n\in\N}$ with range $S$ defined on some probability space
$(\Omega, \mathcal{A}, \P)$ whose distributions are given by $\P[X^{s,a}_n = s'] = \Tr[]{s}{a}{s'}$.

We let $T_n$ be the maximum likelihood approximation to $T$ after performing $n$ random experiments
on each state-action pair, that means
\[
  \Tr[_n]{s}{a}{s'} \coloneqq \frac{|\{i \in \{1,\dots,n\} \mid X^{s, a}_i = s'\}|}{n}.
\]
Note that, in particular, we never over-estimate a probability $0$, i.e.
\begin{equation}\label{eq:sampling-property}
  \Tr[]{s}{a}{s'}=0\Rightarrow \Tr[_n]{s}{a}{s'}=0\qquad\text{for all }s,s'\in S,a\in A(s).
\end{equation}

The $n$-th approximation of $M$ is  given by
$M_n = (S, A, T_n, R)$ with $f_n$ and $g_n$ being the
(state and state-action)
value iteration functions of $M_n$.  By the
law of large numbers, we then infer
that $\P\big[ \lim_{n\to\infty} \Tr[_n]{s}{ a}{s'} = \Tr[]{s}{ a}{s'} \big] = 1$ for all $s,s'\in S,a\in A(s)$.
Hence the sequences $(f_n)$ and $(g_n)$
converge to true iteration functions $f_M$ and $g_M$ of $M$ almost
surely.

\begin{assumption}
  \label{as:mdp-setting}
  Let $M=(S,A,T,R)$ be an MDP with finite value, let $(M_n)$ be given
  by a sampling process as described above with $T_n\to T$.  We denote by
  $f_n\colon \Rp^S\to\Rp^S$ and $f\colon \Rp^S\to\Rp^S$ as well as
  $g_n\colon \Rp^{S\times A}\to\Rp^{S\times A}$ and
  $g\colon \Rp^{S\times A}\to\Rp^{S\times A}$ the
  Bellman operators of $M_n$ and $M$, respectively
  (see~\eqref{eq:bellman-ex}).
\end{assumption}

\begin{toappendix}
  \label{app:mdp}
  Since we are we are working in the \emph{non-discounted} setting, in
  order to ensure well-definedness of the least fixpoint of the
  Bellman iteration functions $f$ and $g$ of a MDP, which
  correspond to the optimal state- or state-action-values of $M$,
  respectively, we need to restrict the reward
  functions by asking that when staying inside a MEC the reward is
  zero.

  \begin{definition}[Proper reward]
    \label{de:proper-reward}
    Let $M = (S,A,T)$ be an MDP and let $R$ be a corresponding reward
    function. We say that $R$ is a \emph{proper reward} for $M$ if for all
    $(S',A')\in\mathrm{MEC}(M)$ and $s,s'\in S'$, $a\in A'(s)$ it holds
    $R(s,a,s')=0$.
  \end{definition}
\end{toappendix}

Note that \Cref{as:mdp-setting} implies \Cref{as:approximate-setting}.
From now on we will concentrate mostly on the state-value Bellman operator
$f_M$ since the result for the state-action-value operator will be derivable
as a corollary.
We can show that iterations based on a Mann scheme using MDP functions stay bounded. %

\begin{lemmarep}
  \label{le:boundedness}
  Under \Cref{as:mdp-setting}, given a Mann scheme
  $\mathcal{S}$ and an initial point $v\in\Rp^S$,
  consider the iteration
  $\mathcal{F}=(\mathcal{S},(f_n),v)$.
  Then, the sequence $(x_n^{\mathcal{F}})$
  is bounded.
\end{lemmarep}
\begin{proof}
Let us consider the iteration $\mathcal{F}\sqcup\mathrm{id} = ((f_n \sqcup \mathrm{id}), (\alpha_n), (\beta_n), v)$ instead.

For all states $s, t$ and actions $a$ with $0 < \Tr[]{s}{ a}{t } < 1$ fix numbers
\[0 < \underline{T}^{\app}(s, t) < \Tr[]{s}{a}{t} < \overline{T}^{\app}(s, t) < 1\]
Note that eventually
\begin{equation}\label{Corollary:SampledMDPs-Eq:Bounds}
\underline{T}^{a}(s, t) < \Tr[_n]{s}{ a}{t} < \overline{T}^{a}(s, t) \text{ for all such } s,t,a
\end{equation}
Now for a state $s$ of $M$ and each action $a\in A(s)$ let $\R_a(s) := \{t_1,\dots,t_k\}$ be all the
states that can be reached from $s$ via $a$ with positive probability.
For all $t_i\in\R_a(s)$ we add a new action $b = b(s, a, t_i)$ to $M$ with $\R_b(s) = \R_a(s)$ and
\[\Tr[]{s}{b}{t_j} = \begin{cases} \underline{T}^{a}(s, t_j) & \text{if } j \neq i \\ 1 - \sum_{\gamma\neq i} \underline{T}^{a}(s, t_\gamma) &\text{if } j=i \end{cases}\]
Let $N$ be the MDP which arises from $M$ by adding all these new actions and setting all positive
rewards to some large $R$.
Note that $N$ still has the same MECs as $M$, therefore its value function $g$ has a fixpoint
and so does $g \sqcup \mathrm{id}$.
In particular the Mann iteration with $g \sqcup \mathrm{id}$ in every step stays bounded from every
starting point.
Thus, is remains to prove that almost surely we eventually have $f_n \sqcup \mathrm{id} \le g \sqcup \mathrm{id}$.

It suffices to show that this is implied by \eqref{Corollary:SampledMDPs-Eq:Bounds}.
Fix $v\in\Rp^d$ and $s\in S$, clearly $\mathrm{id}(v)(s) \le (g \sqcup \mathrm{id})(v)(s)$. It remains to bound $f_n(v)(s)$, let $a\in A(s)$ be the action where $f_n(v)(s)$ is maximized, then
\[
f_n(v)(s) \le R + \sum_{t\in S} \Tr[_n]{s}{ a}{t} \cdot v(t) = R + \sum_{t\in\R_a(s)} \Tr[_n]{s}{ a}{t } \cdot v(t)
\]
Let $\R_a(s) = \{t_1,\dots,t_k\}$ with $v(t_1) \le \dots \le v(t_k)$, then successively shifting probability mass to the maximal value, we see that
\begin{align*}
f_n(v)(s) &\le \sum_{i=1}^{k-1} R + \underline{T}^{\app}(s, t_i) \cdot v(t_i) + \Big( 1 - \sum_{i=1}^{k-1} \underline{T}^{\app}(s, t_i) \Big) \cdot v(t_k)\\
&= R + \sum_{t\in S} \Tr[]{s}{b}{t} \cdot v(t) \le g(v)(s) \le (g \sqcup \mathrm{id})(v)(s)
\end{align*}
with $b = b(s, a, t_k)$.
  \qed
\end{proof}
 
Still, \Cref{thm:approximate-convergence} does not apply directly since approximations are obtained by a generic sampling process and thus convergence is neither monotone nor normal, in general. Still, we can conclude using a combination of the results about convergence in the ``exact case'' (\Cref{th:fixed}) and for power contractions (\Cref{th:convergence-for-power-contractions}).
In order to state the result more abstractly and concisely, we give a name to Mann schemes satisfying the two mentioned properties.

\begin{definition}[Regular Mann scheme]
  \label{de:regular-scheme}
  A Mann scheme $\mathcal{S}$ is \emph{regular} if
  \begin{enumerate}
  \item it is exact and
  \item in the approximated setting (\Cref{as:approximate-setting}),
    the sequence $(x_n^{\mathcal{F}})$ generated by the iteration
    $\mathcal{F}=(\mathcal{S},(f_n),x_0)$ converges to $\mu f$ for all
    $x_0\in X$, provided that it is bounded and $f$ is a power
    contraction.
\end{enumerate}
\end{definition}

By \Cref{th:fixed} and \Cref{th:convergence-for-power-contractions},
all Mann-Kleene schemes are regular, while it is currently unknown
whether all relaxed Mann-Kleene schemes are regular.

\subsection{Simple MDPs}
\label{ss:simple}

In general, the Bellman operators of MDPs are not contractions. 
We show that if we restrict to a suitable subclass
MDPs, they are power contractions.
To single out such class of MDPs, we recall the notion of (maximal)
end-components
\cite{bk:principles-mc,da:verification-probabilistic-thesis}.

\begin{definition}[Maximal end-components, simple MDPs]
  \label{de:MEC}
  Let an MDP $M=(S,A,T,R)$ be given. A tuple $(S',A')$ identifying a
  subprocess
  $M'=(S',A',T',R')$ with $\Tr[']{s}{a}{s'}=\Tr[]{s}{a}{s'}$ for all
  $s,s'\in S',a\in A'(s)$
  and $R'=\restr{R}{S'\times A'\times S'}$ is
  called an \emph{end-component} if the following conditions hold:
  \begin{itemize}
  \item $F(M')=\emptyset$,
  \item for all $s\in S',s'\in S\backslash S',a\in A'(s)$, we have $\Tr[]{s}{a}{s'}=0$
   (the process is \emph{closed}),
  \item for all $s,s'\in S'$, there exists a path in $M'$
    from $s$ to $s'$.\footnote{
      Given two states $s, t \in S$, a path
      from $s$ to $t$ is a
      sequence $s_0=s,s_1,\dots,s_k=t\in S$ and $a_1,\dots,a_k\in A$
      with $a_j\in A(s_{j-1})$ and $\Tr[]{s_{j-1}}{a_j}{s_j}>0$ for all
      $1\leq j\leq k$.}
  \end{itemize}
  An end-component $(S',A')$ is called a \emph{maximal end-component} (MEC) iff
  there exists no larger end-component (with respect to (pointwise)
  inclusion). We will write $\mathrm{MEC}(M)$ for the set of all
  maximal end-components of the MDP $M$.

  We  call an MDP $M$ \emph{simple} if $\mathrm{MEC}(M)=\emptyset$.
\end{definition}

Intuitively, an end-component is a strongly connected subset $S'$ of
non-final states such that each state has an action giving probability zero of leaving $S'$.

We now show convergence of regular Mann schemes for simple MDPs.
We proceed by introducing a class of functions which we will refer to as
\like~functions, and proving convergence of a regular Mann scheme for
functions in this class. Then the result for simple MDPs follows by
observing that the corresponding Bellman operators are in the class.

\begin{definition}[Witness-bounded~functions]
  \label{de:MDP-like}
  A monotone function
  $u$ over $\Rp^S$ is
  called a \emph{\witness} if
  \begin{enumerate}
  \item
    \label{de:MDP-like:bound}
    for all $v, v' \in \Rp^S$ it holds
    $|u(v) - u(v')| \leq u(|v-v'|)$;

  \item
    \label{de:MDP-like:power}
    there are $k \in \N$, $0\leq c < 1$ such that for all
    $v \in \Rp^S$ we have
    $\| u^k(v) \| \leq c \| v \|$.

  \end{enumerate}

  The function
  $h\colon \Rp^S\rightarrow\Rp^S$ is
  called \emph{\like} if there is a
  \witness~$u\colon
  \Rp^S\rightarrow\Rp^S$ such that
  $|h(v) - h(v')| \leq u(|v-v'|)$ for all
  $v, v' \in \Rp^S$.
\end{definition}

Note that a \witness~$u$ is \like, as one can take $u$ itself as its
\witness. Also, we can show that \like~functions are power contractions and they admit
a least fixpoint.

\begin{propositionrep}[Fixpoints of \like~functions]
  \label{pr:MDP-fixpoint}
  We assume a \like~function $h\colon \Rp^S\rightarrow\Rp^S$ with
  \witness~$u\colon \Rp^S\rightarrow\Rp^S$, $k \in \N$, and $0\leq c < 1$
  such that $\| u^k(v) \| \leq c \| v \|$. Then
  \begin{enumerate}
  \item
    \label{pr:MDP-fixpoint:power}
    $h$ is a power contraction;

  \item
    \label{pr:MDP-fixpoint:fix}
    $h$ admits a least fixpoint $\mu h$ such that $\| \mu h \| \leq \frac{1}{1-c} \|h^k(0)\|$.
  \end{enumerate}
\end{propositionrep}

\begin{proof}
  \eqref{pr:MDP-fixpoint:power} We show that $h^k$ is a contraction
  with factor $c$.  For all $v, v' \in \Rp^S$,
  repeatedly applying property \eqref{de:MDP-like:bound} of \Cref{de:MDP-like} and using monotonicity
  of $u$ we obtain
  \begin{center}
    $| h^k(v) - h^k(v')| \leq u(|h^{k-1}(v)-h^{k-1}(v')|) \leq  \ldots
     \leq u^k(|v-v'|)$
   \end{center}
   hence
   \begin{center}
     $\|h^k(v) - h^k(v')\| \leq \|u^k(|v-v'|)\| \leq c \|\ |v-v'| \ \| = c \|v-v'\|$.
   \end{center}
   where the second inequality uses \eqref{de:MDP-like:power} of
   \Cref{de:MDP-like}.  \bigskip

   \eqref{pr:MDP-fixpoint:fix} First observe that, by the previous
   point, for all $i \in \N$ it holds
   \begin{equation}
     \label{eq:step}
    \|h^{k(i+1)}(0) - h^{ki}(0) \| \leq c^i \|h^k(0) - h^0(0)\| = c^i \|h^k(0)\|
  \end{equation}
  where we write  $0$ to denote the constant function $0$.

  This means that if we iterate $h$ over $0$, we obtain
  \begin{align*}
    \|h^{k(i+1)}(0)\|
    & = \| \sum_{j=0}^i ( h^{k(j+1)}(0) - h^{kj}(0)) \| \\
    & \leq  \sum_{j=0}^i \| h^{k(j+1)}(0) - h^{kj}(0) \| & \mbox{[by \eqref{eq:step}]}\\
    & \leq \sum_{j=0}^i c^j \|h^k(0)\| & \mbox{[since $c<1$ and $\sum_{j=0}^\infty c^j = \frac{1}{1-c}$]}\\
    & \leq \frac{1}{1-c}  \|h^k(0)\|
  \end{align*}

  Since all iterates are below a fixed bound
  $b = \frac{1}{1-c} \|h^k(0)\|$, we can restrict $h$ to the lattice
  $[0,b]^S$. By Knaster-Tarski the restriction has a least fixpoint
  $\mu h$ which is also the least fixpoint of the original
  function. Hence $h$ has a
  fixpoint $\mu h$ and $\| \mu h \| \leq \frac{1}{1-c} \|h^k(0)\|$.
  \qed
\end{proof}

\begin{toappendix}
We need a preliminary technical lemma.

\begin{lemmarep}[Ranking MDPs]
  \label{le:rank}
  Let $M=(S,A,T,R)$ be a simple MDP.
  Then there is a function $\rank\colon S\to\N_0$ such that for all $s \in S$,
  it holds $\rank[s] < |S|$, $\rank[s] =0$ implies $s \in F(M)$ and
  for all $a \in A(s)$ there exists $s' \in S$ with $\Tr[]{s}{a}{s'} >0$ and
  $\rank[s'] < \rank[s]$.
\end{lemmarep}

\begin{proof}
  The function $\rank$ can be defined inductively as follows.
  \begin{itemize}

  \item $\rank[s] = 0$ for $s \in F(M)$;
  \item Assume that the states with $\rank[s] \leq k$ have been
    already identified and let $S_k$ be the set of such
    states. Consider the set of states
    \begin{center}
      $S' = \{ s \in S \setminus S_k \mid \forall a \in A(s). \exists
      s' \in S_k.\ \Tr[]{s}{a}{s'} >0 \}$.
    \end{center}
    We have $S' \neq \emptyset$, otherwise, if for all
    $s \in S \setminus S_k$ there is $a \in A(s)$ such that for all
    $s'$ with $\Tr[]{s'}{a}{s'' } > 0$ it holds
    $s'' \in S \setminus S_k$, then $S \setminus S_k$ would contain an
    end-component, violating the hypotheses.
 
    Then let $\rank[s] = k+1$ for all all $s \in S'$.
  \end{itemize}

  Function $\rank$ satisfies the desired properties by
  construction. In particular, since at each step $S' \neq \emptyset$,
  there will be at most $|S|$ iterations, hence $\rank[s] \leq |S|$.
  \qed
\end{proof}
\end{toappendix}

The Bellman operators of a simple MDP are \like, which will shed more
intuition on the concept. First, such functions are power contractions
(cf.~\Cref{de:MDP-like}\eqref{de:MDP-like:power}) and second, observe
that expectation and maximum are the central operators of the Bellman
equations and they satisfy
\Cref{de:MDP-like}\eqref{de:MDP-like:bound}), unlike, for instance,
the minimum. In fact the witness $u$ for an MDP function is based on
an MDP with the same parameters, but setting all rewards to $0$.

The fact that the Bellman operators of simple MDPs are power
contractions is probably folklore (see e.g. \cite{EZ62:OPSDS}). Here
we will prove a more general statement.

\begin{propositionrep}[Simple MDP functions are \like]
  \label{pr:MDP-mdp-like}
  Let $M=(S,A,T,R)$ be a simple MDP.  Then, its Bellman operator $f_M$
  is \like.
\end{propositionrep}

\begin{proof}
  Let us prove that $f_M$ is \like
  with \witness~$u\colon \Rp^S \to \Rp^S$ defined by
  \begin{center}
    $u(v)(s) = \max_{a \in A(s)} \sum_{s'\in S}\Tr[]{s}{a}{s'} v(s')$.
  \end{center}
  In fact, for
  all $v, v' \in \Rp^S$ we have,
  for all $s \in S$:
  \begin{align*}
    &|f(v)(s) - f(v')(s)| =\\
    & \quad = | \max_{a \in A(s)} \sum_{s'\in S}\Tr[]{s}{a}{s'}(R(s,a,s')+v(s')) -\\
    & \qquad \max_{a \in A(s)} \sum_{s'\in S}\Tr[]{s}{a}{s'}(R(s,a,s')+v'(s'))|\\
    & \quad \leq \max_{a \in A(s)} |\sum_{s'\in S}\Tr[]{s}{a}{s'}(R(s,a,s')+v(s')) -\\
    & \qquad \qquad \quad \sum_{s'\in S}\Tr[]{s}{a}{s'}(R(s,a,s')+v'(s'))|\\
    & \quad \leq |\max_{a \in A(s)} \sum_{s'\in S} \Tr[]{s}{a}{s'} (v(s') - v'(s'))\\
    & \quad \leq \max_{a \in A(s)} \sum_{s'\in S} \Tr[]{s}{a}{s'} |v(s') - v'(s')|\\
    & \quad = u(|v-v'|)(s)
  \end{align*}
  Thus, $|f(v) - f(v')| \leq u(|v-v'|)$ as desired.

  The function $u$ is a \witness, in fact:

  \begin{enumerate}

  \item for all $v, v' \in \Rp^S$ we have
    $|f(v) - f(v')| \leq u(|v-v'|)$.
    The proof is analogous to the one above.

  \item In order to show that $u$ is a power contraction, we consider
    the function $\rank[s]$ from \Cref{le:rank} and prove, by
    induction that for all $s \in S$ it holds that
    $u^{\rank[s]+1}(v)(s) \leq c_s \|v\|$ for a suitable constant
    $c_s$, with $0 \leq c_s < 1$.

    \begin{itemize}
    \item ($\rank[s]=0$) In this case $s \in F(M)$ and thus
      $u^{0+1}(v)(s) = u(v)(s)=0 \leq 0 \cdot \|v\|$, as desired.

    \item ($\rank[s] >0$) First observe that, given any $a
     \in A(s)$, by \Cref{le:rank}, there is is a state $s_a' \in S$ such that
      $\Tr[]{s}{a}{s_a' } >0$ and $\rank[s'] < \rank[s]$. Hence
      \begin{align*}
        & u^{\rank[s]}(v)(s) = &\\
        & = \max_{a \in A(s)} \sum_{s''\in S}\MTr[]{s}{s''} u^{\rank[s]-1}(v)(s'')\\
        & = \max_{a \in A(s)}  \MTr[]{s}{s_a' } u^{\rank[s]-1}(v)(s'_a) + \sum_{s'' \neq s'_a} \MTr[]{s}{s''} u^{\rank[s]-1}(v)(s'')\\
        & \leq \max_{a \in A(s)}  \MTr[]{s}{s_a' } c_{s'_a} \|v\|+  \sum_{s'' \neq s'_a} \MTr[]{s}{s''} \|v\|\\
        & = \max_{a \in A(s)}  (\MTr[]{s}{s'_a } c_{s'_a} +  \sum_{s'' \neq s'_a} \MTr[]{s}{s''}) \|v\|\\
        & = c_{s} \cdot  \|v\|
      \end{align*}
      where
      $c_s = \max_{a \in A(s)} (\MTr[]{s}{s'_a } c_{s'_a} + \sum_{s''
        \neq s'_a} \MTr[]{s}{s''})<1$.

      Given the above, we conclude that
      $u^{|S|}(v)(s) \leq c_s \cdot \|v\|$ for all $s \in S$ and thus
      $\|u^{|S|}(v)\| \leq c \cdot \|v\|$ where
      $c = \max_{s\in S} c_s < 1$, as desired.
    \end{itemize}

  \end{enumerate}
  \qed
\end{proof}

In particular, we can use \Cref{th:convergence-for-power-contractions} to get convergence
of  Mann iteration.

\begin{theorem}[Convergence for simple MDPs]
  Let $M, M_n$ be MDPs as in \Cref{as:mdp-setting}. Given a regular
  Mann scheme $\mathcal{S}$ and an initial point $v\in\Rp^S$,
  consider the iteration
  $\mathcal{F}=(\mathcal{S},(f_n),v)$.
  If $M$ is simple, then $x_n^{\mathcal{F}}\to \opts{M}$.
\end{theorem}

\subsection{General MDPs}
\label{ss:general-mdps}

When MDPs have end-components
\cite{bk:principles-mc,bccfkkpu:mdps-learning,hm:interval-iteration-mdps},
the iteration functions are no longer power contractions, which implies that they might not admit fixpoints. 
As anticipated in \Cref{se:MDP-basics}, we will work under the assumption that all MDPs have
finite value, i.e., the Bellman operator has a (finite) fixpoint and thus a least fixpoint.  We observe here
that MDPs with finite value can be naturally characterized as MDPs
where no reward is given inside end-components.
As observed later in \Cref{ss:non-discounted} this setting
properly
generalises MDPs with discounted return.

\begin{lemmarep}[Characterising MDPs with finite value]
  An MDP $M$ has a finite value, i.e., its Bellman operator $f_M$ has a (least) fixpoint, if and
  only if for all maximal end-components $(S',A')$ and all $s\in S',a\in A'(s),$ and $s'\in S'$ with
  $\Tr[]{s}{a}{s'}>0$, we have $R(s,a,s')=0$.
\end{lemmarep}
\begin{proof}
  By existence of an optimal stationary policy, the value $\opts{M}$
  is finite (i.e., $f_M$ has a fixpoint) if and only if
  $v^\pi = v^\pi_M$ is finite for all policies $\pi\in\Pi(M)$ (i.e.,
  $f_M^\pi$ has a fixpoint).

  Now, fix a policy $\pi\in\Pi(M)$.
  Note that end-components of $M^\pi$ are also end-components of $M$.

  Assume that $R^\pi(s,s')=0$ for all $s,s'\in S'$ with $\MTr[^\pi]{s}{s'}>0$ in every end-component
  given by $S'$ of $M^\pi$.\footnote{Here, we abbreviate $R^\pi(s,s')=R(s,\pi(s),s')$ and $\MTr[^\pi]{s}{s'}=\Tr[]{s}{\pi(s)}{s'}$.}

  Then, we have for all $v\in\Rp^S$ with $\restr{v}{S'}=0$ that
  \begin{align*}
    f_M^\pi(v)(s)&=\sum_{s'\in S}\MTr[^\pi]{s}{s'}(R^\pi(s,s')+v(s'))\\
    &=\sum_{s'\in S'}\MTr[^\pi]{s}{s'}R^\pi(s,s')\\
    &=0
  \end{align*}
  for all $s\in S'$ whence
  \[\restr{v^\pi}{S'}=\lim_{n\to\infty}\restr{(f_M^\pi)^n(0)}{S'}=0.\]
  On the other hand, let $s\in S$ be such that $v^\pi(s)>0$.
  By the above, we have that $s\notin S'$ for all end-components $S'$.\footnote{For the sake of
  readability, we assume that $F=\emptyset$ (the argument is the same with final states treated
  analogously to end-component states).}
  Then, there exists a path of length $|S|$ to a state $s'\in S'$ for some end-component $S'$ of
  $M^\pi$ as otherwise, $s$ together with the set of states reachable from $s$ would form an end
  component.
  Let us denote the probability of this path by $p(s)$.
  But then
  \begin{align*}
    \max_{s\in S}v^{\pi}(s)&\leq\max_{s\in S,v^\pi(s)>0}|S|\cdot R+\sum_{s'\in S}\P^{\pi}(\mathbf{s}_n=s'\mid \mathbf{s}_0=s)v^{\pi}(s')\\
    &\leq |S|\cdot R+(1-\min_{s\in S,v^\pi(s)>0}p(s))\max_{s'\in S}v^{\pi}(s')\\
    &\leq |S|\cdot R / \min_{s\in S,v^\pi(s)>0}p(s)<\infty.
  \end{align*}

  Thus, if $R(s,a,s')=0$ for all $s,s'\in S',a\in A'(s)$ with $\Tr[]{s}{a}{s'}>0$ in end-components
  $(S',A')$ of $M$, whence $R^\pi(s,s')=0$ for all $s,s'\in S'$ with $\MTr[^\pi]{s}{s'}>0$ in end
  components $S'$ of $M^\pi$ for all $\pi\in\Pi(M)$, we have $v^\pi$ is finite for all
  $\pi\in\Pi(M)$, whence also $\opts{M}$ is finite.

  Now assume that $r=R^\pi(s,s')>0$ for some $s,s'\in S'$ with $\MTr[^\pi]{s}{s'}>0$ in some end
  component of $M^\pi$ given by $S'$.
  Then, we have
  \begin{align*}
    v^\pi(s)&=\E^\pi[\sum_{n\in\N_0}\mathbf{r}_n\mid \mathbf{s}_0=s]\\
    &\geq r\sum_{n\in\N_0}\P^\pi(\mathbf{s}_n=s,\mathbf{s}_{n+1}=s'\mid\mathbf{s}_0=s)\\
    &=r\MTr[^\pi]{s}{s'}\sum_{n\in\N_0}\P^\pi(\mathbf{s}_n=s\mid\mathbf{s}_0=s)\\
    &=\infty
  \end{align*}
  since the state $s$ is essential (see, e.g.,~\cite{f:markov-chains}).

  Thus, if $\opts{M}$ is finite, whence all $v^\pi$ are finite, also $R^\pi(s,s')=0$ for all $s,s'\in S'$
  with $\MTr[^\pi]{s}{s'}>0$ in end-components $S'$ of $M^\pi$ and $\pi\in\Pi(M)$.
  By definition of end-components, this gives $R(s,a,s')=0$ for all $s,s'\in S',a\in A'(s)$ with
  $\Tr[]{s}{a}{s'}>0$ in end-components $(S',A')$ of $M$.
  \qed
\end{proof}

Still, the existence of multiple fixpoints, with the possibility of
getting stuck at a fixpoint that is not the least, brings additional
complications. An idea explored in the literature
\cite{bccfkkpu:mdps-learning,hm:interval-iteration-mdps} is to work on
an MDP obtained from the original one by quotienting each MEC into a
single node.  In order to deal with this quotient we need some
notation.  Given a function $\merge\colon S \to Q$, we define:
\begin{itemize}
\item \emph{Reindexing}:  $\reindex{\merge}\colon \Rp^Q \to \Rp^S$  defined, for all $w \in \Rp^Q$, by $\reindex{\merge}(w) = w \circ \merge$
\item \emph{Maximum}: $\maxf{\merge}\colon \Rp^S \to \Rp^Q$ defined, for all $v \in \Rp^Q$ and $e \in Q$ by $\maxf{\merge}(v)(e) = \max \{ v(s) \mid s \in S \land \merge(s)=e\}$.
\end{itemize}

When dealing with quotients, it is convenient to work with
MDPs $M = (S,A,T,R)$ where the action sets for different states are
disjoint, i.e., for $s, s' \in S$ if $s \neq s'$ then
$A(s) \cap A(s') = \emptyset$.
It is clear that any MDP can be transformed into an equivalent MDP
with disjoint action sets simply by relabelling the actions, without
changing the Bellman functions, whence we will, from now on,
w.l.o.g., work under the disjointness assumption.

\begin{definition}[Quotient MDP]
  \label{def:quotient-mdp}
  Let $M = (S,A,T,R)$ be an MDP with finite value.  Let $\equiv$
  denote the equivalence identifying states in the same MEC.  Consider
  $\merge\colon S \to S/_{\equiv}$ defined by mapping each state to
  its equivalence class: $\merge(s) = [s]$.  The quotient MDP is
  $M_\merge = (S_\merge, A_\merge, T_\merge, R_\merge)$ where
  $S_\merge =S/\!\equiv$,
  $A_\merge([s]) = \bigcup_{s' \in [s]} A(s')$.  Moreover, for
  $s, s' \in S$, $a \in A(s)$:  
  {\small
  \begin{align*}
    \Tr[_\merge]{[s]}{a}{[s']})=\sum_{s''\in[s']}\Tr[]{s}{a}{s''}
    &
    \qquad
    R_\merge([s], a,[s']) = \frac{\sum_{s''\in [s']} \Tr[]{s}{a}{s'}
    R(s,a,s')}{\Tr[_\merge]{[s]}{a}{[s']}}.  
  \end{align*}
  }

  The \emph{reduced quotient} $\hat{M}_\merge$ is defined as the MDP
  derived from $M_\merge$ by removing self-loops, i.e., modifying the
  definition of the set of enabled actions as follows:
  \begin{center}
    $\hat{A}_\merge([s]) = A_\merge([s]) \setminus \{ a \in
    A_\merge([s])\mid \Tr[_\merge]{[s]}{a}{[s]}=1\}$.
  \end{center}
  and adapting $\hat{T}_\merge$ accordingly.
\end{definition}

The quotient is an MDP with only singleton end-components, i.e.,
states with actions producing self-loops.  Since the reduced quotient
is obtained from the quotient by removing self-loops, it clearly has
no end-components, allowing to reuse results from
\Cref{ss:simple}. 

\begin{toappendix}
  We need another technical lemma.

  \begin{lemma}[Relating MDPs and quotients]
    \label{le:MDP-lift}
    Let $M = (S,A,T,R)$ be an MDP with finite value, let $\equiv$ by
    the equivalence identifying states in the same MEC, with
    $\merge\colon S \to S/_{\equiv}$ sending each state to its
    equivalence class, and let
    $M_\merge = (S_\merge, A_\merge, T_\merge,R_\merge)$ be its
    quotient as in \Cref{def:quotient-mdp}. Let $f_r = f_{M_\merge}$
    and $\hat{f}_\merge = f_{\hat{M}_\merge}$,
    then $f_\merge= \maxf{\merge} \circ f \circ \reindex{\merge}$.
    Moreover,
    $\hat{f}_\merge \leq f_\merge \leq \hat{f}_\merge \sqcup id$.
  \end{lemma}

\begin{proof}
  Let $w \in \Rp^Q$.  Then for $e \in S/_{\equiv} = Q$ we have

  \begin{align*}
    & \maxf{\merge} \circ f \circ \reindex{\merge}(w)(e)  = \\
    & = \max_{s \in e} f(w \circ \merge)(s)\\
    & = \max_{s \in e} \max_{a \in A(s)} \sum_{s' \in S} \Tr[]{s}{a}{s'} (R(s,a,s')+ w(\merge(s')))\\
    & = \max_{s \in e} \max_{a \in A(s)} \sum_{e' \in S} (\sum_{s' \in e'} \Tr[]{s}{a}{s'} (R(s,a,s')+ w(e')))\\
    & = \max_{s \in e} \max_{a \in A(s)} \sum_{e' \in S} (\Tr[_\merge]{e}{a}{e'} w(e') + \sum_{s' \in e'} \Tr[]{s}{a}{s'} R(s,a,s'))\\
    & = \max_{s \in e} \max_{a \in A(s)} \sum_{e' \in S} (\Tr[_\merge]{e}{a}{e'} w(e') + \Tr[_\merge]{e}{a}{e'} R_\merge(e,a,e'))\\
    & = \max_{a \in A_\merge(e)} \sum_{e' \in S} (\Tr[_\merge]{e}{a}{e'} (w(e') + R_\merge(e,a,e'))\\
    & = f_\merge(w)(e)
  \end{align*}

  The second part,
  $\hat{f}_\merge \leq f_\merge \leq f_\merge \sqcup id$ is immediate
  from the fact that $M_\merge$ and $\hat{M}_\merge$ only differ in
  the enabled actions. The first inequality
  $\hat{f}_\merge \leq f_\merge$ is due to the fact that
  $\hat{M}_\merge$ has less enabled actions. The second
  $f_\merge \leq f_\merge \sqcup id$ derives from the fact that the
  only additional actions in $M_\merge$ are self-loops with
  probability $1$.
  \qed
\end{proof}
\end{toappendix}

The following proposition restates results from
\cite{bccfkkpu:mdps-learning,hm:interval-iteration-mdps} on collapsing
MECs.

\begin{propositionrep}[Fixpoints of non-discounted MDPs]
  \label{pr:MDP-has-fixpoint}
  Let $M$ be an MDP with finite value.  Let
  $\hat{M}_\merge$ be the reduced quotient
  (\Cref{def:quotient-mdp}).
  Then
  $\opts{M} = \reindex{\merge}(\opts{\hat{M}_\merge})$.
\end{propositionrep}

\begin{proof}
  First of all observe that, since $\hat{M}_\merge = (\hat{S}_\merge, \hat{A}_\merge, \hat{T}_\merge,
  \hat{R}_\merge)$ is a simple MDP, by
  \Cref{pr:MDP-mdp-like}, $\hat{f}_\merge$ is a \like~function
  and thus, by \Cref{pr:MDP-fixpoint}\eqref{pr:MDP-fixpoint:fix},
  it admits a least fixpoint $\mu \hat{f}_\merge$.

  By \Cref{le:MDP-lift}, $\hat{f}_\merge \leq f_\merge \leq \hat{f}_\merge \sqcup id$
  and thus $\mu \hat{f}_\merge \leq \mu f_\merge \leq \mu(\hat{f}_\merge \sqcup id)$. It is easy
  to see that $\mu (\hat{f}_\merge \sqcup id) = \mu \hat{f}_\merge$, since they produce
  the same sequence of Kleene iterates from $0$, and thus
  $\mu f_\merge = \mu \hat{f}_\merge$

  In order to conclude we show $\mu f = \reindex{\merge}(\mu f_\merge)$.
  First note that $\reindex{\merge}(\mu f_\merge)$ is a pre-fixpoint of $f$ and thus $\mu f \leq \reindex{\merge}(\mu f_\merge)$. In fact

  \begin{align*}
    & f(\reindex{\merge}(\mu f_\merge)) \leq\\
    & \leq \reindex{\merge}(\maxf{\merge}(f(\reindex{\merge}(\mu f_\merge))))\\
    &  = \reindex{\merge}(f_\merge(\mu f_\merge)) & \mbox{[by \Cref{le:MDP-lift}, $f_\merge = \maxf{\merge} \circ f \circ \reindex{\merge}$ ]}\\
    & = \reindex{\merge} (\mu f_\merge)
  \end{align*}

  For the converse inequality, first note that if
  $v \in \Rp^S$ is a fixpoint of $f$ then
  \begin{center}
    $v = \reindex{\merge}(\maxf{\merge}(v))$,
  \end{center}
  i.e., for all MECs $E$, for any two states
  $s, s' \in E$, $v(s)=v(s')$. In fact, assume that
  $v \neq \reindex{\merge}(\maxf{\merge}(v))$ and let $E$ be a MEC where $v$ has not the
  same value on all states. Let $\pi$ be a policy which ``cycles''
  in $E$, i.e., such that for all $s \in E$ it holds
  $\Tr[]{s}{\pi(s)}{s'} > 0$ implies $s' \in E$. Let $s \in E$ be
  such that $v(s) = \min \{ v(s') \mid s' \in E \}$ and
  $\Tr[]{s}{\pi(s)}{s'} > 0$ for some $s'$ such that
  $v(s) < v(s')$. A simple calculation shows that
  $f(v)(s) \geq f^\pi(v)(s) > v(s)$, hence $v$ cannot be a
  fixpoint for $f$.

  Using the above, we show that $\maxf{\merge}(\mu f)$ is a
  pre-fixpoint of $f_\merge$:
  \begin{align*}
    & f_\merge(\maxf{\merge}(\mu f)) =\\
    & = \maxf{\merge}(f(\reindex{\merge}(\maxf{\merge}(\mu f)))) & \mbox{[by \Cref{le:MDP-lift}, $f_\merge = \maxf{\merge} \circ f \circ \reindex{\merge}$ ]}\\
    &  = \maxf{\merge}(f(\mu f)) & \mbox{[since $\reindex{\merge}(\maxf{\merge}(\mu f))=\mu f$]}\\
    & = \maxf{\merge}(\mu f)
  \end{align*}
  Therefore $\mu f_\merge \leq \maxf{\merge}(\mu f)$ and thus
  \begin{center}
    $\reindex{\merge}(\mu f_\merge) \leq \mu \reindex{\merge}(\maxf{\merge}(\mu f))) = \mu f$,
  \end{center}
  as desired.
  \qed
\end{proof}

We show that for the successive approximation of (the value iteration
function of) an MDP obtained by sampling as described above, the least
fixpoints of the approximations always converge to the least fixpoint
of the underlying MDP.

\begin{toappendix}
  \begin{lemma}[\cite{MR0240802}]
  Let $(f_n)$ be a sequence of contractions $f_n\colon [0,1]^S \to [0,1]^S$ with respective fixpoints
  $a_n$ and $f\colon [0,1]^S \to [0,1]^S$ be a contraction with fixpoint $a$.
  If $(f_n)$ converges pointwise to $f$, then $(a_n)$ converges to $a$.
  \end{lemma}
\end{toappendix}

\begin{theoremrep}
  Under \Cref{as:mdp-setting}, we have
  $\lim_{n\to\infty} \mu f_n = \mu f$.
\end{theoremrep}
\begin{proof}
  Fix a policy $\pi\colon S \to A$ and consider the
  ``Markov chain'' $M^\pi$ together with its value iteration
  function $f^\pi\colon \Rp^S \to \Rp^S$.

  Let $S'$ be the set of states $s \in S$ with
  $\mu f^\pi(s) = 0$ and consider
  \[
    \hat f^\pi(v)(s) =
    \begin{cases}
      f^\pi(v)(s) & \mbox{if  $s\notin S'$}\\
      0 & \text{if } s\in S'
    \end{cases}
  \]
  Note that, if $\equiv$ is the function identifying states in the
  same MEC of $M^\pi$ and $\merge\colon S \to S/_{\equiv}$ is the
  corresponding quotient function then $\hat f^\pi$ is
  $\reindex{\merge} \circ \hat{f}^\pi$. Thus, by
  \Cref{pr:MDP-mdp-like}, it is a \emph{power contraction}
  and, by \Cref{pr:MDP-has-fixpoint}, its unique
  fixpoint is equal to $\mu f^\pi$.

  We can analogously consider for $n\in\N$ the Markov Chain
  $M_n^\pi$, the set $S_n'$ and the maps
  $f_n^\pi$ and $\hat f_n^\pi$. Now note that $(f_n)$ converges
  to $f$ almost surely and moreover if $\Tr[]{s}{a}{t } = 0$ for some
  $a\in A$ and $s,t\in S$ then also $\Tr[_n]{s}{a}{t}  = 0$ for all
  $n\in\N$. It follows that almost surely $S_n' = S'$
  for all $n$ above a fixed threshold $N$. But this implies that
  $\hat f_n^\pi$ converges to $\hat f^\pi$ and the same is
  true for the $k$-fold iterations. Thus, in this case we have
\[\lim_{n\to\infty} \mu f_n^\pi = \lim_{n\to\infty} \mu \hat f_n^\pi = \lim_{n\to\infty} \mu (\hat f_n^\pi)^{k} = \mu (\hat f^\pi)^{k} = \mu \hat f^\pi = \mu f^\pi\]
Since $\pi$ is arbitrary and there are only finitely many
policies, we obtain
\[\lim_{n\to\infty} \mu f_n = \lim_{n\to\infty} \max_{\pi} \mu f_n^\pi = \max_{\pi} \lim_{n\to\infty} \mu f_n^\pi = \max_{\pi} \mu f^\pi = \mu f\]
almost surely.
  \qed
\end{proof}

Now, we can show that, also for general MDPs, a dampened Mann
iteration based on a regular Mann scheme converges to the optimal
value function. 
We first prove a preliminary result concerning
\like~functions.

\begin{propositionrep}[Mann iteration for \like~functions]
  \label{pr:mann-mdp-like}
  Let $h\colon \Rp^S\rightarrow \Rp^S$ be \like.  Let
  $(h_n)$ be a sequence of functions converging
  to $h$.  Given a regular Mann scheme
  $\mathcal{S} = ((\alpha_n),(\beta_n))$ and arbitrary $v\in \Rp^S$,
  consider the iteration
  $\mathcal{F}^{\app}=(\mathcal{S},(h_n\sqcup\id),v)$.

  If $(x_n^{\mathcal{F}^{\app}})$ is bounded, then $x_n^{\mathcal{F}^{\app}} \to \mu h$.
\end{propositionrep}
\begin{proof}
  Let us write $v_n^{\app}$ instead of $v_n^{\mathcal{F}^{\app}}$. Consider
  the iteration $\mathcal{F}=(\mathcal{S},h \sqcup id,v_0)$ and let us
  write $v_n$ for $v_n^{\mathcal{F}}$.

  Since $h \sqcup id$ is non-expansive, monotone and has at least one fixpoint, we know
  that $v_n \to \mu(h \sqcup id) = \mu h$ by regularity.

  In order to conclude we show that $v_n^{\app}$ converges to the same
  limit of $v_n$. Consider
  \begin{align*}
    & |v_{n+1}^{\app} - v_{n+1}| = \\
    & = | (1 -\beta_n)(\alpha_n (v_n^{\app} -v_n) + (1-\alpha_n) ((h_n \sqcup id)(v_n^{\app}) - (h \sqcup id)(v_n)))|\\
    &\leq (1 -\beta_n)(\alpha_n |v_n^{\app} -v_n|  + (1-\alpha_n) | (h_n \sqcup id)(v_n^{\app}) - (h \sqcup id)(v_n)|\\
    & \leq (1 -\beta_n)(\alpha_n |v_n^{\app} -v_n|  +
           (1-\alpha_n) (| h_n(v_n^{\app}) -h(v_n)| \sqcup  |id(v_n^{\app}) - id(v_n)|)
  \end{align*}

  Now, if $u$ is the function witnessing that $h$ is \like~we have
  \begin{align*}
    & | h_n(v_n^{\app}) -h(v_n)| \leq\\
    & \leq | h_n(v_n^{\app}) -h(v_n^{\app})| + | h(v_n^{\app}) -h(v_n)| \\
    & \leq \epsilon_n + u(| v_n^{\app} - v_n|)
  \end{align*}
  where $\epsilon_n = \sup_{i\geq n} \|h_n(v_n^{\app}) -h(v_n^{\app})\|$ ($\to 0$ since $h_n\to h$ and $(v_n^{\app})$ bounded).
  We can use this in the analysis of $|v_{n+1}^{\app} - v_{n+1}|$ and deduce that
  \begin{align*}
    & |v_{n+1}^{\app} - v_{n+1}| \leq \\
    & \leq (1 -\beta_n)(\alpha_n |v_n^{\app} -v_n|  +
           (1-\alpha_n) \max\{ u(|v_n^{\app} -v_n|) + \epsilon_n, |v_n^{\app}-v_n|\}\\
    & = (1 -\beta_n)(\alpha_n |v_n^{\app} -v_n|  +  (1-\alpha_n) (h'_n \sqcup id)(|v_n^{\app} -v_n|))
  \end{align*}
  where $h'_n(v) = u(v) + \epsilon_n$.

  If we denote by $(v_n')$ the sequence
  \begin{quote}
    $v_0' = | v_0^{\app}-v_0|$ \\
    $v_{n+1}' = (1-\beta_n) \left(\alpha_n v_n' + (1-\alpha_n)((h_n' \sqcup id)(v_n'))\right)$
  \end{quote}
  then $|v_n^{\app} - v_n| \leq v_n'$. Observe that for $v'_n$ we
  again obtain a dampened Mann iteration.

  Now, observe that $\epsilon_n$ is a decreasing sequence, converging
  to $0$ since $h_n \to h$ (and $(v_n^{\app})$ bounded). Therefore, we have that $h_n' \sqcup id \to u \sqcup id$
  monotonically from above. Moreover $h_n'$ is \like, in fact
  $|h_n'(v) - h_n'(v')| = |u(v) - u(v')| \leq u(|v-v'|)$.

  Hence by \Cref{pr:MDP-fixpoint} there exists $c$ such that:
  \begin{center}
    $|\mu h_n'| \leq \frac{1}{1-c} \|(h_n')^k(0)\| \leq \frac{k}{1-c} \epsilon_n$
  \end{center}
  since $(h_n')^k(0) \leq k \epsilon_n$.
  Hence $\mu h_n' \to 0 = \mu h'$ where $h' = u\sqcup \id$.

  Thus, by \Cref{thm:approximate-convergence}\eqref{it:thm:mt-mon}, we
  have that $v_n \to \mu(h' \sqcup id) = \mu h' = 0$. Since $v'_n$ is
  an upper bound for $|v_n^{\app} - v_n|$ we conclude that
  $v_n^{\app}$ converges to the same limit of $v_n$, i.e.,
  $\mu(h \sqcup id) = \mu h$, as desired.
  \qed
\end{proof}

\begin{toappendix}
\begin{lemma}
  \label{le:merge-upper-bound}
  Let $(h_n)$ and $(h'_n)$ be sequences of functions with
  $h_n\colon \Rp^S \to \Rp^S$,
  $h_n'\colon \Rp^{S'} \to \Rp^{S'}$ and
  $\merge\colon S \to S'$ be a function such that
  $h_n\circ \reindex{\merge} \le \reindex{\merge}\circ h'_n$.
  
  Let $x^\app_n, (x')^\app_n$ be the corresponding dampened Mann
  iterations, for which we assume that $x^\app_0\le
  \reindex{\merge}((x')^\app_0)$. Then $x^\app_n \le \reindex{\merge}((x')^\app_n)$ holds for
  all $n$.
\end{lemma}

\begin{proof}
  Let $w\in[0,1]^{S'}$ and $s\in S$. We first observe that
  $(c\cdot \reindex{\merge}(w))(s) = c\cdot w(\merge(s)) = (c\cdot w)(\merge(s)) = \reindex{\merge}(c\cdot
  w)(s)$, hence $c\cdot \reindex{\merge}(w) = \reindex{\merge}(c\cdot w)$. Additionally for
  $w_1,w_2\in[0,1]^{S'}$ we obtain
  $(\reindex{\merge}(w_1)+\reindex{\merge}(w_2))(s) = \reindex{\merge}(w_1)(s)+\reindex{\merge}(w_2)(s) =
  w_1(\merge(s))+w_2(\merge(s)) = (w_1+w_2)(\merge(s)) = \reindex{\merge}(w_1+w_2)(s)$, hence
  $\reindex{\merge}(w_1)+\reindex{\merge}(w_2) = \reindex{\merge}(w_1+w_2)$.
  
  Using these facts we can show by induction that
  \begin{eqnarray*}
    x^\app_{n+1} & = & (1-\beta_n)\cdot (\alpha_n\cdot x_n +
                       (1-\alpha_n)\cdot
                       h_n(x^\app_n)) \\
                 & \le & (1-\beta_n)\cdot (\alpha_n\cdot \reindex{\merge}((x')^\app_n) + (1-\alpha_n)\cdot
                         h_n(\reindex{\merge}((x')^\app_n))) \\
                 & \le & (1-\beta_n)\cdot (\alpha_n\cdot \reindex{\merge}((x')^\app_n) + (1-\alpha_n)\cdot
                         \reindex{\merge}(h'_n((x')^\app_n))) \\
                 & = & (1-\beta_n))\cdot \reindex{\merge}((\alpha_n\cdot (x')^\app_n + (1-\alpha_n)\cdot
                       h'_n((x')^\app_n)) \\
                 & = & \reindex{\merge}((x')^\app_{n+1}),
  \end{eqnarray*}
  \qed
\end{proof}
\end{toappendix}

With this, we can now observe that the dampened Mann iteration for an
arbitrary sampled MDP indeed converges to the optimal value function.

\begin{maintheoremrep}[Mann iteration over MDPs]%
  \label{th:MDP-bound}
  Let $M, M_n$ be MDPs as in \Cref{as:mdp-setting}. Given a regular Mann scheme
  $\mathcal{S}$ and any $v\in\Rp^S$, consider the iteration
  $\mathcal{F}^{\app}=(\mathcal{S},(f_n),v)$, generating a sequence
  $(v_n^{\mathcal{F}^{\app}})$.
  Then  $\lim_{n\to \infty} v_n^{\mathcal{F}^{\app}} = \opts{M}=\mu f_M$.
\end{maintheoremrep}

\begin{proofsketch}
  The full proof can be found in the appendix.
  
  \begin{itemize}
  \item Upper bound ($\lim_{n\to\infty}v_n^{\mathcal{F}^{\app}}\le \mu f_M$):
    Eventually, the approximating MDPs have the same maximal end
    components as the exact MDP, and we observe that each such
    maximal end-component can be merged to a singleton end-component
    (with loop) to obtain an over-approximation
    (\Cref{def:quotient-mdp} and \Cref{pr:MDP-has-fixpoint}). Now, an
    MDP with only singleton end-components can be over-approximated by
    $f_M\sqcup \id$, where $M$ is a simple MDP. We know that simple
    MDPs are \like~and that the sequence $(v_n^{\mathcal{F}^{\app}})$ is bounded
    whence we can conclude with \Cref{pr:mann-mdp-like}.

  \item Lower bound
   ($\mu f_M \le \lim_{n\to\infty}x_n^{\mathcal{F}}$): The expected
   reward for an MDP $M$ is the maximum of the expected rewards for
   Markov chains $M^\pi$ over all policies $\pi$.  In a Markov
   chain all states in an end-component have reward $0$. Fixing the
   value to $0$ in the states of an end-component gives us a power
   contraction with the same fixpoint.  From this one can deduce that
   the true expected reward ($\mu f_M$) is always a lower bound for
   the result of the dampened Mann iteration.\qedhere
  \end{itemize}  
\end{proofsketch}

\begin{proof}
  Let us write $v_n^{\app}$ instead of $v_n^{\mathcal{F}}$. We show that
  the cluster points of the sequence can be bounded from above and
  from below by $\mu f_M$, and then we conclude.

  First of all, observe that we can assume that all approximations
  $M_n$ have the same MECs which are in turn the MECs of $M$. In fact,
  by \Cref{as:mdp-setting}, the approximations $M_n$ are
  obtained by sampling and $T_n \to T$ (almost surely). Hence we can
  find an index $n_0$ such that for all $n \geq n_0$
  $\Tr[]{s}{a}{s'} > 0$ if and only if $\Tr[_n]{s}{a}{s'} > 0$ and
  consider the iteration starting from $v_{n_0}^{\app}$.

  Hereafter we work under this assumption, indicate by $\equiv$ is the
  function identifying states in the same MEC of $M^\pi$ and by
  $\merge\colon S \to S/_{\equiv}$ the corresponding quotient
  function.

  \paragraph{Upper bound}  Let $(f_\merge)$ be the function associated
  to the quotient $(M_n)_\merge$ defined as in
  \Cref{def:quotient-mdp}.
  Let $\mathcal{S} = ((\alpha_n),(\beta_n))$.
  By \Cref{le:merge-upper-bound}, the sequence
  $(w_n^{\app})$ generated by the iteration
  $\mathcal{F}=(\mathcal{S},(f_\merge),w_0^{\app})$ with
  $w_0^{\app}$ such that $v_0^{\app} \leq \reindex{\merge}(w_0^{\app})$ is such that
  that
  \begin{center}
    $v_n^{\app} \leq \reindex{\merge}(w_n^{\app})$
  \end{center}
  If we denote $(\hat{f}_\merge)$ be the function associated with the reduced quotient
  $(\hat{M}_n)_\merge$ then, by \Cref{le:MDP-lift},
  $(f_\merge) \leq (\hat{f}_\merge) \sqcup id$. Hence, if we consider the
  iteration
  $\mathcal{F}=(\mathcal{S},(\hat{f}_\merge) \sqcup id,z_0^{\app})$,
  with $z_0^{\app}= w_0^{\app}$, and call $(z_n^{\app})$ the induced sequence then

  \begin{center}
    $w_n^{\app} \leq z_n^{\app}$
  \end{center}

  Moreover, it is clear that $(\hat{f}_\merge)$ converges to $\hat{f}_\merge$, the
  function associated to the reduced quotient $\hat{M}_\merge$ of $M$.
  Since $\hat{f}_\merge$ is \like~and $(z_n^{\app})$ is bounded by \Cref{le:boundedness}, by
  \Cref{pr:mann-mdp-like}, $z_n^{\app}$ converges to $\mu \hat{f}_\merge$.

  Putting things together and using monotonicity and continuity of
  $\reindex{\merge}$, we have that
  $v_n^{\app} \leq \reindex{\merge}(w_n^{\app}) \leq \reindex{\merge}(z^{\app}_n)$. Moreover
  $\reindex{\merge}(z_n^{\app})$ converges to $\reindex{\merge}(\mu \hat{f}_\merge) = \mu f_M$, with
  the last equality given by \Cref{pr:MDP-has-fixpoint}.

  We thus conclude that $\mu f_M$ is an upper bound for the
  cluster points of $x_n^{\app}$.

  \paragraph{Lower bound:}
  Let $f=f_M$. Consider some fixed policy $\pi\colon S \to A$ and let
  $M_n^\pi$ be the Markov chain obtained by the $n$-th
  approximation of $M$. Clearly the corresponding iteration function
  satisfies $f_n^\pi \leq f_n$. Moreover, if we indicate by
  $E = \bigcup_n MEC(M_n)$, the set of states in some end-component of
  $M_n^\pi$ and, given
  $h\colon \Rp^S \to \Rp^S$ we write $h^E$ for the function

  \[h^E(v)(s)=\begin{cases}
      0&\text{if $s \in E$}\\
      h(v)(s)&\text{otherwise}.
    \end{cases}\]

  Then, if $\hat{f}^\pi_\merge$ is the function associated to the
  reduced quotient (see \Cref{def:quotient-mdp}), it holds
  $(f^\pi)^E = \reindex{\merge} \circ \hat{f}^\pi_\merge\circ\max_{\merge}$. Since the
  reduced quotient is simple, we infer that $(f^\pi)^E$ is a
  power contraction (note that $\max_\merge\circ\reindex{\merge}=id$ whence $((f^\pi)^E)^n = \reindex{\merge} \circ (\hat{f}^\pi_\merge)^n\circ\max_{\merge}$) and, by
  \Cref{pr:MDP-has-fixpoint},
  $\mu (f^\pi)^E = \reindex{\merge}(\mu\hat{f}^\pi_\merge) = \mu f^\pi$.

  Consider the sequence $y_n^\pi$
  generated by the iteration
  $\mathcal{F}=(\mathcal{S},(f_n^\pi)^E,y_0^\pi)$
  with $y_0^\pi \leq x_0^{\app}$. Remembering that
  $(f_n^E)^\pi \leq f_n^\pi \leq f_n$ we have that
  \begin{center}
    $y_n^\pi \leq x_n^{\app}$
  \end{center}

  Since $(f_n^\pi)^E$ converges to $(f^\pi)^E$ and
  $(f^\pi)^E$ is a power contraction, by boundedness of
  $y_n^\pi$ and regularity of the
  scheme, $y_n^\pi$ converges to $\mu(f^\pi)^E=\mu f^\pi$, which is
  therefore a lower bound for the cluster points of $x_n^{\app}$.

  The above applies to all possible policies $\pi\colon S \to A$,
  hence also $\max_{\pi} \mu f^\pi = \mu f$ is a lower
  bound.

  \bigskip Given that the upper and lower bound coincide with $\mu f$
  this is necessarily the limit of $x_n^{\app}$.
  \qed
\end{proof}

The convergence result for the sequence generated by the state-value
operator can be used to derive a similar result for the
state-action-value operator.

\begin{corollaryrep}[Mann iteration over MDPs -- state-action-value operator]
  Under \Cref{as:mdp-setting}, given a Mann-Kleene scheme
  $\mathcal{S}$ and any initial vector $q\in\Rp^{S \times A}$, consider the iteration
  $\mathcal{G}^\app=(\mathcal{S},(g_n),q)$, producing a sequence
  $(q_n^{\mathcal{F}^{\app}})$.
  Then,  $\lim_{n\to \infty} q_n^{\mathcal{G}^{\app}} = \optsa{M} = \mu g_M$.
\end{corollaryrep}

\begin{proof}

  Let us write $q_n$ instead of
  $q_n^{\mathcal{G}^\app}$.  Define $v \in \Rp^S$ by
  $v(s) = \max_{a \in A(s)} q(s,a)$ and consider the sequence $(v_n)$
  generated by the iteration
  $\mathcal{F}^{\app}=(\mathcal{S},(f_n),v)$ as in
  \Cref{th:MDP-bound}, for which it is known that
  $\lim_{n\to \infty} v_n^{\app} = v_M$.

  As a first step, we can show, by an inductive argument, that for
  all $n \in \N$ and $s \in S$
  \begin{equation}
    \label{eq:vn-qn-bound}
    \max_{a \in A(s)} q_n(a,s) \leq v_n(s)
  \end{equation}
  The base case holds by the choice of $v$.
  For the inductive step, observe that
  \begin{align*}
    v_{n+1}(s)
    & = (1-\beta_n)(\alpha_n v_n(s) + (1-\alpha_n) f_n(v_n)(s))\\
    & =(1-\beta_n)(\alpha_n v_n(s) + (1-\alpha_n) \max_{a\in A(s)} \sum_{s'\in S} \Tr[_n]{s}{a}{s'} (R(s,a,s') + v_n(s')))\\
    & \geq (1-\beta_n)(\alpha_n \max_{a\in A(s)} q_n(s,a) +
    (1-\alpha_n) \max_{a\in A(s)} \sum_{s'\in S} \Tr[_n]{s}{a}{s'}
    (R(s,a,s') \\
    & \qquad + \max_{a'\in A(s')} q_n(s',a')))\\
    & \quad \mbox{[by inductive hypothesis $v_n(s') \geq \max_{a\in A(s')} q_n(s',a')$]}\\
    & \geq \max_{a\in A(s)} (1-\beta_n)(\alpha_n q_n(s,a) + (1-\alpha_n) \sum_{s'\in S} \Tr[_n]{s}{a}{s'} (R(s,a,s')\\
    & \qquad  + \max_{a\in A(s')} q_n(s',a'))) \\
    & = \max_{a\in A(s)} (1-\beta_n)(\alpha_n q_n(s,a) + (1-\alpha_n) g_n(q_n)(s,a))\\
    & =\max_{a\in A(s)} q_{n+1}(s,a)\\
  \end{align*}

  We can now use the above to get an upper bound for $q_{n+1}(s,a)$.
  In fact
  \begin{align*}
    & q_{n+1}(s,a) \\
    &=(1-\beta_n)(\alpha_n q_n(s,a) + (1-\alpha_n) g_n(q_n)(s,a))\\
    &=(1-\beta_n)(\alpha_n q_n(s,a) + (1-\alpha_n) (\sum_{s'\in S}\Tr[_n]{s}{a}{s'}(R(s,a,s') + \max_{a'\in A(s')} q_n(s',a'))))\\
    & \leq (1-\beta_n)(\alpha_n q_n(s,a) + (1-\alpha_n) (\sum_{s'\in S}\Tr[_n]{s}{a}{s'}(R(s,a,s') + v_n(s'))))
  \end{align*}
  where the last inequality is justified by~\eqref{eq:vn-qn-bound}.

  Recalling that
  $\beta_n \to 0$ and $\alpha_n \to 0$, we have
  \begin{align*}
    & \lim_{n\to\infty} (1-\beta_n)(\alpha_n q_n(s,a) + (1-\alpha_n) (\sum_{s'\in S}\Tr[_n]{s}{a}{s'}(R(s,a,s') + v_n(s'))))\\
    & = \lim_{n\to\infty}  \sum_{s'\in S}\Tr[_n]{s}{a}{s'}(R(s,a,s') + v_n(s'))\\
    &
      = \sum_{s'\in S}\Tr[]{s}{a}{s'}(R(s,a,s') + \lim_{n\to\infty} v_n(s'))\\
    & = \sum_{s'\in S}\Tr[]{s}{a}{s'}(R(s,a,s') +  \opts{M}(s'))\\
    & = \optsa{M}(s,a)
  \end{align*}
  and thus
  \begin{equation}
    \label{eq:q-sup}
    \limsup_{n\to\infty} q_n(s,a) \leq \optsa{M}(s,a).
  \end{equation}

  Together with \Cref{lm:lower}, this implies
  $\lim_{n \to\infty} q_n = \optsa{M}$.
  \qed
\end{proof}

\subsection{About the (Non-)Discounted Reward Setting}
\label{ss:non-discounted}

In the context of (model-based) reinforcement learning, usually the
discounted total reward criterion is used, that is the aim of the
agent is to optimize the discounted total reward for some discount
factor $\gamma<1$, as described in \Cref{se:intro}.
This setting yields some desirable properties, mainly that the
resulting discounted Bellman operator is contractive, making the
optimal value its unique fixpoint and always well-defined.  It can
also be motivated by interpreting it as a scenario in which the agent
prefers an immediate reward over long-term rewards or where the
system terminates 
(or the agent stops its interaction)
at any step
with probability $1-\gamma$.

However,
we argue that the non-discounted total reward criterion
(i.e., $\gamma=1$)  not only can be handled with model-based
reinforcement learning approaches such as the iteration scheme treated
in this section, but is also more general and strictly more
expressive.

The total discounted reward setting can be interpreted as a
special case of the non-discounted one by adding a sink state which is
reached by every state with probability $1-\gamma$ (akin to its
interpretation as a system that stops with probability $1-\gamma$ in
each step).  Then the resulting MDP is not only simple in the
sense of \Cref{de:MEC} but its Bellman operators are, with slight
modifications, also contractions whence we can apply \Cref{th:MDP-bound}
to derive convergence results.
In particular, this gives a convergence result of simple versions of
model-based algorithms such as the Dyna-Q
algorithm~\cite{s:dyna-integrated-architecture,s:planning-incremental-dyn-prog}
where we consider only full-sweep updates and no direct update step.
In addition, the non-discounted case is often the preferable setting:
If, as suggested above, we interpret the discounted setting as an
instance of the non-discounted setting, in the former it is impossible
to have states where the system may \emph{not} terminate. This is a
restriction that, in practical applications, one might want to lift,
i.e., there might be system states in which it should be impossible
for the system to terminate. Even more importantly, adding discounts
does not always faithfully model the payoff or concrete benefit for
the agent. It might not be relevant if rewards are obtained now or in
the future, for instance if we simply want to compute the probability
of eventually reaching some
goal.%

\section{Reaching the Least Fixpoint by Fast Iteration}

\subsection{A Generic Approximation Algorithm}
\label{se:mann-error-algo}

In the previous section we proved a convergence result for MDPs, approximated by sampling. It would be desirable to extend such results
to more general settings involving probabilistic systems which are explored by sampling.
Indeed, our results allow the formulation of a
generic algorithm that can compute a sequence of values
converging to the least fixpoint of $f$ almost surely.  This can be
done in all cases where we can estimate the error $\|f_n - f\|$ at
least with a certain probability. The procedure, reported in~\Cref{Algorithm:GenericApproximationAlgorithm}, relies on a very simple idea:
make sure algorithmically that the sequence $(f_n)$ converges normally
to $f$ almost surely.
The correctness of
the procedure follows by the Borel-Cantelli Lemma~\cite{alma991015056379706467}.

\begin{algorithm}[t]
\noindent\KwIn{A map $n \mapsto f_n$ such that $(f_n)$ uniformly converges to $f$ almost
  surely. \\ Parameters $\gamma_i,\delta_i\in \Rp$ such that
  $\gamma_i>0$ and $\sum_{i\in\N} \gamma_i
    < \infty$ (analogously for $\delta_i$)}
\noindent\KwOut{A sequence $(x_i)$ in $\Rp^d$}
  
$x_0 := 0$\;
$i := 0$\;
\While{true} {
  Choose $n_i$ such that $\P\big[\|f_{n_i} - f\| \ge
  \gamma_i \big] \le \delta_i$\;
  $x_{i+1} := (1 - \beta_i)\cdot \big(\alpha_i \cdot x_i + (1 - \alpha_i) \cdot f_{n_i}(x_i)\big)$\;
  $i := i+1$\;
}
\caption{Dampened Mann Iteration with Probabilistic Guarantees}
\label{Algorithm:GenericApproximationAlgorithm}
\end{algorithm}

\begin{theoremrep}
For exact Mann schemes the sequence $(x_i)$ produced by \Cref{Algorithm:GenericApproximationAlgorithm} converges to the least fixpoint of $f$ almost surely.
\end{theoremrep}
\begin{proof}
By construction we have
\[\sum_{i=1}^\infty \P\big[\|f_{n_i} - f\| > \gamma_i \big]
  \le \sum_{i=1}^\infty \delta_i < \infty\]

The Borel-Cantelli Lemma \cite{alma991015056379706467} states:

Let $(A_n)_{n\in \N}$ be a sequence of events in a probability space
and assume that $\sum_{n=1}^\infty \P[A_n] < \infty$, i.e., the
sum of probabilities is finite. Then

\[ \P[\bigcap_{n=1}^\infty \bigcup_{k=n}^\infty A_k] =
  \P[\{\omega\in \Omega \mid \text{there are infinitely many
    $n\in\N$ with $\omega\in A_n$}\}] = 0 \]

So by applying this lemma we get that
\[\P\big[\|f_{n_i} - f\| > \gamma_i \text{ for infinitely many } i\big] = 0\]
Hence, almost surely we have $\|f_{n_i} - f\| \le \gamma_i$ eventually
and by \Cref{thm:approximate-convergence}\eqref{it:thm:mt-norm} the
sequence produced by the algorithm converges to the solution vector of
its input in that case.
 \qed
\end{proof}

The only precondition for the algorithm is that the indices $n_i$ in
line~$6$ can be chosen effectively. When the sequence $(f_n)$ is
generated by a sampling process as described above, this can often be
done by using Bernstein's
inequality~\cite{409cf137-dbb5-3eb1-8cfe-0743c3dc925f} -- a
quantitative version of the law of large numbers.

\begin{theorem}[Bernstein's inequality ]
  Let $(A_n)$ be a sequence of independent events on a probability
  space $(\Omega, \mathcal{A}, \P)$ with fixed probability
  $\P(A_i) = p$ for all $i$.  For $n\in\N$ let
  $S_n(\omega)=|\{1\le i\le n \mid \omega\in A_i\}|$. Then for all
  $\epsilon>0$ we have
  \[\P\bigg[\bigg|\frac{S_n}n - p\bigg|\ge\epsilon\bigg]\le
    2e^{-2\epsilon^2n}\]
\end{theorem}
If we sample the transition probabilities
$\Tr[_n]{s}{a}{s '}$ as detailed in \Cref{sec:mdp-sampling}, we can
estimate the error as below, which allows us to computationally
determine the indices needed in line $4$ of the algorithm:
\[ \P\big[|\Tr[_n]{s}{a}{s'} - \Tr[]{s}{a}{s'}| \ge \epsilon
  \big] \leq 2 e^{- 2 \epsilon^2n} \]

\subsection{Application: Simple Stochastic Games}
\label{se:SSGs}

As an example application of the technique in
Section~\ref{se:mann-error-algo} we consider (simple) stochastic games
or SSGs. An SSG is given by a finite set of nodes partitioned into
$V_{\mathrm{min}}$, $V_{\mathrm{max}}$, $V_{\mathrm{av}}$ and
$V_{\mathrm{sink}}$, and the following data:
$\eta\colon V_{\mathrm{min}}\cup V_{\mathrm{max}}\to \mathcal{P}(V)$
(successor function for Min and Max nodes),
$\eta_{\mathrm{av}}\colon V_{\mathrm{av}} \to \mathcal{D}(V)$ (where
$\mathcal{D}(V)$ is the set of probability distributions over $V$)
and a map $w\colon V_{\mathrm{sink}} \to [0,1]$ that assigns to each
sink state a payoff. The value vector of an SSG is the least fixpoint
of the monotone and non-expansive map $f\colon [0,1]^V \to [0,1]^V$
given by
\[
  f(p)(v) =
  \begin{cases}
    \min_{u\in\eta(v)} p(u) & \mathrm{if}\  v\in V_{\mathrm{min}} \\
    \max_{u\in\eta(v)} p(u) & \mathrm{if}\  v\in V_{\mathrm{max}} \\
    \sum_{u\in V} \eta_{\mathrm{av}}(v)(u) \cdot p(u) &
    \mathrm{if}\ v\in V_{\mathrm{av}} \\
    w(v) & \mathrm{if}\ v\in V_{\mathrm{sink}}
  \end{cases}
\]
As above we assume that the transition probabilities of the average
nodes are not known precisely but can only be approximated by
sampling. Formally, for each average node
$v\in V_{\mathrm{av}}$  we have  a collection of independent and identically
distributed random variables $(X^v_n)_{n\in\N}$ with values in $V$
and distribution $\P[X^v_n = u] = \eta_{\mathrm{av}}(v)(u)$
for all $u\in V$, $n\in\N$. For all $n$ we can estimate the true
transition probabilities by the relative frequency
\[(\eta_{\mathrm{av}})_n(v)(u) = \frac{|\{1\le k \le n \mid X^v_k =
    u\}|}{n}\] If we denote by $f_n$ the fixpoint function of the
SSG based on $(\eta_{\mathrm{av}})_n$, we have
\[\|f - f_n\| = \max_{v\in V_{\mathrm{av}}} \max_{u\in V}
  |\eta_{\mathrm{av}}(v)(u) - (\eta_{\mathrm{av}})_n(v)(u)|\] In
particular, by the law of large numbers the sequence of maps $(f_n)$
converges to the correct map $f$ with probability $1$. But more
importantly -- using Bernstein's inequality or similar estimates -- we
can, for every $\delta, \epsilon>0$ effectively produce an index
$n\in\N$ such that $\P[\|f - f_n\| \ge \delta] \le \epsilon$.
Using \Cref{Algorithm:GenericApproximationAlgorithm} we can therefore
iterate to the value of an SSG without knowing the true transition
probabilities by simply using approximations obtained from sampling
for each value update provided that sampling and value updates are
interleaved in a suitable way.

In \Cref{se:mdp} we have seen by a careful analysis that for the
special case of Bellman objective functions we do not have to pay
attention to the way in which sampling and value updates are
interleaved.
Here, on the other hand, we showed that our methods
are flexible enough to be applicable to
many other problems which can be phrased as computing
least fixpoints of monotone and non-expansive functions. Along the lines of the
example presented here, it is not strictly necessary
to redo an analysis of the new objective function as we did in
\Cref{se:mdp}. Rather, it is sufficient to be able to effectively
compute probabilistic error bounds for
\Cref{Algorithm:GenericApproximationAlgorithm} to work.

\section{Numerical Experiments}
\label{se:experiments}

While our main goal was to merge the existing theory behind the
dampened Mann iteration and approximative settings such as model-based
reinforcement learning, with a focus on formal proofs of convergence,
in this section, we want to gain some practice-based intuition about
how a dampened Mann iteration approximates the least fixpoints and how
it compares to classical Kleene or (undampened) Mann iteration.

Here, an iteration scheme $\mathcal{S}=((\alpha_n),(\beta_n))$
with $\alpha_n = 0$ will be called a \emph{Kleene} scheme,
otherwise (mostly, choosing $\alpha_n = 1/n$) a \emph{Mann} scheme.  It
is called \emph{undampened} if $\beta_n\equiv0$ and \emph{dampened}
otherwise (again, choosing $\beta_n =  1/n$ in general).  Furthermore, we
will consider schemes \emph{with resetting}, which iterate a single approximation from an initial point, i.e., 
$x_n^{\mathcal{F}}=T_n(f_n,T_{n-1}(f_n,\dots T_0(f_n,x)\dots))$
without reusing previous results. We will typically start with $x=0$
in a resetting iteration.

The experiments were performed on an Intel i7-11850H processor with 16GB of RAM.

\paragraph{Known function.} First of all, the dampened Mann iteration
does not guarantee, in any way, faster convergence when compared to classical Kleene iteration.
Its benefits lie in a more general applicability not in a better efficiency.
In fact, as to be expected when the exact function is known and the iteration starts at $0$, we can
see both in a simple example using the function
\[f(x)=\min(\max(\nicefrac{1}{2}x + \nicefrac{1}{2}, x), \nicefrac{1}{2}x + 1)\]
depicted in \Cref{fig:simple-exact}, as well as a more complex example using the (state-value)
Bellman operator of the Markov decision process depicted in \Cref{fig:mdp-num-example}, that
Kleene iteration converges faster than iteration schemes using positive parameters
$\alpha_n,\beta_n$ (\Cref{fig:simple-exact-results,fig:mdp-example-exact-results}).

\begin{figure}
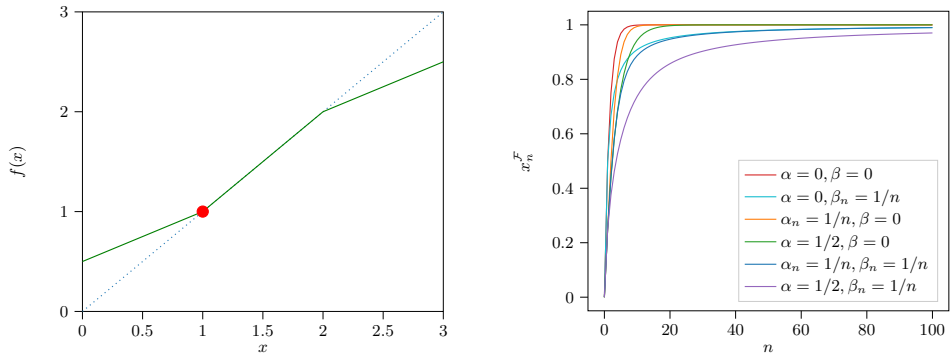

  \centering
  \begin{subfigure}[b]{.45\textwidth}
    \centering
\clearpage{}%

\clearpage{}%

    \caption{The simple piecewise linear function $f$ (green) with its least fixpoint $\mu f=1$
      (red)}
    \label{fig:simple-exact}
  \end{subfigure}
  \hfill
  \begin{subfigure}[b]{.45\textwidth}
    \centering
\clearpage{}%
%
\clearpage{}%

    \caption{Generated sequences $(x_n^{\mathcal{F}})$ of the
      different iteration schemes (starting from $0$)}
    \label{fig:simple-exact-results}
  \end{subfigure}
  \caption{Results of different iteration schemes for an exact
    function $f$.}
\end{figure}

On the other hand, as we already discussed in \Cref{se:mann}, positive dampening factors
allow the iteration to converge to the least fixpoint for \emph{any} initial point whereas it is
easy to check that both (undampened) Kleene and undampened Mann iterations, in general, do not
converge or get stuck at another fixpoint (e.g., take the flip map $f(x,y)=(y,x)$ with initial
point $x_0=(1,0)$).

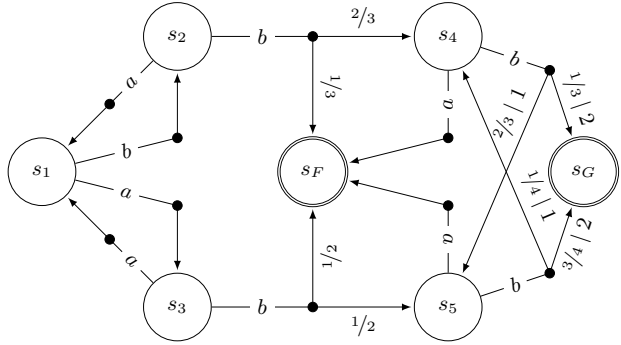
\begin{figure}
\clearpage{}%
\scalebox{0.9}{
\begin{tikzpicture}[every node/.style={draw}]
  \tikzstyle{action}=[>=latex, nodes={font=\footnotesize,sloped,draw=none},fill,draw,black]
  \path[shape=circle, minimum size=1cm]
    (0,0) node(s1)[]{$s_1$}
    (2,2) node(s2)[]{$s_2$}
    (2,-2) node(s3)[]{$s_3$}
    (6,2) node(s4)[]{$s_4$}
    (6,-2) node(s5)[]{$s_5$}
    (4,0) node(sF)[accepting]{$s_F$}
    (8,0) node(sG)[accepting]{$s_G$};

  \path[action] (1,1) circle (2pt)
    edge node[fill=white]{$a$}  (s2)
    edge [post] node[above]{} (s1);
  \path[action] (1,-1) circle (2pt)
    edge node[fill=white]{$a$}  (s3)
    edge [post] node[below]{} (s1);
  \path[action] (2,0.5) circle (2pt)
    edge node[fill=white]{$b$}  (s1)
    edge [post] node[below]{} (s2);
  \path[action] (2,-0.5) circle (2pt)
    edge node[fill=white]{$a$}  (s1)
    edge [post] node[above]{} (s3);
  \path[action] (4,2) circle (2pt)
    edge node[fill=white]{$b$} (s2)
    edge [post] node[above]{$\nicefrac{1}{3}$} (sF)
    edge [post] node[above]{$\nicefrac{2}{3}$} (s4);
  \path[action] (4,-2) circle (2pt)
    edge node[fill=white]{$b$} (s3)
    edge [post] node[below]{$\nicefrac{1}{2}$} (sF)
    edge [post] node[below]{$\nicefrac{1}{2}$} (s5);
  \path[action] (6,0.5) circle (2pt)
    edge node[fill=white]{$a$}  (s4)
    edge [post] node[above]{} (sF);
  \path[action] (6,-0.5) circle (2pt)
    edge node[fill=white]{$a$}  (s5)
    edge [post] node[below]{} (sF);
  \path[action] (7.5,1.5) circle (2pt)
    edge node[fill=white]{$b$} (s4)
    edge [post] node[above]{$\nicefrac{1}{3}\mid 2$} (sG)
    edge [post] node[above, near start]{$\nicefrac{2}{3}\mid 1$} (s5);
  \path[action] (7.5,-1.5) circle (2pt)
    edge node[fill=white]{$b$} (s5)
    edge [post] node[below]{$\nicefrac{3}{4}\mid 2$} (sG)
    edge [post] node[above right=1pt]{$\nicefrac{1}{4}\mid 1$} (s4);
\end{tikzpicture}
}

\clearpage{}%

  \centering
  \caption{An MDP with one maximal end-component ($s_1,s_2,$ and $s_3$) where unspecified transition
  probabilities are $1$ and unspecified rewards are $0$.}
  \label{fig:mdp-num-example}
\end{figure}

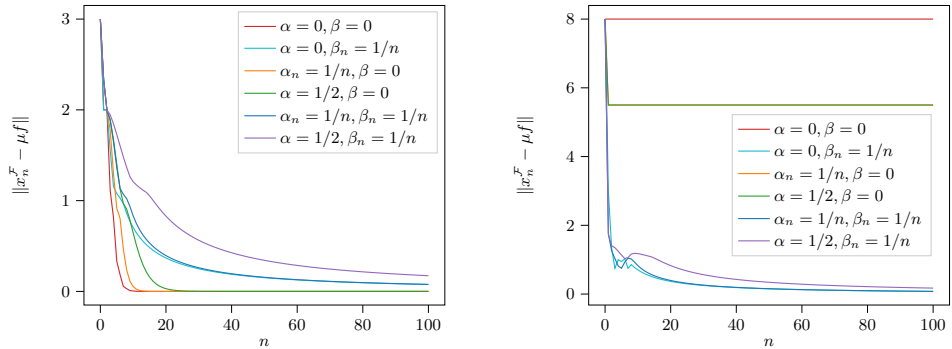
\begin{figure}
  \centering
  \begin{subfigure}[b]{.45\textwidth}
    \centering
\clearpage{}%
\begin{tikzpicture}[scale=0.7]

\definecolor{crimson2143940}{RGB}{214,39,40}
\definecolor{darkgray176}{RGB}{176,176,176}
\definecolor{darkorange25512714}{RGB}{255,127,14}
\definecolor{darkturquoise23190207}{RGB}{23,190,207}
\definecolor{forestgreen4416044}{RGB}{44,160,44}
\definecolor{lightgray204}{RGB}{204,204,204}
\definecolor{mediumpurple148103189}{RGB}{148,103,189}
\definecolor{steelblue31119180}{RGB}{31,119,180}

\begin{axis}[
legend cell align={left},
legend style={fill opacity=0.8, draw opacity=1, text opacity=1, draw=lightgray204},
tick align=outside,
tick pos=left,
x grid style={darkgray176},
xlabel={\(\displaystyle n\)},
xmin=-5, xmax=105,
xtick style={color=black},
y grid style={darkgray176},
ylabel={\(\displaystyle \|x_n^{\mathcal{F}}-\mu f\|\)},
ymin=-0.15, ymax=3.15,
ytick style={color=black}
]
\addplot [semithick, crimson2143940]
table {%
0 3
1 2
2 2
3 1.11111111111111
4 0.8125
5 0.333333333333333
6 0.185185185185186
7 0.0555555555555558
8 0.0308641975308646
9 0.00925925925925952
10 0.00514403292181109
11 0.00154320987654355
12 0.000857338820302145
13 0.000257201646090888
14 0.00014288980338395
15 4.2866941015296e-05
16 2.38149672309174e-05
17 7.14449016969709e-06
18 3.96916120548596e-06
19 1.19074836191224e-06
20 6.61526867951068e-07
21 1.98458060429729e-07
22 1.10254478213889e-07
23 3.30763438860515e-08
24 1.83757467020484e-08
25 5.51272427706806e-09
26 3.06262459837114e-09
27 9.18787712578251e-10
28 5.10437692113896e-10
29 1.53131507474313e-10
30 8.50730597079519e-11
31 2.55222509792929e-11
32 1.41793243813026e-11
33 4.25415258575867e-12
34 2.36344277482203e-12
35 7.0921046813055e-13
36 3.94129173741931e-13
37 1.18571819029967e-13
38 6.61692922676593e-14
39 1.99840144432528e-14
40 1.13242748511766e-14
41 3.5527136788005e-15
42 2.22044604925031e-15
43 8.88178419700125e-16
44 6.66133814775094e-16
45 4.44089209850063e-16
46 4.44089209850063e-16
47 4.44089209850063e-16
48 4.44089209850063e-16
49 4.44089209850063e-16
50 4.44089209850063e-16
51 4.44089209850063e-16
52 4.44089209850063e-16
53 4.44089209850063e-16
54 4.44089209850063e-16
55 4.44089209850063e-16
56 4.44089209850063e-16
57 4.44089209850063e-16
58 4.44089209850063e-16
59 4.44089209850063e-16
60 4.44089209850063e-16
61 4.44089209850063e-16
62 4.44089209850063e-16
63 4.44089209850063e-16
64 4.44089209850063e-16
65 4.44089209850063e-16
66 4.44089209850063e-16
67 4.44089209850063e-16
68 4.44089209850063e-16
69 4.44089209850063e-16
70 4.44089209850063e-16
71 4.44089209850063e-16
72 4.44089209850063e-16
73 4.44089209850063e-16
74 4.44089209850063e-16
75 4.44089209850063e-16
76 4.44089209850063e-16
77 4.44089209850063e-16
78 4.44089209850063e-16
79 4.44089209850063e-16
80 4.44089209850063e-16
81 4.44089209850063e-16
82 4.44089209850063e-16
83 4.44089209850063e-16
84 4.44089209850063e-16
85 4.44089209850063e-16
86 4.44089209850063e-16
87 4.44089209850063e-16
88 4.44089209850063e-16
89 4.44089209850063e-16
90 4.44089209850063e-16
91 4.44089209850063e-16
92 4.44089209850063e-16
93 4.44089209850063e-16
94 4.44089209850063e-16
95 4.44089209850063e-16
96 4.44089209850063e-16
97 4.44089209850063e-16
98 4.44089209850063e-16
99 4.44089209850063e-16
100 4.44089209850063e-16
};
\addlegendentry{$\alpha=0,\beta=0$}
\addplot [semithick, darkturquoise23190207]
table {%
0 3
1 2
2 2
3 1.55555555555556
4 1.15555555555556
5 1.0775462962963
6 1.02876984126984
7 0.967592592592593
8 0.860768175582991
9 0.773456790123457
10 0.703236064347176
11 0.644461591220851
12 0.594900812493405
13 0.552383401920439
14 0.515559746989788
15 0.483333690557842
16 0.454902577171701
17 0.429629682551779
18 0.407017635777067
19 0.386666674604989
20 0.368253982114532
21 0.351515152717928
22 0.336231886167188
23 0.32222222240598
24 0.309333333656747
25 0.297435897464168
26 0.286419753136329
27 0.276190476194851
28 0.266666666674411
29 0.257777777778458
30 0.249462365592605
31 0.241666666666773
32 0.234343434343624
33 0.227450980392174
34 0.220952380952411
35 0.214814814814818
36 0.209009009009014
37 0.203508771929825
38 0.198290598290599
39 0.193333333333334
40 0.188617886178862
41 0.184126984126984
42 0.179844961240311
43 0.175757575757576
44 0.171851851851852
45 0.168115942028986
46 0.164539007092199
47 0.161111111111112
48 0.157823129251701
49 0.154666666666667
50 0.151633986928105
51 0.148717948717949
52 0.145911949685535
53 0.14320987654321
54 0.140606060606061
55 0.138095238095238
56 0.135672514619884
57 0.133333333333334
58 0.131073446327684
59 0.128888888888889
60 0.126775956284153
61 0.124731182795699
62 0.122751322751323
63 0.120833333333334
64 0.118974358974359
65 0.117171717171717
66 0.11542288557214
67 0.113725490196079
68 0.11207729468599
69 0.110476190476191
70 0.108920187793427
71 0.107407407407407
72 0.105936073059361
73 0.104504504504505
74 0.103111111111112
75 0.101754385964912
76 0.100432900432901
77 0.0991452991452992
78 0.0978902953586498
79 0.0966666666666669
80 0.095473251028807
81 0.0943089430894311
82 0.0931726907630526
83 0.0920634920634922
84 0.0909803921568628
85 0.0899224806201553
86 0.0888888888888888
87 0.0878787878787881
88 0.0868913857677907
89 0.0859259259259262
90 0.0849816849816851
91 0.0840579710144933
92 0.0831541218637997
93 0.0822695035460996
94 0.0814035087719305
95 0.0805555555555559
96 0.0797250859106531
97 0.0789115646258507
98 0.0781144781144782
99 0.0773333333333335
100 0.0765676567656768
};
\addlegendentry{$\alpha=0,\beta_n=1/n$}
\addplot [semithick, darkorange25512714]
table {%
0 3
1 2.33333333333333
2 2
3 1.77777777777778
4 1.2962962962963
5 0.916724537037037
6 0.804241071428571
7 0.468354208002645
8 0.238316219727775
9 0.114275570385476
10 0.0532450516750769
11 0.0244230688457252
12 0.0110896906214162
13 0.00499686533863519
14 0.00223716789815986
15 0.000996066665669204
16 0.000441315312780688
17 0.000194674475104017
18 8.55386516538381e-05
19 3.74518825925207e-05
20 1.63450133912058e-05
21 7.11246782181263e-06
22 3.08664510884071e-06
23 1.3362227895275e-06
24 5.77139193991272e-07
25 2.48751579556838e-07
26 1.07004233473873e-07
27 4.59457198953572e-08
28 1.96947953501336e-08
29 8.42884695373414e-09
30 3.60195873128077e-09
31 1.53710177919208e-09
32 6.5508287683258e-10
33 2.78837619660521e-10
34 1.18549614569474e-10
35 5.03468378099114e-11
36 2.13593587261585e-11
37 9.0525364981886e-12
38 3.83337805942574e-12
39 1.62225788358228e-12
40 6.85451695403572e-13
41 2.89990254032091e-13
42 1.22790666523542e-13
43 5.19584375524573e-14
44 2.19824158875781e-14
45 9.54791801177635e-15
46 4.66293670342566e-15
47 2.66453525910038e-15
48 1.11022302462516e-15
49 6.66133814775094e-16
50 8.88178419700125e-16
51 8.88178419700125e-16
52 4.44089209850063e-16
53 4.44089209850063e-16
54 4.44089209850063e-16
55 8.88178419700125e-16
56 6.66133814775094e-16
57 8.88178419700125e-16
58 6.66133814775094e-16
59 8.88178419700125e-16
60 8.88178419700125e-16
61 6.66133814775094e-16
62 6.66133814775094e-16
63 4.44089209850063e-16
64 4.44089209850063e-16
65 4.44089209850063e-16
66 4.44089209850063e-16
67 4.44089209850063e-16
68 4.44089209850063e-16
69 4.44089209850063e-16
70 4.44089209850063e-16
71 4.44089209850063e-16
72 4.44089209850063e-16
73 4.44089209850063e-16
74 4.44089209850063e-16
75 8.88178419700125e-16
76 4.44089209850063e-16
77 4.44089209850063e-16
78 4.44089209850063e-16
79 4.44089209850063e-16
80 4.44089209850063e-16
81 4.44089209850063e-16
82 4.44089209850063e-16
83 4.44089209850063e-16
84 4.44089209850063e-16
85 4.44089209850063e-16
86 4.44089209850063e-16
87 4.44089209850063e-16
88 4.44089209850063e-16
89 4.44089209850063e-16
90 4.44089209850063e-16
91 4.44089209850063e-16
92 4.44089209850063e-16
93 4.44089209850063e-16
94 4.44089209850063e-16
95 4.44089209850063e-16
96 4.44089209850063e-16
97 4.44089209850063e-16
98 4.44089209850063e-16
99 4.44089209850063e-16
100 4.44089209850063e-16
};
\addlegendentry{$\alpha_n=1/n,\beta=0$}
\addplot [semithick, forestgreen4416044]
table {%
0 3
1 2.33333333333333
2 2
3 1.88888888888889
4 1.67361111111111
5 1.40509259259259
6 1.13107638888889
7 0.970015914351852
8 0.90654613353588
9 0.788698172863619
10 0.644869517887571
11 0.506926956490725
12 0.387879804655685
13 0.291073240533942
14 0.215290472068141
15 0.157498669068346
16 0.114250429736607
17 0.0823359694754893
18 0.0590335404746316
19 0.0421570002395244
20 0.0300109384943859
21 0.0213119968201685
22 0.0151055219827583
23 0.0106904777018921
24 0.00755704436742466
25 0.00533719634022112
26 0.00376677091117195
27 0.00265698590918007
28 0.00187338606548049
29 0.00132046015369802
30 0.000930498909682997
31 0.000655577721312861
32 0.000461816759962819
33 0.000325287467489499
34 0.000229101842694757
35 0.000161347493268238
36 0.000113625276942919
37 8.00150874986638e-05
38 5.63451957260597e-05
39 3.96764602557109e-05
40 2.79384427144258e-05
41 1.96728101435895e-05
42 1.38524568453757e-05
43 9.75403625891325e-06
44 6.86815042039157e-06
45 4.83608184431716e-06
46 3.40522882802397e-06
47 2.39771785937037e-06
48 1.68829835844875e-06
49 1.18877545673968e-06
50 8.37047492208143e-07
51 5.89386386140589e-07
52 4.15001724540787e-07
53 2.92213014674303e-07
54 2.05754384863255e-07
55 1.44876705965302e-07
56 1.02011226132959e-07
57 7.18285877443492e-08
58 5.05762534075416e-08
59 3.56119667088706e-08
60 2.50752485442263e-08
61 1.76560894793454e-08
62 1.24320798100541e-08
63 8.75372796471652e-09
64 6.1637115411628e-09
65 4.34001834470621e-09
66 3.05591196791966e-09
67 2.15174145168362e-09
68 1.51509316204113e-09
69 1.06681374845152e-09
70 7.51169348944813e-10
71 5.28916466180362e-10
72 3.72422981342879e-10
73 2.62232013881203e-10
74 1.84643855760669e-10
75 1.30012223209519e-10
76 9.15449938077018e-11
77 6.44591047205267e-11
78 4.53872495143059e-11
79 3.19584358976499e-11
80 2.25028884415224e-11
81 1.58448809628453e-11
82 1.11568532190631e-11
83 7.85593812224761e-12
84 5.53157519789238e-12
85 3.89488441498997e-12
86 2.74247291542906e-12
87 1.93089988442807e-12
88 1.35980116056089e-12
89 9.57456336436735e-13
90 6.7434946515732e-13
91 4.74953409934642e-13
92 3.34399175017097e-13
93 2.35589325825458e-13
94 1.66089364483923e-13
95 1.17017506795491e-13
96 8.26005930321116e-14
97 5.81756864903582e-14
98 4.10782519111308e-14
99 2.90878432451791e-14
100 2.04281036531029e-14
};
\addlegendentry{$\alpha=1/2,\beta=0$}
\addplot [semithick, steelblue31119180]
table {%
0 3
1 2.33333333333333
2 2
3 1.88888888888889
4 1.64740740740741
5 1.36512345679012
6 1.13527966742252
7 1.07019955804949
8 1.02344505699942
9 0.936709830997085
10 0.831917024732572
11 0.746314515854769
12 0.677311391553754
13 0.620598303618043
14 0.573062470630432
15 0.532556123562535
16 0.497567656521945
17 0.467003065203881
18 0.44004898108099
19 0.416086019907455
20 0.394632655171388
21 0.375307854015194
22 0.357805508015107
23 0.341876474956562
24 0.327315665837985
25 0.31395256413419
26 0.301644136648554
27 0.290269447171065
28 0.27972550599323
29 0.269924031697842
30 0.260788896524312
31 0.252254090802737
32 0.244262086257009
33 0.236762509115125
34 0.229711056203755
35 0.223068603319141
36 0.216800466997022
37 0.210875789592738
38 0.205267024181741
39 0.199949500794894
40 0.194901059332276
41 0.190101737454184
42 0.185533504046277
43 0.181180030656294
44 0.177026494720093
45 0.173059409522329
46 0.169266476737727
47 0.165636458122446
48 0.162159063509407
49 0.158824852735899
50 0.155625149518889
51 0.152551965610686
52 0.149597933828883
53 0.146756248770392
54 0.144020614198755
55 0.141385196243259
56 0.138844581673396
57 0.136393740617115
58 0.134027993179659
59 0.131742979494489
60 0.129534632801086
61 0.127399155198247
62 0.125332995767415
63 0.123332830799844
64 0.121395545895038
65 0.119518219726902
66 0.117698109298933
67 0.115932636531387
68 0.11421937604202
69 0.112556043998181
70 0.110940487932183
71 0.109370677424143
72 0.107844695567246
73 0.106360731139768
74 0.104917071416478
75 0.103512095559271
76 0.102144268533246
77 0.100812135500139
78 0.0995143166459298
79 0.0982495024038694
80 0.0970164490381111
81 0.0958139745565403
82 0.094640954924528
83 0.0934963205540436
84 0.0923790530450563
85 0.0912881821582887
86 0.0902227830003899
87 0.0891819734043364
88 0.0881649114894201
89 0.0871707933866344
90 0.0861988511165086
91 0.0852483506075989
92 0.0843185898448762
93 0.083408897138179
94 0.0825186295017439
95 0.0816471711365792
96 0.0807939320081543
97 0.0799583465124831
98 0.0791398722242516
99 0.078337988721179
100 0.0775521964792263
};
\addlegendentry{$\alpha_n=1/n,\beta_n=1/n$}
\addplot [semithick, mediumpurple148103189]
table {%
0 3
1 2.33333333333333
2 2
3 1.94444444444444
4 1.84722222222222
5 1.72878086419753
6 1.60424933862434
7 1.48277150848765
8 1.36912454989712
9 1.26533524627058
10 1.20820168988606
11 1.17272506263848
12 1.14196706761586
13 1.11511970636262
14 1.09153433873003
15 1.05301037002534
16 1.00628765279432
17 0.95839847778136
18 0.912307063590082
19 0.869128995827632
20 0.829151305593819
21 0.792302581122742
22 0.758369086429902
23 0.727092958305409
24 0.698215515974932
25 0.671495072776599
26 0.646713020738012
27 0.623674657927137
28 0.602207794439475
29 0.582160549608942
30 0.56339898334112
31 0.545804839654002
32 0.529273508761723
33 0.513712234691676
34 0.499038560338212
35 0.485178987862015
36 0.472067828246454
37 0.459646214198149
38 0.447861252810001
39 0.436665297242326
40 0.426015319544081
41 0.415872369382844
42 0.406201105790239
43 0.396969391045331
44 0.38814793753241
45 0.379709999854813
46 0.371631105699405
47 0.363888819961799
48 0.356462537491611
49 0.349333300527267
50 0.342483637484192
51 0.335897420258164
52 0.329559737623478
53 0.323456782656852
54 0.317575752413935
55 0.311904758335104
56 0.306432746068627
57 0.301149423578582
58 0.296045196557317
59 0.291111110292157
60 0.286338797247016
61 0.281720429714595
62 0.277248676976396
63 0.272916666477943
64 0.268717948587108
65 0.264646464555732
66 0.260696517350002
67 0.256862745054378
68 0.25314009658806
69 0.249523809502781
70 0.246009389656764
71 0.242592592582457
72 0.239269406385655
73 0.236036036031147
74 0.232888888885492
75 0.229824561401149
76 0.226839826838187
77 0.223931623930484
78 0.22109704641271
79 0.218333333332782
80 0.215637860081921
81 0.213008130081034
82 0.210441767068088
83 0.207936507936379
84 0.205490196078342
85 0.203100775193736
86 0.200766283524861
87 0.198484848484819
88 0.19625468164792
89 0.19407407407406
90 0.191941391941382
91 0.189855072463761
92 0.18781362007168
93 0.185815602836876
94 0.183859649122805
95 0.181944444444443
96 0.180068728522336
97 0.178231292517006
98 0.176430976430976
99 0.174666666666667
100 0.172937293729373
};
\addlegendentry{$\alpha=1/2,\beta_n=1/n$}
\end{axis}

\end{tikzpicture}
\clearpage{}%

    \caption{Generated sequences $(x_n^{\mathcal{F}})$ with
    initial value $v_0=0$.}
    \label{fig:mdp-example-exact-results}
  \end{subfigure}
  \hfill
  \begin{subfigure}[b]{.45\textwidth}
    \centering
\clearpage{}%
\begin{tikzpicture}[scale=0.7]

\definecolor{crimson2143940}{RGB}{214,39,40}
\definecolor{darkgray176}{RGB}{176,176,176}
\definecolor{darkorange25512714}{RGB}{255,127,14}
\definecolor{darkturquoise23190207}{RGB}{23,190,207}
\definecolor{forestgreen4416044}{RGB}{44,160,44}
\definecolor{lightgray204}{RGB}{204,204,204}
\definecolor{mediumpurple148103189}{RGB}{148,103,189}
\definecolor{steelblue31119180}{RGB}{31,119,180}

\begin{axis}[
legend cell align={left},
legend style={
  fill opacity=0.8,
  draw opacity=1,
  text opacity=1,
  at={(0.95,0.4)},
  anchor=east,
  draw=lightgray204
},
tick align=outside,
tick pos=left,
x grid style={darkgray176},
xlabel={\(\displaystyle n\)},
xmin=-5, xmax=105,
xtick style={color=black},
y grid style={darkgray176},
ylabel={\(\displaystyle \|x_n^{\mathcal{F}}-\mu f\|\)},
ymin=-0.319603960396039, ymax=8.39617161716172,
ytick style={color=black}
]
\addplot [semithick, crimson2143940]
table {%
0 8
1 8
2 8
3 8
4 8
5 8
6 8
7 8
8 8
9 8
10 8
11 8
12 8
13 8
14 8
15 8
16 8
17 8
18 8
19 8
20 8
21 8
22 8
23 8
24 8
25 8
26 8
27 8
28 8
29 8
30 8
31 8
32 8
33 8
34 8
35 8
36 8
37 8
38 8
39 8
40 8
41 8
42 8
43 8
44 8
45 8
46 8
47 8
48 8
49 8
50 8
51 8
52 8
53 8
54 8
55 8
56 8
57 8
58 8
59 8
60 8
61 8
62 8
63 8
64 8
65 8
66 8
67 8
68 8
69 8
70 8
71 8
72 8
73 8
74 8
75 8
76 8
77 8
78 8
79 8
80 8
81 8
82 8
83 8
84 8
85 8
86 8
87 8
88 8
89 8
90 8
91 8
92 8
93 8
94 8
95 8
96 8
97 8
98 8
99 8
100 8
};
\addlegendentry{$\alpha=0,\beta=0$}
\addplot [semithick, darkturquoise23190207]
table {%
0 8
1 3
2 1.33333333333333
3 0.75
4 1
5 0.944444444444444
6 1.02728174603175
7 0.75
8 0.858024691358025
9 0.772530864197531
10 0.702861952861953
11 0.644332990397806
12 0.594848053181387
13 0.55236503037429
14 0.515552126200274
15 0.483331011374028
16 0.454901456467361
17 0.429629285635659
18 0.40701746865449
19 0.386666615067571
20 0.368253956913508
21 0.351515143697107
22 0.336231882332249
23 0.322222221027799
24 0.309333333068723
25 0.29743589725214
26 0.286419753045585
27 0.276190476162038
28 0.26666666666033
29 0.257777777773354
30 0.24946236559041
31 0.241666666665976
32 0.23434343434328
33 0.227450980392049
34 0.220952380952357
35 0.214814814814798
36 0.209009009009005
37 0.203508771929822
38 0.198290598290598
39 0.193333333333333
40 0.188617886178862
41 0.184126984126984
42 0.179844961240311
43 0.175757575757576
44 0.171851851851852
45 0.168115942028986
46 0.164539007092199
47 0.161111111111112
48 0.157823129251701
49 0.154666666666667
50 0.151633986928105
51 0.148717948717949
52 0.145911949685535
53 0.14320987654321
54 0.140606060606061
55 0.138095238095238
56 0.135672514619884
57 0.133333333333334
58 0.131073446327684
59 0.128888888888889
60 0.126775956284153
61 0.124731182795699
62 0.122751322751323
63 0.120833333333334
64 0.118974358974359
65 0.117171717171717
66 0.11542288557214
67 0.113725490196079
68 0.11207729468599
69 0.110476190476191
70 0.108920187793427
71 0.107407407407407
72 0.105936073059361
73 0.104504504504505
74 0.103111111111112
75 0.101754385964912
76 0.100432900432901
77 0.0991452991452992
78 0.0978902953586498
79 0.0966666666666669
80 0.095473251028807
81 0.0943089430894311
82 0.0931726907630526
83 0.0920634920634922
84 0.0909803921568628
85 0.0899224806201553
86 0.0888888888888888
87 0.0878787878787881
88 0.0868913857677907
89 0.0859259259259262
90 0.0849816849816851
91 0.0840579710144933
92 0.0831541218637997
93 0.0822695035460996
94 0.0814035087719305
95 0.0805555555555559
96 0.0797250859106531
97 0.0789115646258507
98 0.0781144781144782
99 0.0773333333333335
100 0.0765676567656768
};
\addlegendentry{$\alpha=0,\beta_n=1/n$}
\addplot [semithick, darkorange25512714]
table {%
0 8
1 5.5
2 5.5
3 5.5
4 5.5
5 5.5
6 5.5
7 5.5
8 5.5
9 5.5
10 5.5
11 5.5
12 5.5
13 5.5
14 5.5
15 5.5
16 5.5
17 5.5
18 5.5
19 5.5
20 5.5
21 5.5
22 5.5
23 5.5
24 5.5
25 5.5
26 5.5
27 5.5
28 5.5
29 5.5
30 5.5
31 5.5
32 5.5
33 5.5
34 5.5
35 5.5
36 5.5
37 5.5
38 5.5
39 5.5
40 5.5
41 5.5
42 5.5
43 5.5
44 5.5
45 5.5
46 5.5
47 5.5
48 5.5
49 5.5
50 5.5
51 5.5
52 5.5
53 5.5
54 5.5
55 5.5
56 5.5
57 5.5
58 5.5
59 5.5
60 5.5
61 5.5
62 5.5
63 5.5
64 5.5
65 5.5
66 5.5
67 5.5
68 5.5
69 5.5
70 5.5
71 5.5
72 5.5
73 5.5
74 5.5
75 5.5
76 5.5
77 5.5
78 5.5
79 5.5
80 5.5
81 5.5
82 5.5
83 5.5
84 5.5
85 5.5
86 5.5
87 5.5
88 5.5
89 5.5
90 5.5
91 5.5
92 5.5
93 5.5
94 5.5
95 5.5
96 5.5
97 5.5
98 5.5
99 5.5
100 5.5
};
\addlegendentry{$\alpha_n=1/n,\beta=0$}
\addplot [semithick, forestgreen4416044]
table {%
0 8
1 5.5
2 5.5
3 5.5
4 5.5
5 5.5
6 5.5
7 5.5
8 5.5
9 5.5
10 5.5
11 5.5
12 5.5
13 5.5
14 5.5
15 5.5
16 5.5
17 5.5
18 5.5
19 5.5
20 5.5
21 5.5
22 5.5
23 5.5
24 5.5
25 5.5
26 5.5
27 5.5
28 5.5
29 5.5
30 5.5
31 5.5
32 5.5
33 5.5
34 5.5
35 5.5
36 5.5
37 5.5
38 5.5
39 5.5
40 5.5
41 5.5
42 5.5
43 5.5
44 5.5
45 5.5
46 5.5
47 5.5
48 5.5
49 5.5
50 5.5
51 5.5
52 5.5
53 5.5
54 5.5
55 5.5
56 5.5
57 5.5
58 5.5
59 5.5
60 5.5
61 5.5
62 5.5
63 5.5
64 5.5
65 5.5
66 5.5
67 5.5
68 5.5
69 5.5
70 5.5
71 5.5
72 5.5
73 5.5
74 5.5
75 5.5
76 5.5
77 5.5
78 5.5
79 5.5
80 5.5
81 5.5
82 5.5
83 5.5
84 5.5
85 5.5
86 5.5
87 5.5
88 5.5
89 5.5
90 5.5
91 5.5
92 5.5
93 5.5
94 5.5
95 5.5
96 5.5
97 5.5
98 5.5
99 5.5
100 5.5
};
\addlegendentry{$\alpha=1/2,\beta=0$}
\addplot [semithick, steelblue31119180]
table {%
0 8
1 1.75
2 1.25925925925926
3 1.03298611111111
4 0.8175
5 0.750231481481481
6 0.928599773242629
7 1.0443105546255
8 1.01946190804497
9 0.925426303499473
10 0.827174238066483
11 0.744439267092074
12 0.676547561412094
13 0.620281591488251
14 0.572930434734912
15 0.532501039297934
16 0.497544689016753
17 0.466993496742425
18 0.440044998056986
19 0.416084363193113
20 0.394631966566858
21 0.375307567991797
22 0.357805389284481
23 0.341876425699154
24 0.327315645413813
25 0.313952555669785
26 0.301644133142318
27 0.290269445719317
28 0.279725505392392
29 0.269924031449271
30 0.260788896421516
31 0.25225409076024
32 0.244262086239446
33 0.23676250910787
34 0.229711056200758
35 0.223068603317903
36 0.21680046699651
37 0.210875789592528
38 0.205267024181654
39 0.199949500794858
40 0.19490105933226
41 0.190101737454178
42 0.185533504046275
43 0.181180030656293
44 0.177026494720093
45 0.173059409522328
46 0.169266476737727
47 0.165636458122446
48 0.162159063509407
49 0.158824852735899
50 0.155625149518889
51 0.152551965610686
52 0.149597933828883
53 0.146756248770392
54 0.144020614198755
55 0.141385196243259
56 0.138844581673396
57 0.136393740617115
58 0.134027993179659
59 0.131742979494489
60 0.129534632801086
61 0.127399155198247
62 0.125332995767415
63 0.123332830799844
64 0.121395545895038
65 0.119518219726902
66 0.117698109298933
67 0.115932636531387
68 0.11421937604202
69 0.112556043998181
70 0.110940487932183
71 0.109370677424143
72 0.107844695567246
73 0.106360731139768
74 0.104917071416478
75 0.103512095559271
76 0.102144268533246
77 0.100812135500139
78 0.0995143166459298
79 0.0982495024038694
80 0.0970164490381111
81 0.0958139745565403
82 0.094640954924528
83 0.0934963205540436
84 0.0923790530450563
85 0.0912881821582887
86 0.0902227830003899
87 0.0891819734043364
88 0.0881649114894201
89 0.0871707933866344
90 0.0861988511165086
91 0.0852483506075989
92 0.0843185898448762
93 0.083408897138179
94 0.0825186295017439
95 0.0816471711365792
96 0.0807939320081543
97 0.0799583465124831
98 0.0791398722242516
99 0.078337988721179
100 0.0775521964792263
};
\addlegendentry{$\alpha_n=1/n,\beta_n=1/n$}
\addplot [semithick, mediumpurple148103189]
table {%
0 8
1 1.75
2 1.40277777777778
3 1.34895833333333
4 1.24548611111111
5 1.13527199074074
6 1.03179770171958
7 1.0634765625
8 1.16351217889982
9 1.18606510868779
10 1.17750969746565
11 1.15744018162229
12 1.1341248751064
13 1.11096549202315
14 1.08616734931171
15 1.04586324609124
16 1.00010967657197
17 0.953745085587005
18 0.909016453377644
19 0.866876816224951
20 0.827636268159221
21 0.791292176053741
22 0.757697573263796
23 0.726646840449096
24 0.697918706084675
25 0.671297110428722
26 0.646580594936979
27 0.623585797061148
28 0.602147985319519
29 0.582120179286715
30 0.563371662868678
31 0.545786307326993
32 0.52926091156617
33 0.513703656136408
34 0.499032708993163
35 0.485174991037207
36 0.472065094738204
37 0.459644342610624
38 0.447859970088497
39 0.436664417322242
40 0.426014715448369
41 0.415871954343972
42 0.406200820448094
43 0.396969194747829
44 0.388147802413496
45 0.379709906796856
46 0.371631041576601
47 0.363888775755867
48 0.356462507002376
49 0.349333279489458
50 0.342483622961953
51 0.335897410229625
52 0.329559730695516
53 0.323456777869116
54 0.317575749104107
55 0.31190475604621
56 0.306432744485246
57 0.301149422482913
58 0.296045195798909
59 0.291111109767045
60 0.286338796883334
61 0.281720429462649
62 0.277248676801811
63 0.272916666356934
64 0.268717948503214
65 0.264646464497555
66 0.26069651730965
67 0.256862745026383
68 0.253140096568634
69 0.249523809489298
70 0.246009389647404
71 0.242592592575957
72 0.239269406381141
73 0.236036036028012
74 0.232888888883314
75 0.229824561399635
76 0.226839826837135
77 0.223931623929753
78 0.221097046412201
79 0.218333333332429
80 0.215637860081676
81 0.213008130080863
82 0.210441767067969
83 0.207936507936296
84 0.205490196078284
85 0.203100775193696
86 0.200766283524833
87 0.198484848484799
88 0.196254681647906
89 0.19407407407405
90 0.191941391941375
91 0.189855072463756
92 0.187813620071676
93 0.185815602836874
94 0.183859649122803
95 0.181944444444442
96 0.180068728522335
97 0.178231292517006
98 0.176430976430976
99 0.174666666666667
100 0.172937293729373
};
\addlegendentry{$\alpha=1/2,\beta_n=1/n$}
\end{axis}

\end{tikzpicture}
\clearpage{}%

    \caption{Generated sequences $(x_n^{\mathcal{F}})$ with
    initial value $v_0=(10,5,4,3,2,1,0)$.}
    \label{fig:mdp-example-exact-results-overapproximation}
  \end{subfigure}
  \caption{Results of different iteration schemes for the exact Bellman operator
  of the MDP in \Cref{fig:mdp-num-example}.}
\end{figure}

Unsurprisingly, this happens also when considering more
complex functions.  Note that the Markov decision process in
\Cref{fig:mdp-num-example} is not simple whence its Bellman operator
is not a power contraction.  Here, starting from an initial point that
overestimates the state-values in end-components, no undampened
iteration converges to the least fixpoint.  Indeed, the Kleene
iteration gets stuck in a loop while the undampened Mann iterations
converge to an overapproximation of the true optimal value
function.  The dampened
versions, however, do converge to the least fixpoint (see
\Cref{fig:mdp-example-exact-results-overapproximation}).

\paragraph{Unknown function.} Now, the main focus of this paper
was a setting in which we do not have access to the exact function.
In this approximative setting, we already mentioned in
\Cref{se:approximation} that the considered undampened iterations can only
guarantee convergence to the least fixpoint with resetting, i.e.,
starting the iteration anew from the bottom $0$ with each new
approximation $f_n$ (with carefully chosen number of iterations if
necessary).
We see this in \Cref{fig:dim2-values} for the
relatively simple monotone non-expansive function sequence
\begin{equation}
  \label{eq:dim2-function-sequence}
  f_n\colon [0,1]^2\to [0,1]^2, x\mapsto\max(x, (1-a)\cdot x^n + a)
\end{equation}
where the least fixpoint $\mu f=a$ of the limit function
$f(x)=\max(x, a)$ was chosen to be $(3/4,1/2)$.  It is clearly visible
that undampened iteration without resetting cannot reach the least
fixpoint, even when starting from $0$ (red and orange lines in
\Cref{fig:dim2-values-undamp} with Kleene iteration converging to
$\sim(0.95,0.5)$ and Mann iteration to $\sim(0.9,0.35)$).  Versions
with resetting as well as dampened versions, on the other hand, do
converge to the least fixpoint, the latter even when starting from an
overapproximation such as $(7/8,7/8)$
(cf. \Cref{fig:dim2-values-damp}).  This also serves as a visual
representation of how the dampened versions approach the fixpoint.

\begin{figure}
  \centering
  \begin{subfigure}[t]{.45\textwidth}
    \centering
\clearpage{}%
\begin{tikzpicture}[scale=0.7]

\definecolor{crimson2143940}{RGB}{214,39,40}
\definecolor{darkgray176}{RGB}{176,176,176}
\definecolor{darkorange25512714}{RGB}{255,127,14}
\definecolor{darkturquoise23190207}{RGB}{23,190,207}
\definecolor{forestgreen4416044}{RGB}{44,160,44}
\definecolor{lightgray204}{RGB}{204,204,204}

\begin{axis}[
legend cell align={left},
legend style={fill opacity=0.8, draw opacity=1, text opacity=1, draw=lightgray204},
tick align=outside,
tick pos=left,
x grid style={darkgray176},
xlabel={\(\displaystyle x\)},
xmin=0, xmax=1,
xtick style={color=black},
y grid style={darkgray176},
ylabel={\(\displaystyle y\)},
ymin=0, ymax=1,
ytick style={color=black}
]
\addplot [semithick, crimson2143940, mark=x, mark size=3, mark options={solid}]
table {%
0 0
0.75 0.25
0.912379763209582 0.4375
0.953392627901275 0.467034287235767
0.966648345593802 0.467034287235767
0.970138969236968 0.467034287235767
0.970138969236968 0.467034287235767
0.970138969236968 0.467034287235767
0.970138969236968 0.467034287235767
0.970138969236968 0.467034287235767
0.970138969236968 0.467034287235767
0.970138969236968 0.467034287235767
0.970138969236968 0.467034287235767
0.970138969236968 0.467034287235767
0.970138969236968 0.467034287235767
0.970138969236968 0.467034287235767
0.970138969236968 0.467034287235767
0.970138969236968 0.467034287235767
0.970138969236968 0.467034287235767
0.970138969236968 0.467034287235767
0.970138969236968 0.467034287235767
0.970138969236968 0.467034287235767
0.970138969236968 0.467034287235767
0.970138969236968 0.467034287235767
0.970138969236968 0.467034287235767
0.970138969236968 0.467034287235767
0.970138969236968 0.467034287235767
0.970138969236968 0.467034287235767
0.970138969236968 0.467034287235767
0.970138969236968 0.467034287235767
0.970138969236968 0.467034287235767
0.970138969236968 0.467034287235767
0.970138969236968 0.467034287235767
0.970138969236968 0.467034287235767
0.970138969236968 0.467034287235767
0.970138969236968 0.467034287235767
0.970138969236968 0.467034287235767
0.970138969236968 0.467034287235767
0.970138969236968 0.467034287235767
0.970138969236968 0.467034287235767
0.970138969236968 0.467034287235767
0.970138969236968 0.467034287235767
0.970138969236968 0.467034287235767
0.970138969236968 0.467034287235767
0.970138969236968 0.467034287235767
0.970138969236968 0.467034287235767
0.970138969236968 0.467034287235767
0.970138969236968 0.467034287235767
0.970138969236968 0.467034287235767
0.970138969236968 0.467034287235767
0.970138969236968 0.467034287235767
0.970138969236968 0.467034287235767
0.970138969236968 0.467034287235767
0.970138969236968 0.467034287235767
0.970138969236968 0.467034287235767
0.970138969236968 0.467034287235767
0.970138969236968 0.467034287235767
0.970138969236968 0.467034287235767
0.970138969236968 0.467034287235767
0.970138969236968 0.467034287235767
0.970138969236968 0.467034287235767
0.970138969236968 0.467034287235767
0.970138969236968 0.467034287235767
0.970138969236968 0.467034287235767
0.970138969236968 0.467034287235767
0.970138969236968 0.467034287235767
0.970138969236968 0.467034287235767
0.970138969236968 0.467034287235767
0.970138969236968 0.467034287235767
0.970138969236968 0.467034287235767
0.970138969236968 0.467034287235767
0.970138969236968 0.467034287235767
0.970138969236968 0.467034287235767
0.970138969236968 0.467034287235767
0.970138969236968 0.467034287235767
0.970138969236968 0.467034287235767
0.970138969236968 0.467034287235767
0.970138969236968 0.467034287235767
0.970138969236968 0.467034287235767
0.970138969236968 0.467034287235767
0.970138969236968 0.467034287235767
0.970138969236968 0.467034287235767
0.970138969236968 0.467034287235767
0.970138969236968 0.467034287235767
0.970138969236968 0.467034287235767
0.970138969236968 0.467034287235767
0.970138969236968 0.467034287235767
0.970138969236968 0.467034287235767
0.970138969236968 0.467034287235767
0.970138969236968 0.467034287235767
0.970138969236968 0.467034287235767
0.970138969236968 0.467034287235767
0.970138969236968 0.467034287235767
0.970138969236968 0.467034287235767
0.970138969236968 0.467034287235767
0.970138969236968 0.467034287235767
0.970138969236968 0.467034287235767
0.970138969236968 0.467034287235767
0.970138969236968 0.467034287235767
0.970138969236968 0.467034287235767
0.970138969236968 0.467034287235767
};
\addlegendentry{Kleene}
\addplot [semithick, darkturquoise23190207, mark=x, mark size=3, mark options={solid}]
table {%
0 0
0.75 0.25
0.912379763209582 0.4375
0.93792627567541 0.401155976634935
0.936235802908688 0.324939916841686
0.920454567676672 0.281371994677297
0.893892587631446 0.263761399916039
0.86207324671662 0.2564015211342
0.83279721833544 0.253076586539659
0.810401130812052 0.251504945045798
0.794536515424212 0.250743376125315
0.78334686061948 0.250369195311949
0.775303891935865 0.250183915177645
0.769407230358992 0.250091771405018
0.765012785952897 0.250045835148174
0.76169425821729 0.250022903915148
0.759161093629534 0.250011448284956
0.7572102473291 0.250005723159432
0.755696820578188 0.250002861317669
0.754515562815584 0.250001430589243
0.753588864798075 0.250000715276202
0.752858760447447 0.25000035763324
0.752281476972655 0.250000178815341
0.751823646610724 0.250000089407335
0.751459626537308 0.25000004470358
0.751169572861194 0.250000022351767
0.750938037108693 0.250000011175877
0.750752930676517 0.250000005587937
0.750604752277758 0.250000002793968
0.750486005882145 0.250000001396984
0.750390758320805 0.250000000698492
0.750314300421356 0.250000000349246
0.750252885656258 0.250000000174623
0.750203527350241 0.250000000087311
0.750163840487536 0.250000000043656
0.750131917735145 0.250000000021828
0.750106231886977 0.250000000010914
0.750085558827679 0.250000000005457
0.750068916514635 0.250000000002728
0.750055516520426 0.250000000001364
0.750044725463313 0.250000000000682
0.750036034248652 0.250000000000341
0.750029033498904 0.250000000000171
0.750023393899477 0.250000000000085
0.750018850458094 0.250000000000043
0.750015189885664 0.250000000000021
0.75001224047097 0.250000000000011
0.75000986395048 0.250000000000005
0.750007948976176 0.250000000000003
0.750006405864637 0.250000000000001
0.75000516237424 0.250000000000001
0.750004160307961 0.25
0.75000335277954 0.25
0.750002702012926 0.25
0.750002177570599 0.25
0.750001754926945 0.25
0.750001414319301 0.25
0.750001139822524 0.25
0.750000918603519 0.25
0.750000740320684 0.25
0.750000596640058 0.25
0.75000048084549 0.25
0.750000387524524 0.25
0.750000312315292 0.25
0.750000251702569 0.25
0.750000202853416 0.25
0.750000163484744 0.25
0.750000131756584 0.25
0.750000106186077 0.25
0.750000085578162 0.25
0.750000068969715 0.25
0.750000055584538 0.25
0.750000044797072 0.25
0.750000036103169 0.25
0.750000029096521 0.25
0.750000023449676 0.25
0.750000018898731 0.25
0.750000015231002 0.25
0.75000001227508 0.25
0.750000009892822 0.25
0.750000007972896 0.25
0.750000006425575 0.25
0.750000005178547 0.25
0.750000004173533 0.25
0.750000003363564 0.25
0.750000002710789 0.25
0.750000002184699 0.25
0.750000001760709 0.25
0.750000001419003 0.25
0.750000001143614 0.25
0.75000000092167 0.25
0.750000000742799 0.25
0.750000000598642 0.25
0.750000000482462 0.25
0.750000000388829 0.25
0.750000000313368 0.25
0.750000000252552 0.25
0.750000000203539 0.25
0.750000000164037 0.25
0.750000000132202 0.25
0.750000000106545 0.25
};
\addlegendentry{Kleene (reset.)}
\addplot [semithick, darkorange25512714, mark=x, mark size=3, mark options={solid}]
table {%
0 0
0.375 0.125
0.663273277230987 0.270833333333333
0.802758839212155 0.334490534484876
0.864014819672534 0.334490534484876
0.889426520557958 0.334490534484876
0.896388376532584 0.334490534484876
0.896388376532584 0.334490534484876
0.896388376532584 0.334490534484876
0.896388376532584 0.334490534484876
0.896388376532584 0.334490534484876
0.896388376532584 0.334490534484876
0.896388376532584 0.334490534484876
0.896388376532584 0.334490534484876
0.896388376532584 0.334490534484876
0.896388376532584 0.334490534484876
0.896388376532584 0.334490534484876
0.896388376532584 0.334490534484876
0.896388376532584 0.334490534484876
0.896388376532584 0.334490534484876
0.896388376532584 0.334490534484876
0.896388376532584 0.334490534484876
0.896388376532584 0.334490534484876
0.896388376532584 0.334490534484876
0.896388376532584 0.334490534484876
0.896388376532584 0.334490534484876
0.896388376532584 0.334490534484876
0.896388376532584 0.334490534484876
0.896388376532584 0.334490534484876
0.896388376532584 0.334490534484876
0.896388376532584 0.334490534484876
0.896388376532584 0.334490534484876
0.896388376532584 0.334490534484876
0.896388376532584 0.334490534484876
0.896388376532584 0.334490534484876
0.896388376532584 0.334490534484876
0.896388376532584 0.334490534484876
0.896388376532584 0.334490534484876
0.896388376532584 0.334490534484876
0.896388376532584 0.334490534484876
0.896388376532584 0.334490534484876
0.896388376532584 0.334490534484876
0.896388376532584 0.334490534484876
0.896388376532584 0.334490534484876
0.896388376532584 0.334490534484876
0.896388376532585 0.334490534484876
0.896388376532584 0.334490534484876
0.896388376532584 0.334490534484876
0.896388376532584 0.334490534484876
0.896388376532584 0.334490534484876
0.896388376532584 0.334490534484876
0.896388376532584 0.334490534484876
0.896388376532584 0.334490534484876
0.896388376532584 0.334490534484876
0.896388376532584 0.334490534484876
0.896388376532584 0.334490534484876
0.896388376532584 0.334490534484876
0.896388376532584 0.334490534484876
0.896388376532584 0.334490534484876
0.896388376532584 0.334490534484876
0.896388376532584 0.334490534484876
0.896388376532584 0.334490534484876
0.896388376532584 0.334490534484876
0.896388376532584 0.334490534484876
0.896388376532584 0.334490534484876
0.896388376532584 0.334490534484876
0.896388376532584 0.334490534484876
0.896388376532584 0.334490534484876
0.896388376532584 0.334490534484876
0.896388376532584 0.334490534484876
0.896388376532585 0.334490534484876
0.896388376532585 0.334490534484876
0.896388376532584 0.334490534484876
0.896388376532584 0.334490534484876
0.896388376532584 0.334490534484876
0.896388376532584 0.334490534484876
0.896388376532584 0.334490534484876
0.896388376532584 0.334490534484876
0.896388376532584 0.334490534484876
0.896388376532585 0.334490534484876
0.896388376532584 0.334490534484876
0.896388376532584 0.334490534484876
0.896388376532584 0.334490534484876
0.896388376532584 0.334490534484876
0.896388376532584 0.334490534484876
0.896388376532584 0.334490534484876
0.896388376532584 0.334490534484876
0.896388376532584 0.334490534484876
0.896388376532584 0.334490534484876
0.896388376532584 0.334490534484876
0.896388376532584 0.334490534484876
0.896388376532584 0.334490534484876
0.896388376532584 0.334490534484876
0.896388376532585 0.334490534484876
0.896388376532585 0.334490534484876
0.896388376532585 0.334490534484876
0.896388376532585 0.334490534484876
0.896388376532585 0.334490534484876
0.896388376532585 0.334490534484876
0.896388376532585 0.334490534484876
0.896388376532585 0.334490534484876
};
\addlegendentry{Mann}
\addplot [semithick, forestgreen4416044, mark=x, mark size=3, mark options={solid}]
table {%
0 0
0.375 0.125
0.663273277230987 0.270833333333333
0.792836646413868 0.307327925667403
0.844574996645066 0.296598397840474
0.860351109905479 0.276794545004682
0.857390747843526 0.263279779482592
0.843966797159169 0.256361720856604
0.826078284045052 0.253073704124903
0.808583424889604 0.251504749039092
0.794171625417433 0.250743363117507
0.783290816111822 0.25036919445775
0.775297121233116 0.250183915122226
0.76940657314485 0.250091771401495
0.765012733835067 0.250045835147957
0.761694254791917 0.250022903915135
0.759161093440392 0.250011448284955
0.757210247320205 0.250005723159432
0.755696820577827 0.250002861317669
0.754515562815571 0.250001430589243
0.753588864798075 0.250000715276202
0.752858760447447 0.25000035763324
0.752281476972655 0.250000178815341
0.751823646610724 0.250000089407335
0.751459626537308 0.25000004470358
0.751169572861194 0.250000022351767
0.750938037108693 0.250000011175877
0.750752930676517 0.250000005587937
0.750604752277759 0.250000002793968
0.750486005882145 0.250000001396984
0.750390758320805 0.250000000698492
0.750314300421356 0.250000000349246
0.750252885656258 0.250000000174623
0.750203527350242 0.250000000087311
0.750163840487537 0.250000000043656
0.750131917735146 0.250000000021828
0.750106231886977 0.250000000010914
0.750085558827679 0.250000000005457
0.750068916514635 0.250000000002728
0.750055516520427 0.250000000001364
0.750044725463313 0.250000000000682
0.750036034248652 0.250000000000341
0.750029033498905 0.250000000000171
0.750023393899477 0.250000000000085
0.750018850458094 0.250000000000043
0.750015189885664 0.250000000000021
0.750012240470971 0.250000000000011
0.75000986395048 0.250000000000005
0.750007948976176 0.250000000000003
0.750006405864637 0.250000000000001
0.75000516237424 0.250000000000001
0.750004160307961 0.25
0.75000335277954 0.25
0.750002702012926 0.25
0.750002177570599 0.25
0.750001754926945 0.25
0.750001414319301 0.25
0.750001139822524 0.25
0.750000918603519 0.25
0.750000740320684 0.25
0.750000596640058 0.25
0.750000480845491 0.25
0.750000387524524 0.25
0.750000312315292 0.25
0.750000251702569 0.25
0.750000202853416 0.25
0.750000163484744 0.25
0.750000131756584 0.25
0.750000106186078 0.25
0.750000085578162 0.25
0.750000068969716 0.25
0.750000055584538 0.25
0.750000044797073 0.25
0.750000036103169 0.25
0.750000029096521 0.25
0.750000023449676 0.25
0.750000018898731 0.25
0.750000015231002 0.25
0.75000001227508 0.25
0.750000009892822 0.25
0.750000007972897 0.25
0.750000006425575 0.25
0.750000005178547 0.25
0.750000004173533 0.25
0.750000003363565 0.25
0.750000002710789 0.25
0.750000002184699 0.25
0.750000001760709 0.25
0.750000001419004 0.25
0.750000001143614 0.25
0.75000000092167 0.25
0.7500000007428 0.25
0.750000000598643 0.25
0.750000000482462 0.25
0.750000000388829 0.25
0.750000000313369 0.25
0.750000000252553 0.25
0.750000000203539 0.25
0.750000000164038 0.25
0.750000000132203 0.25
0.750000000106546 0.25
};
\addlegendentry{Mann (reset)}
\addplot [semithick, black, mark=x, mark size=3, mark options={solid}, only marks, forget plot]
table {%
0.75 0.25
};
\end{axis}

\end{tikzpicture}
\clearpage{}%

    \caption{Generated sequences $(x_n^{\mathcal{F}})$ on the function
    sequence~\eqref{eq:dim2-function-sequence} without dampening.}
    \label{fig:dim2-values-undamp}
  \end{subfigure}
  \hfill
  \begin{subfigure}[t]{.45\textwidth}
    \centering
\clearpage{}%
\begin{tikzpicture}[scale=0.7]

\definecolor{crimson2143940}{RGB}{214,39,40}
\definecolor{violet28210083}{RGB}{148,0,211}
\definecolor{darkgray176}{RGB}{176,176,176}
\definecolor{darkorange25512714}{RGB}{255,127,14}
\definecolor{darkturquoise23190207}{RGB}{23,190,207}
\definecolor{forestgreen4416044}{RGB}{44,160,44}
\definecolor{lightgray204}{RGB}{204,204,204}

\begin{axis}[
legend cell align={left},
legend style={
  fill opacity=0.8,
  draw opacity=1,
  text opacity=1,
  at={(0.03,0.97)},
  anchor=north west,
  draw=lightgray204
},
tick align=outside,
tick pos=left,
x grid style={darkgray176},
xlabel={\(\displaystyle x\)},
xmin=0, xmax=1,
xtick style={color=black},
y grid style={darkgray176},
ylabel={\(\displaystyle y\)},
ymin=0, ymax=1,
ytick style={color=black}
]
\addplot [semithick, crimson2143940, mark=x, mark size=3, mark options={solid}]
table {%
0 0
0.75 0.25
0.769111493137114 0.295030808940326
0.758740009391953 0.269051373602029
0.753969399023836 0.255957566324907
0.751900750539202 0.251947243683122
0.750946914353436 0.250685628723688
0.750485577051012 0.250250953964589
0.75025461130967 0.250093861311242
0.750135913769046 0.250035615155214
0.750073635619525 0.250013660590296
0.750040400289315 0.250005285507392
0.75002240901281 0.250002060172609
0.750012549710712 0.250000808179501
0.750007088635035 0.250000318850066
0.750004034948061 0.250000126441755
0.750002312867694 0.250000050374971
0.750001334259586 0.250000020155187
0.75000077425417 0.250000008095786
0.750000451739074 0.250000003263651
0.750000264901497 0.250000001320107
0.750000156072297 0.250000000535645
0.750000092359328 0.250000000217982
0.750000054882302 0.250000000088953
0.750000032739774 0.250000000036394
0.750000019602722 0.250000000014927
0.750000011777928 0.250000000006136
0.750000007099953 0.250000000002528
0.750000004293438 0.250000000001044
0.750000002604073 0.250000000000432
0.750000001583943 0.250000000000179
0.750000000966072 0.250000000000074
0.750000000590761 0.250000000000031
0.750000000362159 0.250000000000013
0.750000000222551 0.250000000000005
0.750000000137076 0.250000000000002
0.750000000084616 0.250000000000001
0.750000000052345 0.25
0.750000000032448 0.25
0.750000000020154 0.25
0.750000000012543 0.25
0.75000000000782 0.25
0.750000000004884 0.25
0.750000000003056 0.25
0.750000000001915 0.25
0.750000000001202 0.25
0.750000000000756 0.25
0.750000000000476 0.25
0.7500000000003 0.25
0.75000000000019 0.25
0.75000000000012 0.25
0.750000000000076 0.25
0.750000000000048 0.25
0.750000000000031 0.25
0.750000000000019 0.25
0.750000000000012 0.25
0.750000000000008 0.25
0.750000000000005 0.25
0.750000000000003 0.25
0.750000000000002 0.25
0.750000000000001 0.25
0.750000000000001 0.25
0.750000000000001 0.25
0.75 0.25
0.75 0.25
0.75 0.25
0.75 0.25
0.75 0.25
0.75 0.25
0.75 0.25
0.75 0.25
0.75 0.25
0.75 0.25
0.75 0.25
0.75 0.25
0.75 0.25
0.75 0.25
0.75 0.25
0.75 0.25
0.75 0.25
0.75 0.25
0.75 0.25
0.75 0.25
0.75 0.25
0.75 0.25
0.75 0.25
0.75 0.25
0.75 0.25
0.75 0.25
0.75 0.25
0.75 0.25
0.75 0.25
0.75 0.25
0.75 0.25
0.75 0.25
0.75 0.25
0.75 0.25
0.75 0.25
0.75 0.25
0.75 0.25
0.75 0.25
};
\addlegendentry{Damp. Kleene ($x_1$)}
\addplot [semithick, darkturquoise23190207, mark=x, mark size=3, mark options={solid}]
table {%
0.875 0.875
0.806977845339618 0.529837087637276
0.771330210504779 0.345435970651573
0.758796840956141 0.274136626909536
0.753970291045016 0.256184898662225
0.751900758972133 0.251951570243785
0.750946914401226 0.250685664046144
0.750485577051174 0.25025095408835
0.750254611309671 0.250093861311427
0.750135913769046 0.250035615155214
0.750073635619525 0.250013660590296
0.750040400289315 0.250005285507392
0.75002240901281 0.250002060172609
0.750012549710712 0.250000808179501
0.750007088635035 0.250000318850066
0.750004034948061 0.250000126441755
0.750002312867694 0.250000050374971
0.750001334259586 0.250000020155187
0.75000077425417 0.250000008095786
0.750000451739074 0.250000003263651
0.750000264901497 0.250000001320107
0.750000156072297 0.250000000535645
0.750000092359328 0.250000000217982
0.750000054882302 0.250000000088953
0.750000032739774 0.250000000036394
0.750000019602722 0.250000000014927
0.750000011777928 0.250000000006136
0.750000007099953 0.250000000002528
0.750000004293438 0.250000000001044
0.750000002604073 0.250000000000432
0.750000001583943 0.250000000000179
0.750000000966072 0.250000000000074
0.750000000590761 0.250000000000031
0.750000000362159 0.250000000000013
0.750000000222551 0.250000000000005
0.750000000137076 0.250000000000002
0.750000000084616 0.250000000000001
0.750000000052345 0.25
0.750000000032448 0.25
0.750000000020154 0.25
0.750000000012543 0.25
0.75000000000782 0.25
0.750000000004884 0.25
0.750000000003056 0.25
0.750000000001915 0.25
0.750000000001202 0.25
0.750000000000756 0.25
0.750000000000476 0.25
0.7500000000003 0.25
0.75000000000019 0.25
0.75000000000012 0.25
0.750000000000076 0.25
0.750000000000048 0.25
0.750000000000031 0.25
0.750000000000019 0.25
0.750000000000012 0.25
0.750000000000008 0.25
0.750000000000005 0.25
0.750000000000003 0.25
0.750000000000002 0.25
0.750000000000001 0.25
0.750000000000001 0.25
0.750000000000001 0.25
0.75 0.25
0.75 0.25
0.75 0.25
0.75 0.25
0.75 0.25
0.75 0.25
0.75 0.25
0.75 0.25
0.75 0.25
0.75 0.25
0.75 0.25
0.75 0.25
0.75 0.25
0.75 0.25
0.75 0.25
0.75 0.25
0.75 0.25
0.75 0.25
0.75 0.25
0.75 0.25
0.75 0.25
0.75 0.25
0.75 0.25
0.75 0.25
0.75 0.25
0.75 0.25
0.75 0.25
0.75 0.25
0.75 0.25
0.75 0.25
0.75 0.25
0.75 0.25
0.75 0.25
0.75 0.25
0.75 0.25
0.75 0.25
0.75 0.25
0.75 0.25
};
\addlegendentry{Damp. Kleene ($x_2$)}
\addplot [semithick, darkorange25512714, mark=x, mark size=3, mark options={solid}]
table {%
0 0
0.375 0.125
0.534525161425502 0.191683782744626
0.60453054891518 0.209018527291101
0.641657487484532 0.216724340558741
0.664477553583759 0.222445473749409
0.679888892903718 0.226931415337955
0.690955192900905 0.230404630201095
0.699252508035723 0.233102223701986
0.705679447492157 0.23522527564004
0.710786927867125 0.236923966148938
0.714931409392261 0.238305961272659
0.718353409686661 0.239447920093794
0.721220793510869 0.240404836517531
0.723654045591727 0.24121671069248
0.725741718499509 0.241913112858137
0.727550238740263 0.242516269992986
0.72913033221339 0.243043160057682
0.73052135024727 0.243506947872784
0.731754251534199 0.243917983491177
0.732853700565422 0.244284507170953
0.733839571764697 0.244613155030061
0.734728045591232 0.2449093272974
0.735532419322048 0.245177460395331
0.736263714993648 0.245421230699511
0.736931140954886 0.245643709056864
0.737542446308484 0.245847479328598
0.738104195982513 0.246034730314423
0.738621986297299 0.246207327745902
0.739100615435686 0.246366871191019
0.739544219390406 0.24651473941784
0.739956381234213 0.246652126846183
0.740340219594089 0.246780073055792
0.740698460780845 0.246899486839497
0.741033497973477 0.24701116593727
0.741347440076336 0.247115813325499
0.741642152281502 0.247214050739905
0.741919289925936 0.24730642996254
0.742180326895512 0.247393442290568
0.742426579568765 0.247475526517982
0.742659227092633 0.247553075694463
0.742879328626217 0.247626442873481
0.74308783806607 0.247695946020819
0.74328561666983 0.247761872222521
0.743473443918308 0.247824481305628
0.743652026894867 0.247884008964657
0.743822008411819 0.247940669470418
0.743983974073883 0.247994658024509
0.744138458436664 0.248046152812146
0.74428595039191 0.248095316797256
0.744426897889928 0.248142299296613
0.744561712091955 0.248187237363966
0.744690771030773 0.248230257010245
0.744814422845855 0.248271474281944
0.74493298864936 0.248310996216449
0.745046765070937 0.248348921690309
0.745156026522365 0.24838534217412
0.745261027217175 0.248420342405724
0.745362002975487 0.248454000991828
0.745459172840091 0.248486390946696
0.745552740526313 0.248517580175437
0.745642895725165 0.248547631908388
0.74572981527671 0.248576605092236
0.745813664228424 0.248604554742808
0.745894596791408 0.248631532263803
0.74597275720571 0.248657585735237
0.746048280524643 0.248682760174881
0.746121293326746 0.248707097775582
0.74619191436303 0.24873063812101
0.746260255146228 0.248753418382076
0.746326420487992 0.248775473495997
0.746390508989299 0.248796836329766
0.746452613488719 0.248817537829573
0.746512821472704 0.248837607157568
0.746571215451562 0.248857071817187
0.746627873304412 0.248875957768137
0.746682868596046 0.248894289532015
0.74673627086832 0.24891209028944
0.746788145908411 0.24892938196947
0.746838555996058 0.248946185332019
0.746887560131655 0.248962520043885
0.74693521424691 0.24897840474897
0.746981571399592 0.248993857133197
0.747026681953732 0.249008893984577
0.747070593746551 0.24902353124885
0.747113352243201 0.249037784081067
0.74715500068037 0.249051666893457
0.747195580199645 0.249065193399882
0.747235129971494 0.249078376657164
0.747273687310599 0.249091229103533
0.747311287783258 0.249103762594419
0.747347965307458 0.249115988435819
0.747383752246209 0.249127917415403
0.747418679494657 0.249139559831552
0.747452776561438 0.249150925520479
0.747486071644735 0.249162023881579
0.747518591703412 0.249172863901137
0.747550362523597 0.249183454174532
0.747581408781059 0.24919380292702
0.747611754099665 0.249203918033222
0.747641421106218 0.249213807035406
};
\addlegendentry{Damp. Mann ($x_1$)}
\addplot [semithick, forestgreen4416044, mark=x, mark size=3, mark options={solid}]
table {%
0.875 0.875
0.473096740996309 0.334526362145138
0.544256821301222 0.233617745330045
0.605362903232776 0.21467431058337
0.641719276112866 0.217256820039673
0.664481584212609 0.22248410385644
0.679889127413472 0.226933751920473
0.690955205241144 0.230404754191597
0.699252508630004 0.233102229653762
0.705679447518603 0.235225275903208
0.710786927868221 0.236923966159776
0.714931409392303 0.238305961273078
0.718353409686663 0.239447920093809
0.721220793510869 0.240404836517532
0.723654045591727 0.24121671069248
0.725741718499509 0.241913112858137
0.727550238740263 0.242516269992986
0.72913033221339 0.243043160057682
0.73052135024727 0.243506947872784
0.731754251534199 0.243917983491177
0.732853700565422 0.244284507170953
0.733839571764697 0.244613155030061
0.734728045591232 0.2449093272974
0.735532419322048 0.245177460395331
0.736263714993648 0.245421230699511
0.736931140954886 0.245643709056864
0.737542446308484 0.245847479328598
0.738104195982513 0.246034730314423
0.738621986297299 0.246207327745902
0.739100615435686 0.246366871191019
0.739544219390406 0.24651473941784
0.739956381234213 0.246652126846183
0.740340219594089 0.246780073055792
0.740698460780845 0.246899486839497
0.741033497973477 0.24701116593727
0.741347440076336 0.247115813325499
0.741642152281502 0.247214050739905
0.741919289925936 0.24730642996254
0.742180326895512 0.247393442290568
0.742426579568765 0.247475526517982
0.742659227092633 0.247553075694463
0.742879328626217 0.247626442873481
0.74308783806607 0.247695946020819
0.74328561666983 0.247761872222521
0.743473443918308 0.247824481305628
0.743652026894867 0.247884008964657
0.743822008411819 0.247940669470418
0.743983974073883 0.247994658024509
0.744138458436664 0.248046152812146
0.74428595039191 0.248095316797256
0.744426897889928 0.248142299296613
0.744561712091955 0.248187237363966
0.744690771030773 0.248230257010245
0.744814422845855 0.248271474281944
0.74493298864936 0.248310996216449
0.745046765070937 0.248348921690309
0.745156026522365 0.24838534217412
0.745261027217175 0.248420342405724
0.745362002975487 0.248454000991828
0.745459172840091 0.248486390946696
0.745552740526313 0.248517580175437
0.745642895725165 0.248547631908388
0.74572981527671 0.248576605092236
0.745813664228424 0.248604554742808
0.745894596791408 0.248631532263803
0.74597275720571 0.248657585735237
0.746048280524643 0.248682760174881
0.746121293326746 0.248707097775582
0.74619191436303 0.24873063812101
0.746260255146228 0.248753418382076
0.746326420487992 0.248775473495997
0.746390508989299 0.248796836329766
0.746452613488719 0.248817537829573
0.746512821472704 0.248837607157568
0.746571215451562 0.248857071817187
0.746627873304412 0.248875957768137
0.746682868596046 0.248894289532015
0.74673627086832 0.24891209028944
0.746788145908411 0.24892938196947
0.746838555996058 0.248946185332019
0.746887560131655 0.248962520043885
0.74693521424691 0.24897840474897
0.746981571399592 0.248993857133197
0.747026681953732 0.249008893984577
0.747070593746551 0.24902353124885
0.747113352243201 0.249037784081067
0.74715500068037 0.249051666893457
0.747195580199645 0.249065193399882
0.747235129971494 0.249078376657164
0.747273687310599 0.249091229103533
0.747311287783258 0.249103762594419
0.747347965307458 0.249115988435819
0.747383752246209 0.249127917415403
0.747418679494657 0.249139559831552
0.747452776561438 0.249150925520479
0.747486071644735 0.249162023881579
0.747518591703412 0.249172863901137
0.747550362523597 0.249183454174532
0.747581408781059 0.24919380292702
0.747611754099665 0.249203918033222
0.747641421106218 0.249213807035406
};
\addlegendentry{Damp. Mann ($x_2$)}
\addplot [thick, black, mark=x, mark size=3, mark options={solid}, only marks, forget plot]
table {%
0.75 0.25
};
\addplot [thick, violet28210083, mark=x, mark size=3, mark options={solid}, only marks, forget plot]
table {%
0.875 0.875
};
\end{axis}

\end{tikzpicture}
\clearpage{}%

    \caption{Generated sequences $(x_n^{\mathcal{F}})$ on the function
    sequence~\eqref{eq:dim2-function-sequence} with dampening starting from $x_1=(0,0)$ and
    $x_2=(7/8,7/8)$ (marked as a violet cross), respectively.}
    \label{fig:dim2-values-damp}
  \end{subfigure}
  \caption{Sequences $(x_n^{\mathcal{F}})$ generated by different iteration iterations on the
  function sequence~\eqref{eq:dim2-function-sequence}. The fixpoint $\mu f$ of the limit function
  is marked as a black cross.}
  \label{fig:dim2-values}
\end{figure}
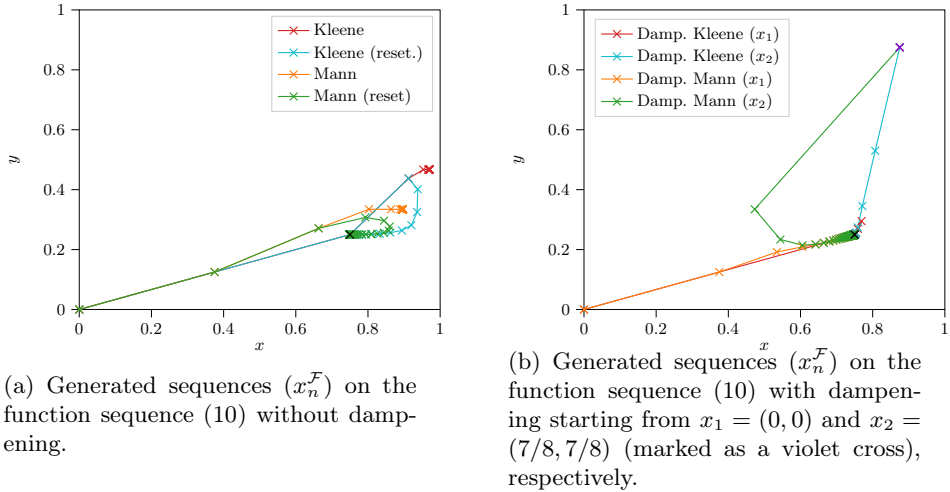

Lastly, we compare the performance of dampened iterations to resetting iterations on randomly generated
MDPs.
\Cref{fig:mdp-approx} shows the results of performing a dampened Mann iteration and an
undampened Kleene iteration with resets on sampled (state-value)-Bellman operators of one hundred
randomly generated normalised MDPs with 50 states (in equal amounts simple Markov chains, MCs with
five end-component, as well as simple MDPs and MDPs with five MECs) for which Kleene iteration on
the exact Bellman operator took less than 1000 steps to stop (we stop
when the update $<10^{-6}$).
The lines denoting the mean error of the value approximations suggest that Kleene iteration with
resets performs better overall, but not significantly better.
The areas showing the 90th percentile and maximum, on the other hand, illustrate that Kleene
iteration with resets is more dependent on the quality of the currently sampled MDP (with error
spikes at points where the MDP were badly sampled) while the dampened Mann iteration reduces the
variance introduced by random sampling.

The main advantage of dampened iterations, when compared to resetting iterations, is their ability
to reuse knowledge gained from old approximations.
In particular, one step $n\mapsto n+1$ only takes one application of $T_n^{\mathcal{S}}(f_n,\cdot)$
while it takes $n+1$ applications of $T_k^{\mathcal{S}}(f_n,\cdot)$ in
a Kleene iteration with resets.
In practice, it might be unnecessary to compute every $x_n^{\mathcal{F}}$ in
a resetting iteration.
In an offline scenario, it might be sufficient to just compute $x_n^{\mathcal{F}}$ for one $n$
(i.e., performing Kleene iteration on one approximation $f_n$), but reinforcement learning is most
often treated in an online setting where live updates to the values are required as new samples are
generated.
Furthermore, even in an offline scenario, the above results showcase that the results of resetting
schemes vary highly with the quality of the current approximation, whence relying on more than a
single approximation $f_n$ might produce better results.
In the context of resetting schemes, this, of course, comes at a high cost of required runtime as is
clearly visible in \Cref{fig:mdp-approx-runtime}, comparing the runtime of a dampened Mann
iteration over 1000 approximation steps with Kleene iteration on $f_{1000}$ (which is slightly but
not significantly faster; note that the corresponding red line is
almost covered by the green line) and with Kleene iteration resetting every 100 or every 50 steps (resulting
in 10 and 20 performed iterations), respectively, on one hundred randomly generated MDPs per number
of states (again half of which were simple and again half of simple and non-simple MDPs were
Markov chains).
The variance in runtime for each iteration scheme was mainly a result of the number of actions
in the MDP (with Markov chains resulting in the lowest runtime). 

\begin{figure}
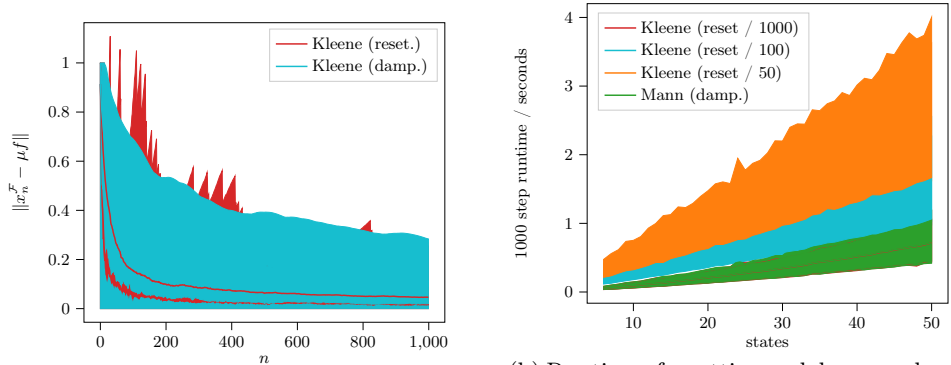

  \centering
  \begin{subfigure}[b]{.45\textwidth}
    \centering
\clearpage{}%

\clearpage{}%

    \caption{Errors of resetting Kleene and dampened Mann scheme on 100 randomly generated MDPs of value $1$.
    The lines denote the mean error, the half-opaque area the 10th to 90th percentile, and the weakly
    visible area the minimum to maximum error.}
    \label{fig:mdp-approx}
  \end{subfigure}
  \hfill
  \begin{subfigure}[b]{.45\textwidth}
    \centering
\clearpage{}%
%
\clearpage{}%

    \caption{Runtime of resetting and dampened Mann scheme on 100 randomly generated MDPs per number
    of states (6 -- 50).
    The lines denote the mean runtime, the area the minimum to maximum runtime for the given scheme.
    Note that the red line is hidden beneath the green one due to overlap.}
    \label{fig:mdp-approx-runtime}
  \end{subfigure}
  \caption{Results on sampled Bellman operators of randomly generated MDPs.}
\end{figure}

\section{Conclusion}
\label{se:conclusion}

\paragraph{Related Work.}
There has been a significant amount of research on convergence properties of Mann iterations
of non-expansive maps. However, most of the existing literature deals
with the undampened case and only a single function $f$, i.e., with a
scheme defined by $x_{n+1} = \alpha_n x_n + (1-\alpha_n)f(x_n)$.

It is known
\cite{borwein_reich_shafrir_1992} that in every Banach space and for
every non-expansive $f$ with at least one fixpoint this scheme
produces a sequence which is $f$-as\-ymp\-tot\-i\-cal\-ly regular (i.e.\ $\|x_n - f(x_n)\| \to 0$)
for all parameter sequences $(\alpha_n)$ with
$\sum (1-\alpha_n) = \infty$. In general the sequence does not need to
converge. (Weak) convergence is shown under additional assumptions on
the map $f$ and/or the surrounding space in various places (see for
example \cite{REICH1979274,b:iterative-approximation-fixed-points,alma99371219572606441}).

Dampened versions similar to ours have been considered in
\cite{KIM200551} where the authors show convergence to a fixpoint
under strong additional assumptions on the parameter sequences in a uniformly
smooth Banach space and in \cite{Yao2008StrongCO} where convergence is
shown for more general parameter sequences in Hilbert spaces. 

Here, in addition,
we deal with the case where
the function $f$ can only be approximated. Motivated by the
applicability to MDPs, we consider the case of monotone non-expansive
functions in finite-dimensional Banach spaces
 equipped
with the
supremum norm, which are neither uniformly smooth nor Hilbert spaces.
They are not even uniformly convex as is assumed e.g. in \cite{REICH1979274}.
We also
introduced a probabilistic setting and devised a generic algorithm,
which applies in many situations and guarantees almost sure
convergence to the least fixpoint of $f$.

As mentioned earlier, our work is inspired by the integration of reinforcement learning algorithms with fixpoint theory
\cite{WD:QL,ma:prioritized-sweeping,r:sarsa,kv:certainty-equivalent},
and in particular by the Dyna-Q algorithm
\cite{s:dyna-integrated-architecture,s:planning-incremental-dyn-prog,kaelbling1996reinforcement}.
Reinforcement learning techniques are by now a central ingredient in
machine learning and we plan to strengthen the foundations behind this
method.

There is some similarity of our results to the theory of stochastic
approximation
\cite{bt:neuro-dynamic-programming,b:random-iterative-models,kc:stochastic-approximation-methods},
going back to \cite{rm:stochastic-approximation}, which has also been
employed in the correctness proof of Q-learning \cite{WD:QL}.  Stochastic
approximation deals with a root finding problem where the function
contains an error term with expected value zero. In that line of work,
the function itself is stochastic with known expected value, while we
separate the (approximated) fixpoint theory from the probabilistic
setting by checking that we converge almost surely via sampling. More
importantly, the focus in stochastic approximation is to consider
either contractions (with unique fixpoints) or convergence to
\emph{some} fixpoint, while the case of convergence to the
\emph{least} fixpoint is not treated. Hence there is no machinery for
leaving a fixpoint once it is reached, such as we have by employing a
dampening factor.

In order to obtain our results on MDPs, the notion of end-components
is quite fundamental. For an introduction to end-components see
\cite{bk:principles-mc}, while the connection of non-uniqueness of
fixpoints of MDP functions to end-components was exploited in
\cite{bccfkkpu:mdps-learning,hm:interval-iteration-mdps} in order to
compute the correct expected return by collapsing maximal end
components. In \cite{bccfkkpu:mdps-learning} the authors assume -- as
we do -- that MDPs are not known exactly but are sampled. They present
a different solution by collapsing maximal end-components, which is
unnecessary in our approach.

\paragraph{Numerical Experiments.}
We made several numerical experiments (see  \Cref{se:experiments}).
For unknown functions the outcome can be
summarized as follows: in terms of precision, the results are
comparable to ``resetting'', where -- for some $n$ -- we perform the
iteration from $0$ with the current approximation $f_n$. Compared to
our suggested form of iteration, resetting has the disadvantage that
the results vary highly with the quality of the current approximation
and that intermediate results are not available.

\paragraph{Future Work.}
For future work we are in particular interested in the case of
parameter sequences where $\alpha_n$ converges to one and the question
is whether we can still guarantee convergence to $\mu f$. This would
be interesting, since such parameters are for instance used in the
context of Q-learning \cite{WD:QL}. Our long-term aim is to devise a
fixpoint theory of approximated functions for which a large number of
reinforcement learning algorithms are special cases. Here we
concentrated on model-based reinforcement learning, for model-free
reinforcement learning we need a theory that allows updates with a
sampled value rather than with the current best approximation and we
plan to work with the weaker assumption that the limit-average of the
functions $f_n$ converges to $f$ (rather than the functions
themselves). This has some overlap to stochastic approximation
  theory as discussed above, even though we do not consider stochastic
  functions. Hence we plan to partially build on the results obtained
  in that area.

To properly match reinforcement learning algorithms we
will look into chaotic and asynchronous iteration~\cite{c:asynchronous-fixed-point,FROMMER2000201} where one can iterate
at different speeds at various states. This is useful when traversing
an MDP or stochastic system, making updates on-the-fly and
locally when better estimates are obtained.

We will also investigate how to iterate to the greatest fixpoint
rather than the least, either by dualizing or adapting ideas from
\cite{KIM200551,hz:modified-mann-iteration} on iterating to the
fixpoint closest to a given value. Including rewards from different
domains, such as negative rewards, is also a direction of future
work.

We plan to identify more cases where approximated fixpoint iteration
works out of the box. For instance, we mentioned that simple
stochastic games are covered by the results in
\Cref{se:mann-error-algo}, but there we needed a way to enforce fast
convergence.
We will investigate whether the case of simple stochastic games (and
related applications) works with the standard iteration scheme. Due to
\Cref{ex:cex-flip-map} it is clear that this
will not be true for all approximated functions, but we plan to find
requirements weaker than the current ones.

Moreover we will further look into error estimates, making guarantees
of the form ``with probability $p$ the error is below $\epsilon$''. In
the case of MDPs, even when they are known exactly, determining
upper/lower bounds and thus error estimates already requires some
non-trivial techniques, such as the results in
\cite{hm:interval-iteration-mdps,qk:sound-value-iteration,hk:optimistic-value-iteration}. It
is our aim to generalize these results to the case where the
probabilities associated to the MDP are not known exactly and, in this
case, we expect to obtain probabilistic error estimates.

\smallskip

\subsubsection*{Acknowledgements.}
This work is partially supported by the European Union –
NextGenerationEU under the National Recovery and Resilience Plan
(NRRP) - Call PRIN 2022 PNRR - Project P2022HXNSC "Resource Awareness
in Programming: Algebra, Rewriting, and Analysis" and by the Deutsche
Forschungsgemeinschaft (DFG, German Research Foundation) -- project
number 434050016 (SpeQt).

\subsubsection*{Disclosure of interests.}
The authors have no competing interests to declare that are relevant
to the content of this article.

\bibliographystyle{plain}
\bibliography{references}
\end{document}